%% file: main.tex
\documentclass[11pt]{article}

\usepackage{amsthm}
\usepackage{graphicx} 
\usepackage{array} 

\usepackage{dsfont}
\usepackage{amsmath, amssymb, amsfonts, verbatim}
\usepackage{hyphenat,epsfig,subcaption,multirow}
\usepackage{nicefrac}
\usepackage{paralist}

\usepackage[font=footnotesize,labelfont=bf]{caption}

\usepackage[usenames,dvipsnames]{xcolor}
\usepackage[ruled]{algorithm2e}

\DeclareFontFamily{U}{mathx}{\hyphenchar\font45}
\DeclareFontShape{U}{mathx}{m}{n}{
      <5> <6> <7> <8> <9> <10>
      <10.95> <12> <14.4> <17.28> <20.74> <24.88>
      mathx10
      }{}
\DeclareSymbolFont{mathx}{U}{mathx}{m}{n}
\DeclareMathSymbol{\bigtimes}{1}{mathx}{"91}

\usepackage{tcolorbox}
\tcbuselibrary{skins,breakable}
\tcbset{enhanced jigsaw}

\usepackage[normalem]{ulem}
\usepackage[compact]{titlesec}

\definecolor{DarkRed}{rgb}{0.5,0.1,0.1}
\definecolor{DarkBlue}{rgb}{0.1,0.1,0.5}

\usepackage{nameref}
\definecolor{ForestGreen}{rgb}{0.1333,0.5451,0.1333}

\definecolor{Red}{rgb}{0.9,0,0}
\usepackage[linktocpage=true,
	pagebackref=true,colorlinks,
	linkcolor=DarkRed,citecolor=ForestGreen,
	bookmarks,bookmarksopen,bookmarksnumbered]
	{hyperref}
\usepackage[noabbrev,nameinlink]{cleveref}
\crefname{property}{property}{Property}
\creflabelformat{property}{(#1)#2#3}
\crefname{equation}{eq}{Eq}
\creflabelformat{equation}{(#1)#2#3}

\usepackage{bm}
\usepackage{url}
\usepackage{xspace}
\usepackage[mathscr]{euscript}

\usepackage{tikz}
\usetikzlibrary{arrows}
\usetikzlibrary{arrows.meta}
\usetikzlibrary{shapes}
\usetikzlibrary{backgrounds}
\usetikzlibrary{positioning}
\usetikzlibrary{decorations.markings}
\usetikzlibrary{patterns}
\usetikzlibrary{calc}
\usetikzlibrary{fit}
\usetikzlibrary{snakes}
\tikzset{vertex/.style={circle, black, fill=Yellow, line width=1pt, draw, minimum width=8pt, minimum height=8pt, inner sep=0pt}}
	
\usepackage{mdframed}

\usepackage[noend]{algpseudocode}
\makeatletter
\def\BState{\State\hskip-\ALG@thistlm}
\makeatother

\usepackage{cite}
\usepackage{enumitem}

\usepackage[margin=1in]{geometry}

\newtheorem{theorem}{Theorem}
\newtheorem{lemma}{Lemma}[section]
\newtheorem{proposition}[lemma]{Proposition}
\newtheorem{corollary}[lemma]{Corollary}
\newtheorem{claim}[lemma]{Claim}
\newtheorem{fact}[lemma]{Fact}

\newtheorem{definition}[lemma]{Definition}

\newtheorem*{claim*}{Claim}
\newtheorem*{theorem*}{Theorem}
\newtheorem*{proposition*}{Proposition}
\newtheorem*{lemma*}{Lemma}
\newtheorem*{problem*}{Problem}

\crefname{lemma}{Lemma}{Lemmas}
\crefname{claim}{Claim}{Claims}

\newtheorem{mdresult}{Result}
\newenvironment{result}{\begin{mdframed}[backgroundcolor=lightgray!40,topline=false,rightline=false,leftline=false,bottomline=false,innertopmargin=2pt, innerleftmargin=10pt]\begin{mdresult}}{\end{mdresult}\end{mdframed}}

\newtheorem{remark}[lemma]{Remark}
\newtheorem*{remark*}{Remark}

\newtheoremstyle{restate}{}{}{\itshape}{}{\bfseries}{~(restated).}{.5em}{\thmnote{#3}}
\theoremstyle{restate}

\theoremstyle{definition}
\newtheorem{mdalg}{algorithm}

\newtheorem{mddist}{Distribution}
\newenvironment{Distribution}{\begin{tbox}\begin{mddist}}{\end{mddist}\end{tbox}}

\allowdisplaybreaks

\DeclareMathOperator*{\argmax}{arg\,max}

\renewcommand{\qed}{\nobreak \ifvmode \relax \else
      \ifdim\lastskip<1.5em \hskip-\lastskip
      \hskip1.5em plus0em minus0.5em \fi \nobreak
      \vrule height0.75em width0.5em depth0.25em\fi}

\newcommand{\Qed}[1]{\rlap{\qed$_{\textnormal{~~~\Cref{#1}}}$}}

\input{macros}

\title{Graph Streaming Lower Bounds for Parameter Estimation and Property Testing  via a  Streaming XOR Lemma}
\author{Sepehr Assadi\footnote{(sepehr@assadi.info) Cheriton School of Computer Science, University of Waterloo. This work was done when this author was at the Department of Computer Science, Rutgers University.} \and 
Vishvajeet N\footnote{(nvishvajeet@gmail.com)  Quantum Information and Computation group, CNRS laboratory LaBRI, Université de Bordeaux. This work was done when this author was at the Department of Computer Science, Rutgers University.}
}

\date{}

\begin{document}
\maketitle

\pagenumbering{roman}

\input{abstract}

\clearpage

\setcounter{tocdepth}{3}
\tableofcontents

\clearpage

\pagenumbering{arabic}
\setcounter{page}{1}

\input{1intro}
\input{2prelim}

\input{3.5direct-product.tex}
\input{4cycle}

\input{5reductions}

\input{6pointer}

\subsection*{Acknowledgement} 
SA is grateful to Gillat Kol, Raghuvansh Saxena, and Huacheng Yu, for their previous collaboration in~\cite{AssadiKSY20} that was the starting point of this project, and 
to David Wajc and Huacheng Yu for illuminating discussions regarding the Streaming XOR Lemma that also prompted us to include~\Cref{sec:xor-examples}. SA is also very grateful to Chin Ho Lee for pointing out the missing step
in the original proof of our Streaming XOR Lemma that is now fixed in this version (see also~\cite{LeePV23}).

\bibliographystyle{alpha}
\bibliography{general}

\clearpage
\appendix
\part*{Appendix}
\input{appendixA-info}

\input{appendixB-xor-lemma-optimality}

\input{appendixC-alg-ngc}
\input{appendixD-comparison}

\end{document}

%% file: macros.tex
\newcommand{\tvd}[2]{\ensuremath{\norm{#1 - #2}_{\mathrm{tvd}}}}
\newcommand{\Ot}{\ensuremath{\widetilde{O}}}
\newcommand{\eps}{\ensuremath{\varepsilon}}

\newcommand{\Bracket}[1]{\Big[#1\Big]}
\newcommand{\bracket}[1]{\left[#1\right]}
\newcommand{\paren}[1]{\ensuremath{\left(#1\right)}\xspace}
\newcommand{\card}[1]{\left\vert{#1}\right\vert}

\newcommand{\IN}{\ensuremath{\mathbb{N}}}

\newcommand{\norm}[1]{\ensuremath{\|#1\|}}

\newcommand{\conc}{\ensuremath{\,||\,}}

\newcommand{\ceil}[1]{{\left\lceil{#1}\right\rceil}}

\newcommand{\set}[1]{\ensuremath{\left\{ #1 \right\}}}
\newcommand{\poly}{\mbox{\rm poly}}
\newcommand{\polylog}{\mbox{\rm  polylog}}

\newcommand{\alg}{\ensuremath{\mathcal{A}}\xspace}

\DeclareMathOperator*{\Exp}{\ensuremath{{\mathbb{E}}}}
\DeclareMathOperator*{\Prob}{\ensuremath{\textnormal{Pr}}}
\renewcommand{\Pr}{\Prob}

\newcommand{\Ex}{\Exp}

\newenvironment{tbox}{\begin{tcolorbox}[
		enlarge top by=5pt,
		enlarge bottom by=5pt,
		 breakable,
		 boxsep=2pt,
                  left=5pt,
                  right=7pt,
                  top=10pt,
                  arc=0pt,
                  boxrule=1pt,toprule=1pt,
                  colback=white
                  ]
	}
{\end{tcolorbox}}

\newcommand{\event}{\ensuremath{\mathcal{E}}}

\newcommand{\rv}[1]{\ensuremath{{\mathsf{#1}}}\xspace}
\newcommand{\rA}{\rv{A}}
\newcommand{\rB}{\rv{B}}
\newcommand{\rC}{\rv{C}}
\newcommand{\rD}{\rv{D}}

\newcommand{\rY}{\rv{Y}}
\newcommand{\supp}[1]{\ensuremath{\textnormal{\text{supp}}(#1)}}
\newcommand{\distribution}[1]{\ensuremath{\textnormal{dist}(#1)}\xspace}

\newcommand{\kl}[2]{\ensuremath{\mathbb{D}(#1~||~#2)}}
\newcommand{\II}{\ensuremath{\mathbb{I}}}
\newcommand{\HH}{\ensuremath{\mathbb{H}}}
\newcommand{\mi}[2]{\ensuremath{\def\mione{#1}\def\mitwo{#2}\mireal}}
\newcommand{\mireal}[1][]{
  \ifx\relax#1\relax%
    \II(\mione \,; \mitwo)%
  \else%
    \II(\mione \,; \mitwo\mid #1)%
  \fi
}
\newcommand{\en}[1]{\ensuremath{\HH(#1)}}

\newcommand{\itfacts}[1]{\Cref{fact:it-facts}-(\ref{part:#1})\xspace}

\newcommand{\dist}{\ensuremath{\mathcal{D}}}

\newcommand{\PP}[1]{{P}^{#1}}

\newcommand{\prot}{\pi}
\newcommand{\Prot}{\Pi}

\newcommand{\distNGC}{\ensuremath{\mu_{\textnormal{\textsf{NGC}}}}\xspace}
\newcommand{\PC}{\textnormal{\textbf{PC}}\xspace}

\newcommand{\distPC}{\ensuremath{\mathcal{\mu}_{\textsf{PC}}}}

\newcommand{\fxor}{\ensuremath{f^{\oplus \ell}}}

\newcommand{\bias}{\ensuremath{\mathrm{bias}}}

\newcommand{\rX}{\rv{X}}
\newcommand{\rZ}{\rv{Z}}
\newcommand{\rM}{\rv{M}}

\newcommand{\XOR}{\ensuremath{\textnormal{XOR}}\xspace}

%% file: abstract.tex
\begin{abstract}
	
	We study \textbf{space-pass tradeoffs} in graph streaming algorithms for parameter estimation and property testing problems such as 
	estimating the size of maximum matchings and maximum cuts,  weight of minimum spanning trees, or testing if a graph is connected or cycle-free versus being far from these properties.  
	We develop a new lower bound technique that proves that for many problems of interest, including all the above, obtaining a $(1+\eps)$-approximation requires either $n^{\Omega(1)}$ space or $\Omega(1/\eps)$ passes, 
	even on highly restricted families of graphs such as bounded-degree planar graphs. For multiple of these problems, this bound matches those of existing algorithms and is thus (asymptotically) optimal. 
	
	\medskip
	
	Our results considerably strengthen prior lower bounds even for arbitrary graphs: starting from the influential work of [Verbin, Yu; SODA 2011], there has been a plethora of lower bounds
	for single-pass algorithms for these problems; however, the only multi-pass lower bounds proven very recently in [Assadi, Kol, Saxena, Yu; FOCS 2020]  rules out  sublinear-space algorithms
	 with exponentially smaller $o(\log{(1/\eps)})$ passes for these problems. 
	 
	 \medskip
	 
	 One key ingredient of our proofs is a simple \textbf{streaming XOR Lemma}, a generic hardness amplification result, that we prove: informally speaking, if a 
	 $p$-pass $s$-space streaming algorithm can only solve a decision problem with advantage $\delta > 0$ over random guessing, then it cannot solve XOR of $\ell$ independent copies of the problem with advantage better than $\delta^{\Omega(\ell)}$. 
	This result can be of independent interest and useful for other streaming lower bounds as well.

\end{abstract}

%% file: 1intro.tex

\newcommand{\NGC}{\textnormal{\textbf{NGC}}\xspace}

\section{Introduction}\label{sec:intro}

Consider an $n$-vertex undirected graph $G=(V,E)$ whose edges are arriving one by one in a stream. Suppose we want to process $G$ with a streaming algorithm using small space (e.g., $\polylog{(n)}$ bits), and in a few passes (e.g., a small constant). 
 How well can we \emph{estimate}  parameters of $G$ such as size of maximum cuts and maximum matchings, weight of minimum spanning trees, or number of short cycles? 
How well can we perform \emph{property testing} on $G$, say, decide whether it is connected or cycle-free versus being far from having these properties? 
These questions are highly motivated by the growing need in processing massive graphs and have witnessed a flurry of results in recent years: see, e.g.,~\cite{KoganK15,KapralovKS15,KapralovKSV17,BhaskaraDV18,KapralovK19} on maximum cut, \cite{AssadiKL17,EsfandiariHLMO15,ChitnisCEHMMV16,McGregorV16,McGregorV18,CormodeJMM17,KapralovKS14} on maximum matching size,~\cite{Bar-YossefKS02,BravermanOV13,CormodeJ17,McGregorVV16,BeraC17,BulteauFKP16,KallaugherMPV19}
on subgraph counting,~\cite{GuruswamiVV17,GuruswamiT19,ChouGV20} on CSPs, and~\cite{HuangP16,MonemizadehMPS17,PengS18,CzumajFPS19} on property testing, among others (see also~\cite{VerbinY11,CzumajFPS19,AssadiKSY20} for a 
more detailed discussion of this line of work). 

Despite this extensive attention, the answer to these questions have remained elusive; except for a handful of problems and almost exclusively for {single-pass} algorithms, we have not yet found the ``right'' answers. 
For instance, consider property testing of connectivity: given a  {sparse} graph $G$ and a constant $\eps > 0$, find if $G$ is connected or requires at least $\eps \cdot n$ more edges to become so.  
Huang and Peng~\cite{HuangP16} proved that for {single-pass} algorithms, $n^{1-\Theta(\eps)}$ space is sufficient and necessary for this problem. But until very recently, it was even open if one could solve this problem 
in  $O(\log{n})$ space and two passes. This question was partially addressed by the first author, Kol, Saxena, and Yu~\cite{AssadiKSY20} who proved that any algorithm for this problem requires $n^{\Omega(1)}$ space or 
$\Omega(\log{(1/\eps)})$ passes. But this is still  far from the only known upper bound of $\polylog{(n)}$ space and $O(1/\eps)$ passes obtained via a streaming implementation of the  algorithm of~\cite{ChazelleRT05} 
(see~\cite{PengS18}).

Our goal in this paper is to make further progress on understanding the {limits of multi-pass graph streaming algorithms} for parameter estimation and property testing. 
We present a host of new multi-pass streaming lower bounds that in multiple cases such as property testing of connectivity, \textbf{achieve  optimal lower bounds on the space-pass tradeoffs} for the given problems. 
At the core of our results, similar to~\cite{VerbinY11,AssadiKSY20}, is a new lower bound for a \emph{``gap cycle counting''} problem, wherein the goal is to distinguish between graphs 
consisting of only ``short'' cycles or only ``long'' cycles. Our other streaming lower bounds then follow by easy reductions from this problem. 

Our proof techniques  are potentially useful for addressing other questions along these lines. 
We first use a ``decorrelation'' step to break the strong promise in the input graphs (that the cycles are either {all} short or all long) when proving the lower bound; this however comes at a cost of having to 
prove a lower bound for algorithms that succeed with a low probability of $1/2+1/\poly(n)$. The main ingredient of the proof is then a {hardness amplification} step which allows us to obtain such a 
lower bound from any standard lower bound, i.e., a one with not-so-little probability of success. The key to this argument is a \emph{streaming XOR Lemma}, in spirit of classical Yao's XOR
Lemma~\cite{Yao82a}, that we prove in this paper. We now elaborate on our results and techniques.

\subsection{Gap Cycle Counting  With a Little Bit of ``Noise''}\label{sec:ngc} 

Already a decade ago, Verbin and Yu~\cite{VerbinY11} identified a gap cycle counting problem as an excellent intermediate problem for studying the limitations of graph streaming algorithms for estimation problems: Given a graph $G$ and an integer $k$, 
decide if $G$ is a disjoint union of $k$-cycles or $2k$-cycles. By building on~\cite{GavinskyKKRW07}, they proved that this problem requires $n^{1-O(1/k)}$ space in a single pass and used this to establish lower bounds for several other problems. This work has since been a source of insights and inspirations for numerous other streaming lower bounds, e.g.~\cite{EsfandiariHLMO15,BuryS15,KoganK15,LiW16,HuangP16,GuruswamiVV17,BravermanCKLWY18,KapralovKS15,AssadiKL17,GuruswamiT19,KapralovKSV17,KapralovK19,KallaugherKP18,ChouGV20}. 
These lower bounds were all for single-pass algorithms. Very recently,~\cite{AssadiKSY20} proved that any $p$-pass streaming algorithm for gap cycle counting---and even a variant wherein the goal is 
to distinguish union of $k$-cycles from a  Hamiltonian cycle---requires $n^{1-O(k^{-1/2p})}$ space; in particular, $\Omega(\log{k})$ passes are needed to solve this problem with $\polylog{(n)}$ space. 
The work of~\cite{AssadiKSY20} showed that a large body of graph streaming lower bounds for estimation problems can now be extended to multi-pass algorithms using simple reductions from these gap cycle counting problems. 

A main question that was left explicitly open by both~\cite{VerbinY11,AssadiKSY20} was to determine the \emph{tight} space-pass tradeoff for these gap cycle counting problems (and by extension other streaming problems obtained via  reductions). 
We partially resolve this question by proving an asymptotically tight lower bound for a more relaxed variant that allows for some ``noise'' in the input. In particular, in our \emph{\underline{noisy} gap cycle counting} problem, 
the graph consists of a disjoint union of either $k$-cycles or $2k$-cycles on $\Theta(n)$ vertices, {plus} vertex-disjoint paths of length $k-1$ (the ``noise'') on the remaining vertices; the goal as before is to distinguish between the two cases 
(see~\Cref{def:NGCdef} and~\Cref{fig:k-2k-cycle}). 

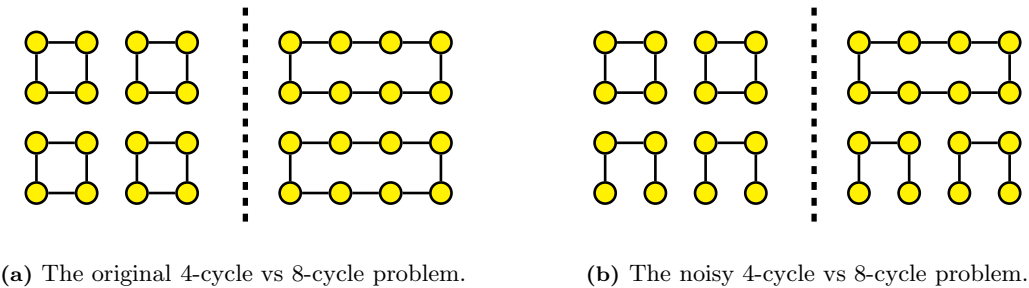
\begin{figure}[h!]	
\centering
\subcaptionbox{\footnotesize The original $4$-cycle vs $8$-cycle problem. \label{fig:k-2k-cycle1}}%
  [.45\linewidth]{\input{figs/k-2k-cycle1}}
\subcaptionbox{\footnotesize The noisy $4$-cycle vs $8$-cycle problem.\label{fig:k-2k-cycle2}}%
  [.45\linewidth]{\input{figs/k-2k-cycle2}}
\caption{An illustration of the graphs in the original gap cycle counting problem for $k=4$ versus the graphs in the new noisy gap cycle counting problem. The actual graph consists of
 $\Theta(n/k)$ copies of these smaller subgraphs.} 
\label{fig:k-2k-cycle}
\end{figure}

We prove the following lower bound for this noisy gap cycle counting problem.

\begin{result}\label{res:ngc}
	For any constant $k > 0$, any $p$-pass streaming algorithm for the noisy gap cycle counting problem requires $n^{1-O(p/k)}$ space to succeed with large constant probability.  
\end{result}

\Cref{res:ngc} obtains asymptotically optimal bounds for noisy gap cycle counting: on one end of the tradeoff, one can solve this problem in just one pass by sampling $\approx n^{1-1/k}$ random vertices and storing
all their edges to find a $k$-cycle or a $(k+1)$-path. On the other end, we can simply ``chase'' the neighborhood of $O(1)$ random vertices in $\approx k$ passes to solve the problem. 
 In the middle of these 
two extremes,  there is  the algorithm that samples $\approx n^{1-p/k}$ vertices and chase all of them in $p$ passes and ``stitch'' them together to form $k$-cycles or $(k+1)$-paths (see~\Cref{sec:alg-cycle}).~\Cref{res:ngc} matches 
all these tradeoffs asymptotically. Moreover, as a corollary, we  obtain that any  algorithm for this problem requires $n^{\Omega(1)}$ space or $\Omega(k)$ passes, \emph{exponentially} improving the bounds of~\cite{AssadiKSY20} (see also~\Cref{sec:technical-comparison} for a brief
technical comparison of our work with that of~\cite{AssadiKSY20}).

We remark that Verbin and Yu conjectured that 
any $p$-pass algorithm for this problem requires $n^{1-2/k}$ space as long as $k < p/2-1$\footnote{The conjecture of~\cite{VerbinY11} is stated more generally for two-party communication protocols and for the no-noise version of the problem;
 the statement here  is an immediate
corollary of this conjecture.}~\cite[Conjecture 5.4]{VerbinY11}. This conjecture as stated is too strong as the  $O(n^{1-p/k})$ space algorithm above refutes it already for $p > 2$. However,~\Cref{res:ngc} settles a qualitatively similar form of this conjecture which allows for an $n^{1-O(p/k)}$-space $p$-pass tradeoff.

\subsection{Graph Streaming Lower Bounds from Noisy Gap Cycle Counting} 

We  use our lower bound in~\Cref{res:ngc} in a similar manner as prior work to prove several new graph streaming lower bounds. The  difference is that we now 
have to  handle the extra noise in the problem; it turns out however that, as expected,  this noise does not have a serious effect on the reductions (it also helps that we prove~\Cref{res:ngc} in a stronger form where, informally speaking,
one endpoint of every noise path is already known to the algorithm; see~\Cref{cor:ngc-strong}). As a result, we can recover \emph{all} graph streaming lower bounds of~\cite{AssadiKSY20}  with a much stronger guarantee:

\begin{result}\label{res:str}
	For any $\eps > 0$, any $p$-pass algorithm for any of the following problems on $n$-vertex graphs requires $n^{1-O(\eps \cdot p)}$ space: 
	\vspace{-5pt}
	\begin{enumerate}[label=$-$]
		\item $(1+\eps)$-approximation of maximum matching size, maximum cut value,  maximum acyclic subgraph, and minimum spanning tree weight; 
		\item property testing of connectivity, bipartiteness, and cycle-freeness for  parameter $\eps$. 
	\end{enumerate}
	Moreover, these lower bounds continue to hold even on bounded-degree planar graphs. 
\end{result}
Prior to our work, $n^{1-O(\eps)}$ space lower bounds for single-pass algorithms have been obtained in \cite{KoganK15,KapralovKS15} for maximum cut, \cite{EsfandiariHLMO15,BuryS15} for maximum matching,  
\cite{GuruswamiVV17,ChouGV20} for maximum acyclic subgraph, \cite{FeigenbaumKMSZ05,HuangP16} for minimum spanning tree, and \cite{HuangP16} for the property testing problems. 
These results were recently extended by~\cite{AssadiKSY20} to $p$-pass algorithms with the space of $n^{1-O(\eps^{1/2p})}$ and thus $\Omega(\log{(1/\eps)})$ passes for $n^{o(1)}$-space algorithms.  
Our~\Cref{res:str}  exponentially improves the dependence on number of passes in~\cite{AssadiKSY20}, and in particular implies that any $n^{o(1)}$-space streaming algorithm for these problems require $\Omega(1/\eps)$ passes. 
For multiple of these problems, this bound can be matched by already known upper bounds and is thus optimal. We elaborate on these results further in~\Cref{sec:reductions}. 

We conclude this part by remarking that many of the problems we consider in~\Cref{res:str} have been also studied in \emph{random order streams}; see, e.g.~\cite{KapralovKS14,KapralovMNT20,CzumajFPS19,PengS18,MonemizadehMPS17}.  
In particular, Monemizadeh et.\,al.~\cite{MonemizadehMPS17} showed that $(1+\eps)$-approximation of matching size (in bounded-degree graphs) can be 
done in $O_{\eps}(\log{n})$ space  and a single pass if the edges  are arriving in a random order; similar bounds were 
obtained by Peng and Sohler~\cite{PengS18} for approximating the weight of minimum spanning tree (in bounded-weight graphs) and property testing of connectedness (see
also the work of Czumaj et.al.~\cite{CzumajFPS19} for a recent generalization of these results). Our~\Cref{res:str} thus demonstrate just how much harder solving these problems is in adversarial-order streams even with almost $1/\eps$ passes.

\subsection{Streaming XOR Lemma}\label{sec:intro-xor}

A key part of our proof of~\Cref{res:ngc} is a general hardness amplification step: Let $f$ be a Boolean function over a distribution $x \sim \mu$; for any integer $\ell > 1$, 
consider the $\ell$-fold-XOR-composition of $f$ over the distribution of inputs $x_1,\ldots,x_\ell \sim \mu^{\ell}$, namely, $\fxor := \XOR_{\ell} \circ f = f(x_1) \oplus \cdots \oplus f(x_\ell)$. How much harder is it to compute $\fxor$ compared to $f$? 
Notice that if solving $f$ (with certain resources) has success probability $\leq 1/2+\delta$, and that all the algorithm for $\fxor$ does is to solve each $f(x_i)$ {independently} and take their XOR, then
its success probability would be $\approx 1/2+\delta^{\ell}$. This is simply because XOR of $\ell$ independent random bits with bias $\delta$ only has bias $\approx \delta^{\ell}$ (see~\Cref{sec:xor-bias}).
Can a more clever strategy (with the same resources) beat this naive way of computing $\fxor$?

These questions are generally referred to as \emph{XOR Lemmas} and have been studied extensively in different settings like circuit complexity~\cite{Yao82a,GoldreichNW11,ImpagliazzoW97,Levin85,Impagliazzo95,GirishRZ20}, communication complexity~\cite{ViolaW08,Sherstov11,GirishRZ20}, and query complexity~\cite{Shaltiel03,Sherstov11,BrodyKLS20}. However, despite the extensive attention that similar hardness amplification questions 
such as direct sum and direct product have received in the streaming model (see, e.g.~\cite{Bar-YossefJKS02,JainRS03,JainPY12,BravermanRWY13,RaoS16,GuhaH09,MolinaroWY13,PhillipsVZ12} and references therein), we are not aware of any type of XOR Lemma for streaming algorithms. Thus, an important contribution of our work is to prove exactly such a result; considering its generality, we believe this result to be of independent interest. 

\begin{result}\label{res:xor}
	Suppose any $p$-pass $s$-space streaming algorithm for  $f$ over a distribution $\sigma \sim \mu$ succeeds with probability $\leq 1/2+\delta$. Then, any $p$-pass $s$-space algorithm for $\fxor$ over the 
	\emph{concatenation} of streams $(\sigma_1,  \cdots, \sigma_\ell) \sim \mu^\ell$ only succeeds with probability $\leq 1/2 + \delta^{\Omega(\ell)}$.
\end{result}

In~\Cref{sec:xor-examples}, we further discuss the notion of ``weak'' vs ``strong'' XOR Lemmas in the context of~\Cref{res:xor} and and in particular show the optimality of this result.  

\medskip
Let us now  mention how~\Cref{res:xor} is used in the proof of~\Cref{res:ngc}. 
Consider the following problem: given a graph $G$ in (noisy) gap cycle counting and a single vertex $v \in G$, ``chase'' the depth-$(k/2)$ neighborhood of $v$ to see if they form a $k$-cycle or a $(k+1)$-path. 
This problem is quite similar to the pointer chasing problem studied extensively in communication complexity and 
streaming, e.g., in~\cite{NisanW91,PonzioRV99,Yehudayoff16,ChakrabartiCM08,GuhaM08,GuhaM09,FeigenbaumKMSZ08,JainRS03,GuruswamiO13,AssadiCK19,BafnaGGS19,GolowichS20} (see~\Cref{def:PC}). 
The gap cycle counting problem then can be thought of as $\approx n/k$ instances of this problem that are highly \emph{correlated}: they are all in the same graph and  they all either form a $k$-cycle or a $(k+1)$-path. 
The first step of our lower bound is an argument that  ``decorrelates'' these instances which implies that one of them should be solved with probability of 
success $1/2+\Omega(k/n)$. This probability of success is  way below the threshold for any of the standard pointer chasing lower bounds to kick in. This is where we use our streaming XOR Lemma: we give a  reduction
that embeds XOR of $\ell$ instances of depth-$(k/2\ell)$ pointer chasing as a single depth-$(k/2)$ instance; applying our~\Cref{res:xor} then reduces our task to proving a 
lower bound for pointer chasing with probability of success $1/2+(k/n)^{\Theta(1/\ell})$ (in $k/2\ell$ passes), which brings us to the ``standard'' territory. The last step is then to prove this lower 
bound over our hard instances which are different from standard ones, e.g., in~\cite{GuhaM09,NisanW91,Yehudayoff16}. 

\subsection{Erratum and Recent Developments on Related XOR Lemmas}\label{sec:erratum} 

The proof of Streaming XOR Lemma in~\Cref{res:xor} in the previous version of the paper had missed a crucial step: it erroneously assumed an additional independence property which does not hold in general (specifically, the very last step
of the proof, namely, last equation in Lemma 3.4 of the previous version is not correct; see also~\Cref{rem:prior-bug} here). In this version, we replace this (wrong) assumption with a direct-product type argument (see~\Cref{lem:xor-1-copy}) and recover the original 
result up to constant factors (similar direct-product type arguments have also appeared in~\cite{AssadiS23,AssadiKZ24}). 

We shall also mention that the key motivation behind proving~\Cref{res:xor} was that, at the time of the publication of this result, no such XOR lemmas were known for communication complexity (even bounded-round). Had such a result been known already, 
we could have applied it in a black-box way in conjunction with the rest of the arguments in this paper and obtain our lower bounds in a simpler way (even in the two-player communication model, similar to~\cite{VerbinY11,AssadiKSY20} instead of only focusing on 
streaming algorithms). Since the publication of our results in 2021, such XOR Lemmas have been developed in~\cite{Yu22} and subsequently in~\cite{IyerR24}. As such, one can now alternatively use these (stronger) results and obtain 
our~\Cref{res:ngc} and~\Cref{res:str} (even in the stronger two-player communication complexity model) without reliance on~\Cref{res:xor}. 

%% file: figs/k-2k-cycle1.tex

\begin{tikzpicture}

	\node[vertex] (A11){};
	\node[vertex] (A12)[right=10pt of A11]{};
	\node[vertex] (A13)[below=10pt of A11]{};
	\node[vertex](A14)[right=10pt of A13]{};
	
	\node[vertex](B11)[right=10pt of A12]{};
	\node[vertex] (B12)[right=10pt of B11]{};
	\node[vertex] (B13)[below=10pt of B11]{};
	\node[vertex] (B14)[right=10pt of B13]{};
	
	\node[vertex] (C11)[below=10pt of A13]{};
	\node[vertex] (C12)[right=10pt of C11]{};
	\node[vertex] (C13)[below=10pt of C11]{};
	\node[vertex] (C14)[right=10pt of C13]{};
	
	\node[vertex](D11)[right=10pt of C12]{};
	\node[vertex](D12)[right=10pt of D11]{};
	\node[vertex](D13)[below=10pt of D11]{};
	\node[vertex](D14)[right=10pt of D13]{};
	
	\draw[line width=1pt, black] 
	(A11) -- (A12)
	(A12) -- (A14)
	(A14) -- (A13)
	(A13) -- (A11)
	
	(B11) -- (B12)
	(B12) -- (B14)
	(B14) -- (B13)
	(B13) -- (B11)
	
	(C11) -- (C12)
	(C12) -- (C14)
	(C14) -- (C13)
	(C13) -- (C11)
	
	(D11) -- (D12)
	(D12) -- (D14)
	(D14) -- (D13)
	(D13) -- (D11);

	\node (T) [above right=10pt and 15pt of B12]{};
	\node (B) [below right=10pt and 15pt of D14]{};
	
	\draw[line width=2pt, black, dashed] (T) -- (B);
	
	\node[vertex] (xA11)[right=30pt of B12]{};
	\node[vertex] (xA12)[right=10pt of xA11]{};
	\node[vertex] (xA13)[below=10pt of xA11]{};
	\node[vertex] (xA14)[right=10pt of xA13]{};
	
	\node[vertex] (xB11)[right=10pt of xA12]{};
	\node[vertex] (xB12)[right=10pt of xB11]{};
	\node[vertex] (xB13)[below=10pt of xB11]{};
	\node[vertex] (xB14)[right=10pt of xB13]{};
	
	\node[vertex] (xC11)[below=10pt of xA13]{};
	\node[vertex] (xC12)[right=10pt of xC11]{};
	\node[vertex] (xC13)[below=10pt of xC11]{};
	\node[vertex] (xC14)[right=10pt of xC13]{};
	
	\node[vertex] (xD11)[right=10pt of xC12]{};
	\node[vertex] (xD12)[right=10pt of xD11]{};
	\node[vertex] (xD13)[below=10pt of xD11]{};
	\node[vertex] (xD14)[right=10pt of xD13]{};
	
	\draw[line width=1pt, black] 
	(xA11) -- (xA12)
	(xA12) -- (xB11)
	(xA14) -- (xA13)
	(xA13) -- (xA11)
	
	(xB11) -- (xB12)
	(xB12) -- (xB14)
	(xB14) -- (xB13)
	(xB13) -- (xA14)
	
	(xC11) -- (xC12)
	(xC12) -- (xD11)
	(xC14) -- (xC13)
	(xC13) -- (xC11)
	
	(xD11) -- (xD12)
	(xD12) -- (xD14)
	(xD14) -- (xD13)
	(xD13) -- (xC14);

\end{tikzpicture}

%% file: figs/k-2k-cycle2.tex

\begin{tikzpicture}

	\node[vertex] (A11){};
	\node[vertex] (A12)[right=10pt of A11]{};
	\node[vertex] (A13)[below=10pt of A11]{};
	\node[vertex] (A14)[right=10pt of A13]{};
	
	\node[vertex] (B11)[right=10pt of A12]{};
	\node[vertex] (B12)[right=10pt of B11]{};
	\node[vertex] (B13)[below=10pt of B11]{};
	\node[vertex] (B14)[right=10pt of B13]{};
	
	\node[vertex] (C11)[below=10pt of A13]{};
	\node[vertex] (C12)[right=10pt of C11]{};
	\node[vertex] (C13)[below=10pt of C11]{};
	\node[vertex] (C14)[right=10pt of C13]{};
	
	\node[vertex] (D11)[right=10pt of C12]{};
	\node[vertex] (D12)[right=10pt of D11]{};
	\node[vertex] (D13)[below=10pt of D11]{};
	\node[vertex] (D14)[right=10pt of D13]{};
	
	\draw[line width=1pt, black] 
	(A11) -- (A12)
	(A12) -- (A14)
	(A14) -- (A13)
	(A13) -- (A11)
	
	(B11) -- (B12)
	(B12) -- (B14)
	(B14) -- (B13)
	(B13) -- (B11)
	
	(C11) -- (C12)
	(C12) -- (C14)
	(C13) -- (C11)
	
	(D11) -- (D12)
	(D12) -- (D14)
	(D13) -- (D11);

	\node (T) [above right=10pt and 15pt of B12]{};
	\node (B) [below right=10pt and 15pt of D14]{};
	
	\draw[line width=2pt, black, dashed] (T) -- (B);
	
	\node[vertex] (xA11)[right=30pt of B12]{};
	\node[vertex] (xA12)[right=10pt of xA11]{};
	\node[vertex] (xA13)[below=10pt of xA11]{};
	\node[vertex] (xA14)[right=10pt of xA13]{};
	
	\node[vertex] (xB11)[right=10pt of xA12]{};
	\node[vertex] (xB12)[right=10pt of xB11]{};
	\node[vertex] (xB13)[below=10pt of xB11]{};
	\node[vertex] (xB14)[right=10pt of xB13]{};
	
	\node[vertex] (xC11)[below=10pt of xA13]{};
	\node[vertex] (xC12)[right=10pt of xC11]{};
	\node[vertex] (xC13)[below=10pt of xC11]{};
	\node[vertex] (xC14)[right=10pt of xC13]{};
	
	\node[vertex] (xD11)[right=10pt of xC12]{};
	\node[vertex] (xD12)[right=10pt of xD11]{};
	\node[vertex] (xD13)[below=10pt of xD11]{};
	\node[vertex] (xD14)[right=10pt of xD13]{};
	
	\draw[line width=1pt, black] 
	(xA11) -- (xA12)
	(xA12) -- (xB11)
	(xA14) -- (xA13)
	(xA13) -- (xA11)
	
	(xB11) -- (xB12)
	(xB12) -- (xB14)
	(xB14) -- (xB13)
	(xB13) -- (xA14)
	
	(xC11) -- (xC12)
	(xC12) -- (xC14)
	(xC13) -- (xC11)
	
	(xD11) -- (xD12)
	(xD12) -- (xD14)
	(xD13) -- (xD11);


\end{tikzpicture}

%% file: 2prelim.tex

\section{Notation and Preliminaries}\label{sec:prelim}

\paragraph{Notation.} For a Boolean function $f$ and integer $\ell \geq 1$, we use $\fxor$ to denote the composition of $f$ with the $\ell$-fold XOR function: $\fxor(x_1,\ldots,x_\ell) = f(x_1) \oplus \cdots \oplus f(x_\ell)$. 
We denote the input stream by $\sigma$ and $\card{\sigma}$ denote the length of the stream. For any two streams $\sigma_1, \sigma_2$, we use $\sigma_1 \conc \sigma_2$ to denote the $\card{\sigma_1} + \card{\sigma_2}$ 
length stream obtained by concatenating $\sigma_2$ at the end of $\sigma_1$. When it can lead to confusion, we use sans serif font for random variables (e.g. \rv X) and normal font for their realization (e.g. X). 
We use $\supp{\rX}$ to denote the support of random variable $\rX$. For a $0/1$-random variable $\rX$, we define the \emph{bias} of $\rX$ as $\bias(\rX) := {\card{\Pr(\rX=0)-\Pr(\rX=1)}}$; see~\Cref{sec:xor-bias} for more details. 

\paragraph{Information theory.} For random variables $\rX,\rY$, $\en{\rX}$ denotes the Shannon entropy of $\rX$, $\mi{\rX}{\rY}$ denotes the mutual information, $\tvd{\rX}{\rY}$ denotes the total variation distance between the distributions of 
$\rX,\rY$, and $\kl{\rX}{\rY}$ is their KL-divergence. \Cref{sec:info} contains the definitions and standard background on these notions that we need in our proofs. 

\paragraph{Streaming algorithms.} For the purpose of our lower bounds, we shall work with a {more powerful} model than what is typically considered the streaming model (this is the common approach when proving streaming
lower bounds; see, e.g.~\cite{GuhaM08,LiNW14,BravermanGW20}). In particular, we shall define streaming algorithms 
as multi-party communication protocols\footnote{We refer the  reader to~\cite{KushilevitzN97} for the standard 
	definitions from communication complexity used in this paper.} as follows, and then point out the subtle differences with what one typically expect of a streaming algorithm.  

\begin{definition}[\textbf{Streaming algorithms}]\label{def:str}
	For any integers $n,p,s \geq 1$, we define a $p$-pass $s$-space streaming algorithm working on a length-$n$ stream $\sigma=(x_1,\ldots,x_n)$ as a $(n+1)$-player communication protocol between players $P_0,\ldots,P_n$ wherein: 
	\begin{enumerate}[label=$(\roman*)$]
		\item Each player $P_i$ for $i \in [n]$ receives $x_i$ as the \emph{input} and player $P_0$ has no input; the players also have access to \emph{private randomness}.
		\item The players communicate in this order: $P_0$ sends a message to $P_1$ who sends a message to $P_2$ and so on up until $P_n$ who sends a message to $P_0$; this constitutes one \emph{pass} of the algorithm. 
		The players then continue like this for $p$ passes and at the end, $P_0$ outputs the answer. 
		\item Each message of a player in a given round is an \emph{arbitrary} function of its input, 
		\emph{all} the messages received by this player so far, and its private randomness and has size $s$ bits \emph{exactly}. 
	\end{enumerate}
	(We note that this model is non-uniform and is defined for each choice of $n$ individually.)
\end{definition}

Let us  point out  a couple of differences with what one may expect of streaming algorithms.  Firstly, we allow our streaming algorithms to do an unbounded amount of work using an unbounded amount of space \emph{between} the
arrival of each stream element; we only bound the space in transition between two elements. Secondly, we allow streaming algorithms to maintain a 
``state'' for each stream element  across multiple passes (as each player of the streaming algorithm ``remembers'' the messages it receives in previous passes as well). 
Finally, a player $P_0$ is introduced for notational convenience so that every pass of the algorithm involves one message per main players $P_1,\ldots,P_n$. 

Clearly, any lower bound proven for streaming algorithms in~\Cref{def:str} will hold also for more restrictive definitions of streaming algorithm, and that is what we use in this paper. 
We shall note that however almost all streaming lower bounds we are aware of directly 
work with this definition and thus we claim no strengthening in proving our lower bounds under this definition; rather, we merely use this formalism to carry out various reductions between our problems.

%% file: 3.5direct-product.tex

\section{Streaming XOR Lemma}\label{sec:xor-lemma}

Let $\Sigma_n$ be any collection of length-$n$ input streams and $f: \Sigma_n \rightarrow \set{0,1}$ be a function which can be interpreted as a streaming decision problem: 
Given a length-$n$ stream $\sigma \in \Sigma_n$, output $f(\sigma)$. Using $f$, and for any integer $\ell \geq 1$, we can define another streaming decision problem over 
length $(n \ell)$-streams: For $\sigma_1,\ldots,\sigma_\ell \in {\Sigma_n}$, compute $\fxor(\sigma_1,\ldots,\sigma_\ell)$ over the stream $\sigma_1 \conc \sigma_2 \conc \cdots \conc \sigma_\ell$. 
We  prove the following Streaming XOR Lemma for computing $\fxor$ in this section. 

\begin{theorem}[\textbf{Streaming XOR Lemma}]\label{thm:xor-lemma}
	Fix any function $f: \Sigma_n \rightarrow \set{0,1}$, any distribution $\mu$ on $\Sigma_n$, and any integer $\ell > 1$. Suppose any $p$-pass $s$-space streaming algorithm can compute $f$ 
	over $\mu$ with probability at most $\frac{1}{2} + \delta$ for sufficiently small $\delta > 0$. Then, any $p$-pass $s$-space algorithm for $\fxor$ on the stream $\sigma_1 \conc \cdots  \conc\sigma_\ell$ for $(\sigma_1,\ldots,\sigma_\ell) \sim \mu^{\ell}$ 
	succeeds with probability at most $\frac{1}{2} + \delta^{\ell/7}$. 
\end{theorem}

Recall the intuition at the beginning of~\Cref{sec:intro-xor} behind any form of XOR Lemma: taking XOR of \emph{independent} bits dampens their biases exponentially and thus the algorithm
for $\fxor$ that computes each $f(\sigma_i)$ individually satisfies~\Cref{thm:xor-lemma}. In general however, we cannot expect the algorithm to approach these subproblems independently as it may 
instead try to \emph{correlate} its success probabilities across different subproblems (say, with probability $1/2+\delta$ all subproblems are correct and with the remaining probability, all are wrong). 
This is the main barrier in proving any form of XOR Lemma and what we need to overcome in proving~\Cref{thm:xor-lemma}. 

The main ideas of the proof consist of the following  steps: (a) set up a $\ell$-player ``communication game'', with one player per  $\sigma_i$, whose lower bounds also imply lower bounds for streaming algorithms of $\fxor$, 
(b) give enough extra power to this game so that no player is responsible for compressing the input of another player, and show that the players' success in computing each $f(\sigma_i)$ becomes uncorrelated, (c) limit the power of the game
so that streaming lower bounds for $f$ also imply lower bounds for computing each $f(\sigma_i)$ in this game. We now formalize this in the following:

\subsection*{Proof of~\Cref{thm:xor-lemma}}

We set up the following game for proving this theorem (see also~\Cref{fig:game}): 

\begin{tbox}
\begin{enumerate}[label=$(\roman*)$,leftmargin=25pt]
	\item There are a total of $\ell$ players $Q_1,\ldots,Q_\ell$ who receive input streams $\sigma_1,\ldots,\sigma_\ell$, respectively. 
	
	\item The players  communicate with each other in \emph{rounds} via a \emph{blackboard}. In each round, 
	the players go in turn with $Q_1$ writing a message on the board, followed by $Q_2$, all the way to $Q_\ell$; these messages are visible to everyone (and are not altered or erased after  written). 

	\item For any player $Q_i$ and round $j$, we use $M^j_i$ to denote the message written on the board by $Q_i$ in the $j$-th round. 
	We additionally use $B^j_i$ to denote the content of the board \emph{before} the message $M^j_i$ is written and $B^j$ to denote the content of the board \emph{after} round $j$. 	 
	
	\item The messages of each $Q_i$ is generated by a \emph{deterministic multi-pass streaming} algorithm $\alg_i$ that runs on $\sigma_i$ (with one inner player per  element of 
	the stream as in~\Cref{def:str}). In each round $j$, the player $P_0$ of $\alg_i$ is additionally given the content of the board $B^j_i$, then $\alg_i$  makes its $j$-th pass over $\sigma_i$, 
	and then $P_0$ of $\alg_i$ outputs $M^j_i$ on the board. 
	
	\item The \emph{cost} of a protocol is the \underline{maximum size of the memory} of any algorithm $\alg_i$. 
\end{enumerate}
\end{tbox}
Let us emphasize that this game is not at all a standard communication complexity problem: in our game, the communication between the players is \emph{unbounded} and the cost of 
the algorithm is instead governed by the memory of streaming algorithms run by each player as opposed to having computationally unbounded players.

\begin{figure}[t!]	
\centering
\input{figs/game}
\caption{An illustration of the communication game in the proof of~\Cref{thm:xor-lemma} for $n=3$ and $\ell=3$. The solid (red) lines draw the messages of players of each inner streaming algorithm, while dashed (blue) lines draw the 
communicated messages between the players of the game and the blackboard (from player $P_0$ of each inner streaming algorithm). } 
\label{fig:game}
\end{figure}
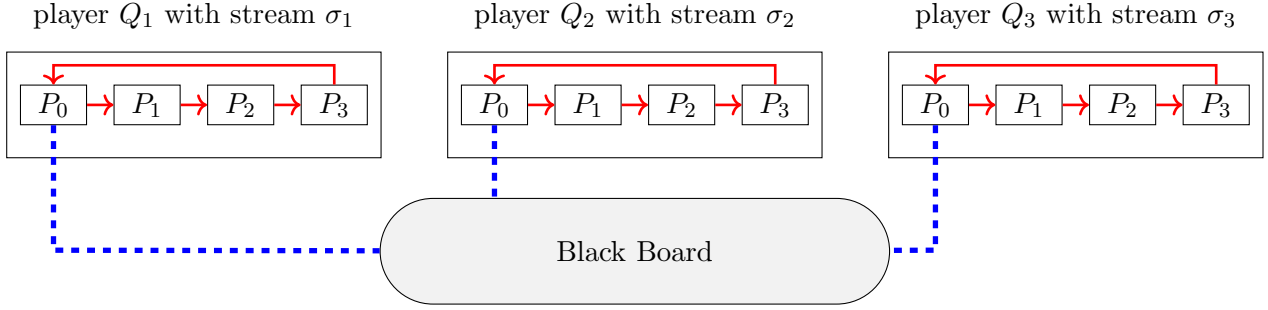

We first show that if we can lower bound the cost of protocols in this game, we immediately get a lower bound for  streaming algorithms of $\fxor$.

\begin{lemma}\label{lem:game-to-stream}
	Any $p$-pass $s$-space algorithm $A$ (deterministic or randomized) for computing $\fxor$  implies a (deterministic) $p$-round protocol $\prot$ with cost at most $s$ and the same probability of success.  
\end{lemma}

\begin{proof}
	Without loss of generality, we can assume $A$ is deterministic 
	as by an averaging argument, there is a fixing of the randomness of the algorithm that gives the same success probability over $\mu^{\ell}$. 
	
	To avoid confusion, let us denote the players of $A$ as $R_0,R_1,\ldots,R_{n \cdot \ell}$. 
	We can generate a protocol $\pi$ from $A$ as follows: 
	\begin{enumerate}[label=$(\roman*)$]
		\item $Q_1$ runs $A$ as $\alg_1$ on $\sigma_1$: $P_0$ (of $\alg_1$) sends a message to $P_1$ and so on until $P_n$ by running the first pass of $A$ over their inputs by simulating $R_0$ to $R_n$ (this incurs a cost of $s$). 
		\item At this point, $P_n$ has the same input and message as player $R_n$ of $A$. Thus, $P_n$ can send the message of $R_n$ to $R_{n+1}$ instead to $P_0$ which finishes the first pass of $\alg_1$ (again by a cost of only $s$ as the message of $R_n$ to $R_{n+1}$ has size $s$). 
		$P_0$ of $\alg_1$ then can post this received message on the blackboard as message $M^1_1$ of player $Q_1$.
		\item Now it is $Q_2$'s turn to run $A$ as $\alg_2$ on $\sigma_2$: $P_0$ (of $\alg_2)$ reads the content of the board and pass it along to $P_1$; this way, $P_1$ to $P_n$ can continue the first pass of $A$ over their inputs by simulating 
		$R_{n+1}$ to $R_{2n}$ (again, cost of only $s$ by the space bound of $A$). Then, like step $(ii)$, the message of $R_{2n}$ to $R_{2n+1}$ will be posted on the board via $P_0$ of $\alg_2$. 
		\item The players continue like this until they run every $p$ passes of $A$ in $p$ rounds over their inputs and output the same answer. 
	\end{enumerate}
	As the  cost of this protocol is $s$ and it fully simulates $A$, we obtain the result. 
\end{proof}

Fix a $p$-round communication protocol $\pi$ with cost at most $s$ in this game and suppose the inputs of players are sampled from the product distribution $\mu^\ell$. For the rest of the proof, we bound 
the probability of success of $\pi$ which will imply~\Cref{thm:xor-lemma} by~\Cref{lem:game-to-stream}. 

To continue we need the following definitions: 
\begin{itemize}
\item	For any $i \in [\ell]$ and any choice of the final board content $\rB^p=B$, we define: 
\[
\bias_{\pi}(i,B) := 2 \cdot \max_{\theta \in \set{0,1}} \Pr_{(\sigma_1,\ldots,\sigma_\ell) \sim \mu^\ell}\paren{f(\sigma_i) = \theta \mid \rB^p=B} - 1;
\]
in other words, $\bias_{\pi}(i,B)$ is equal to $\bias\paren{f(\sigma_i)}$ for  $\sigma_i \sim \mu^{\ell} \mid \rB^p=B$. 
\item For any choice of the final board content $\rB^p=B$, we define:
\[
\bias_{\pi}(B) := 2 \cdot \max_{\theta \in \set{0,1}} \Pr_{(\sigma_1,\ldots,\sigma_\ell) \sim \mu^\ell}\paren{\fxor(\sigma_1,\ldots,\sigma_\ell) = \theta \mid \rB^p=B} - 1; 
\]
in other words, $\bias_{\pi}(B)$ is equal to $\bias\paren{f(\sigma_1) \oplus \cdots \oplus f(\sigma_\ell)}$ for  $\sigma_1,\ldots,\sigma_\ell \sim \mu^{\ell} \mid \rB^p=B$. 
\end{itemize}
With these definitions, we have, 
\begin{align}
	\Pr_{(\sigma_1,\ldots,\sigma_\ell) \sim \mu^\ell}\paren{\text{$\pi$ is correct}} &= \Ex_{B}\bracket{\Pr_{(\sigma_1,\ldots,\sigma_\ell) \sim \mu^\ell}\paren{\text{$\pi$ is correct} \mid \rB^p=B}} \notag \\
	&\leq \Ex_{B}\bracket{\max_{\theta \in \set{0,1}} \Pr_{(\sigma_1,\ldots,\sigma_\ell) \sim \mu^\ell}\paren{\fxor(\sigma_1,\ldots,\sigma_\ell) = \theta \mid \rB^p=B}} \tag{conditioned on $\rB^p=B$, the answer of $\pi$ is fixed to some $\theta \in \set{0,1}$} \\
	&= \Ex_{B}\bracket{\frac{1}{2} \cdot \paren{1+\bias_{\pi}(B)}} = \frac{1}{2} + \frac{1}{2} \cdot \Ex_B\bracket{\bias_\pi(B)} \label{eq:bias-pi}.
\end{align}

As such, $\Ex_B[\bias_\pi(B)]$ measures the {advantage} of $\pi$ in outputting the answer over random guessing. 
Our goal in the remainder of this section is to bound this expectation. 

To prove our result, we would like to use the fact that XOR dampens the bias of \emph{independent} bits (see~\Cref{sec:xor-bias}). 
Thus, we need to establish that these $f(\sigma_i)$ bits are not correlated after conditioning on $\rB^p$, which is done in the following lemma. This uses an analogue 
of the \textbf{rectangle property} of standard communication protocols but for the particular game we defined. 

\begin{lemma}\label{lem:xor-independence}
	For any  $B$, 
	$
	\paren{\sigma_1,\ldots,\sigma_\ell \sim \mu^\ell \mid \rB^p=B} = \bigtimes_{i=1}^{\ell} \paren{\sigma_i \sim \mu^{\ell} \mid \rB^p=B},
	$
	 i.e., the input streams $\sigma_i$'s are independent \underline{even} conditioned on $\rB^p=B$. 
\end{lemma}
\begin{proof}
	The input streams are originally independent, so we need to show that the protocol $\pi$ in this game cannot correlate them after we condition on $B$. 

	Order all $d := p \cdot \ell$ messages in the order of being communicated as $M_1,M_2,\ldots,M_{d}$ where $M_i$ for $i \in [d]$ is $M^j_k$ for right choices of $j \in [p]$, $k \in [\ell]$ (basically $M_{(j-1) \cdot \ell+k} = M^j_k$).
	We claim that for any $d' \in [d]$, and any $M_{\leq d'} = M_1,\ldots,M_{d'}$, there are sets $\mathcal{X}_1(M_{\leq d'}), \ldots, \mathcal{X}_{\ell}(M_{\leq d'})$ such that:  
	\begin{align}
		\set{\text{$(\sigma_1,\ldots,\sigma_\ell)$ creates partial transcript $M_1,\ldots,M_{d'}$}} = \mathcal{X}_1(M_{\leq d'}) \times \cdots \times \mathcal{X}_{\ell}(M_{\leq d'}). \label{eq:pi-rectangle-property}
	\end{align}
	 We prove this by induction. The base case for $d'=0$ holds trivially because the inputs are sampled from $(\Sigma_n)^{\ell}$. Suppose this is true for $d'$ and we prove it for $d'+1$. 
	 The next message communicated, namely, $M_{d'+1}$, is only a function of $M_1,\ldots,M_{d'}$ and the input $\sigma_i$ of the player communicating $M_{d'+1}$. Thus, this message
	 further shrinks the set $\mathcal{X}_i(M_{\leq d'})$ to some $\mathcal{X}_i(M_{\leq d'+1})$ to be consistent with $M_{d'+1}$ and keeps all other $\mathcal{X}_j(M_{\leq d'+1}) = \mathcal{X}_j(M_{\leq d'})$ for $j \neq i$. 
	 This proves~\Cref{eq:pi-rectangle-property}. 
	 
	 The lemma now follows from~\Cref{eq:pi-rectangle-property} for $d'=d$ and since the original distribution of inputs are independent. 
\end{proof}

An immediate consequence of this lemma is the following. 

\begin{claim}\label{clm:xor-products}
	For any $B$, 
	 $
	 	\bias_\pi(B) = \prod_{i=1}^{\ell} \bias_{\pi}(i,B).
	 $
\end{claim}
\begin{proof}
	Fix any $B$ and consider the random variables $f(\sigma_1),\ldots,f(\sigma_\ell)$ for $(\sigma_1,\ldots,\sigma_\ell) \sim \mu^\ell \mid \rB=B$. 
	By~\Cref{lem:xor-independence}, even in the distribution $\mu^{\ell} \mid \rB=B$, $\sigma_i$'s are independent which implies that 
	$f(\sigma_1),\ldots,f(\sigma_\ell)$ are also independent random variables conditioned on $B$. As such, for any $B$, 
	\begin{align*}
		\bias_{\pi}(B) = \bias(f(\sigma_1) \oplus \ldots \oplus f(\sigma_\ell) \mid \rB=B) = \prod_{i=1}^{\ell} \bias(f(\sigma_i) \mid \rB=B) = \prod_{i=1}^{\ell} \bias_{\pi}(i,B), 
	\end{align*}
	where the first and last equalities are by the definitions of $\bias_{\pi}(B)$ and $\bias_{\pi}(i,B)$, and the middle equality is by~\Cref{prop:xor-dampen-bias} and the independence of $f(\sigma_i)$'s conditioned on $B$, namely, the fact that XOR 
	dampens the biases of independent random bits.  
\end{proof}
\noindent
Thus, by plugging in the bounds of~\Cref{clm:xor-products} in~\Cref{eq:bias-pi}, we will have that, 
\begin{align}
	\Pr_{(\sigma_1,\ldots,\sigma_\ell) \sim \mu^\ell}\paren{\text{$\pi$ is correct}} &= \frac{1}{2} + \frac{1}{2} \cdot \Ex_B\bracket{\prod_{i=1}^{\ell} \bias_{\pi}(i,B)} \label{eq:bias-pi-2}.
\end{align}
\begin{remark}\label{rem:prior-bug}
\emph{If} we had an independence guarantee across players to say 
\[
\Ex_B\bracket{\prod_{i=1}^{\ell} \bias_{\pi}(i,B)} = \prod_{i=1}^{\ell} \Ex_B[\bias_{\pi}(i,B)],
\] 
we could have finalized the proof using a direct ``embedding argument'' (by reducing the $\ell$-player game to a $1$-player game and invoking the bound on the biases in the $1$-player game from the theorem statement). 
However, such an independence guarantee is not warranted here.  In the rest of the proof, we show how to implement such an embedding argument
based on a direct-product type argument without relying on this false assumption (see~\Cref{lem:xor-1-copy}). 
\end{remark}
To continue, we need another set of definitions:
\begin{itemize}
	 \item For any $i \in [\ell]$ and any choice of the final board content $\rB=B$, we define the event: 
	\[
	\event_{\pi}(i,B) := \mathds{1}(\bias_{\pi}(i,B) > 2\sqrt{\delta});
	\]
in other words, $\event_{\pi}(i,B)$ denotes the event that bias of $f(\sigma_i)$ for  $\sigma_i \sim \mu^{\ell} \mid \rB=B$ is ``high''. 
	\item For any $S \subseteq [\ell]$ and any choice of the final board content $\rB=B$, we define the event: 
	\[
		\event_\pi(S,B) := \bigwedge_{i \in S} \event_{\pi}(i,B),
	\]
namely, the event that $\event_{\pi}(S,B)$ happens for all $i \in S$.
	\item Finally, for any choice of the final board content $\rB=B$, define the event: 
	\[
		\event_{\pi}(B) := \mathds{1}(\text{there exists some $S \subseteq [\ell]$ with $\card{S} = 2\ell/3$ such that $\event_{\pi}(S,B)$ holds)};
	\]
in other words, $\event_{\pi}(B)$ is the event that $\event_{\pi}(i,B)$ holds for at least two thirds of indices in $[\ell]$. 
\end{itemize}

We show that the probability of $\event_{\pi}(B)$ is quite low. This is a combination of two parts. Firstly, the protocol $\pi$ should not be able to change the bias of any single 
$f(\sigma_i)$ ``easily'': intuitively, this should be  
true as $\pi$ is effectively running an $s$-space streaming algorithm $\alg_i$ on $\sigma_i$ and so we can apply the assumption of~\Cref{thm:xor-lemma}. Secondly, 
the probability that the protocol $\prot$ can change the bias of $f(\sigma_i)$  \emph{simultaneously} for $\Omega(\ell)$ indices $i \in [\ell]$ should be even exponentially smaller: this can be proven
using a direct-product type argument. Both these parts are captured in the next lemma. 

\begin{lemma}\label{lem:xor-1-copy}
	Consider any $S \subseteq [\ell]$ and suppose that 
	\[
	\Pr\paren{\event_\pi(S,B)} > (2\sqrt{\delta})^{\card{S}}.
	\]
	Then, there is a $p$-pass $s$-space streaming algorithm for computing $f(\sigma)$ on $\mu$ 
	with probability of success more than $1/2 + \delta$. 
\end{lemma}

We postpone the proof of this lemma to the next section and for now use it to conclude the proof of~\Cref{thm:xor-lemma}. 
Recall the assumption in the statement of \Cref{thm:xor-lemma} that any $p$-pass $s$-space algorithm for $f(\sigma)$ on $\sigma \sim \mu$ cannot succeed with probability better than $1/2+\delta$. 
Thus, we can alternatively see~\Cref{lem:xor-1-copy} as proving an upper bound on the probability that $\event_{\pi}(S,B)$ happens for a given set $S \subseteq [\ell]$. 
Combining this with a union bound, we have that 
\[
	\Pr\paren{\event_{\pi}(B)} \leq \sum_{S \subseteq [\ell]: \card{S}=2\ell/3} \hspace{-10pt} \Pr\paren{\event_\pi(S,B)} \leq 2^{\ell} \cdot (2\sqrt{\delta})^{2\ell/3}.
\]
We can use this to prove an upper bound on the probability of success of $\pi$ as follows. 
\begin{align*}
	\Ex_B\bracket{\prod_{i=1}^{\ell} \bias_{\pi}(i,B)} 	&\leq  \Pr_B\paren{\event_\pi(B)} + \Ex_B\bracket{\prod_{i=1}^{\ell} \bias_{\pi}(i,B) \mid \overline{\event_\pi(B)}} \tag{as each dropped term is upper bounded by $1$} \\
	&\leq  2^{\ell} \cdot (2\sqrt{\delta})^{2\ell/3} + \Ex_B\bracket{\prod_{i=1}^{\ell} \bias_{\pi}(i,B) \mid \overline{\event_\pi(B)}} \tag{by the equation above} \\
	&\leq 2^{\ell} \cdot (2\sqrt{\delta})^{2\ell/3}+  (2\sqrt{\delta})^{\ell/3} \tag{under $\overline{\event_\pi(B)}$, for $\ell/3$ indices, $\bias_{\pi}(i,B) \leq (2\sqrt{\delta})$ by definition and for remaining it is at most $1$}\\
	&\leq 2 \cdot \delta^{\ell/7} \tag{for sufficiently small $\delta > 0$}. 
\end{align*}
Plugging these bounds in~\Cref{eq:bias-pi-2} gives us, 
\[
	\Pr_{(\sigma_1,\ldots,\sigma_\ell) \sim \mu^\ell}\paren{\text{$\pi$ is correct}} \leq \frac{1}{2} +  \delta^{\ell/7}
\]
for $\delta \leq \delta_0$ for sufficiently small constant $\delta_0$. This concludes the proof of~\Cref{thm:xor-lemma}. 

We now prove~\Cref{lem:xor-1-copy} starting with a warm-up for the case of $p=1$, namely, single-pass algorithms, which showcases many of the main ideas but with a simpler setup. We then show 
how to extend this argument to multi-pass algorithms as well. 
\subsection{Warm-Up: Proof of~\Cref{lem:xor-1-copy} for Single-Pass Algorithms} 

We prove the single-pass version of the arguments in~\Cref{lem:xor-1-copy} first. 
\begin{lemma*}[Special case of \Cref{lem:xor-1-copy} for single-pass algorithms]
Suppose we are working with single-pass algorithms only. Consider any $S \subseteq [\ell]$ and suppose that 
	\[
	\Pr\paren{\event_\pi(S,B)} > (2\sqrt{\delta})^{\card{S}}.
	\]
	Then, there is a single-pass $s$-space streaming algorithm for computing $f(\sigma)$ on $\mu$ 
	with probability of success more than $1/2 + \delta$. 
\end{lemma*}

To prove this lemma, we will turn $\pi$ into a $s$-space streaming algorithm $A$ for computing
$f(\sigma_i)$ on the stream $\sigma_i\sim\mu$ (and not the entire input). It might be helpful to
consult~\Cref{fig:game} when reading this part. For simplicity of notation,
denote $W:=\event_{\pi}(S,B)$ to be the event under consideration. Moreover, throughout the proof, we also use $\rX_i$ to denote the random variable for the stream $\sigma_i$. 

Recall that the messages in the protocol $\pi$ are $M_1,\ldots,M_\ell$ (given we are now only considering single-pass algorithms). For any fixed choice of $i \in [\ell]$ and messages $M_{\leq i} := (M_1,\ldots,M_i)$, 
define the following distribution for remaining messages $M_{>i} := (M_{i+1},\ldots,M_\ell)$: 
\begin{align}
	\mu_{W,i}(\rM_{>i} \mid M_{\leq i}) := 
	\begin{cases} 
		\mu^\ell(\rM_{>i} \mid \rM_{\leq i} = M_{\leq i} , W) \quad &\Pr_{\mu^\ell}\paren{W \mid \rM_{\leq i} = M_{\leq i}} > 0 \\
		\mu^\ell(\rM_{>i} \mid \rM_{\leq i} = M_{\leq i}) \quad &\Pr_{\mu^\ell}\paren{W \mid \rM_{\leq i} = M_{\leq i}} = 0.
	\end{cases}
	\label{eq:dist-mu-W}
\end{align}
In words, distribution $\mu_{W,i}(\rM_{>i} \mid M_{\leq i})$ is essentially the ``distribution'' $\mu^\ell(\rM_{>i} \mid \rM_{\leq i} = M_{\leq i}, W)$, namely, the continuation of messages starting from $M_{\leq i}$ conditioned on $B=(M_1,\ldots,M_\ell)$ satisfying
the event $W$, except the latter is not even a well-defined distribution: it is possible for some choices of $M_{\leq i}$ to never lead to the event $W$ happening. Thus, we are accounting for those events explicitly (as a corner case)
when defining $\mu_{W,i}$. 

We are ready to define the algorithm in the lemma. Given the stream $\sigma\sim\mu$, $A$ works as follows:

\begin{tbox}
\begin{enumerate}[label=$(\roman*)$,leftmargin=25pt]
	\item Sample $i \in S$ uniformly at random ($S \neq \emptyset$ otherwise~\Cref{lem:xor-1-copy} holds vacuously). 
	\item $P_0$ of $A$ samples $\sigma_{<i} \sim \mu^\ell(\rX_{<i} \mid W)$. This way, $P_0$ can run $\pi$ on $\sigma_{<i}$ and obtain the content of the blackboard $M_{<i}$ by $\pi$.
	\item This allows $P_0$ to start running $\alg_i$ on $\sigma_i=\sigma$ and $P_0,\ldots,P_n$ can collectively run the single pass of $\alg_i$ on $\sigma_i$; at the end, $P_0$ knows the message $M_i$ of $\pi$ (in addition to $M_{<i}$).
	\item $P_0$ now samples $M_{>i} \sim \mu_{W,i}(\rM_{>i} \mid M_{\leq i})$ defined in~\Cref{eq:dist-mu-W}, which is possible since all of $M_{\leq i}$ is known to $P_0$ at this point\footnote{This is the step we need 
	the second part in the definition of $\mu_{W,i}$ because even though $M_{<i}$ is necessarily consistent with $W$ (by the sampling), $M_{\leq i}$ is not (necessarily) given $X_i$ is sampled independently of $W$.}. $P_0$ has the blackboard content $B=(M_1,\ldots,M_\ell)$ now. 
	\item $P_0$ computes 
	\[
		g_i(B):=\underset{\theta\in\{0,1\}}{\arg\max}\; \Pr_{\mu^\ell}\paren{f(\sigma_i)=\theta\mid \rB=B}.
	\]
	(breaking the ties arbitrarily) and outputs the same answer for $A$ on $\sigma$. 
\end{enumerate}
\end{tbox}

It is easy to see that $A$ is a single-pass $s$-space algorithm for $f(\sigma)$ on $\mu$.
Thus, it remains to analyze its probability of success. Fix any $i\in S$ and consider the
following two distributions: 
\begin{itemize}
  \item Distribution of random variables induced by $A$ for any fixed choice of $i \in S$ denoted by
  $\dist_{A,i}$:
  \[
    \dist_{A,i} := \mu^\ell(\rX_{<i} \mid W) \times \mu(\rX_i) \times \mu_{W,i}(\rM_{> i} \mid M_{\leq i});
  \]
   note that $\rX_{<i}$ uniquely determines $\rM_{<i}$, and $(\rM_{<i},\rX_i)$ uniquely determines $\rM_{\leq i}$, and thus all distributions above are well defined. 
  \item The ``ideal'' distribution of these random variables denoted by $\dist^*_{i}$:
  \[
    \dist^*_i:= \mu^\ell(\rX_{<i},\rX_i,\rM_{>i} \mid W).
  \]
\end{itemize}

For $i \in S$, define $\event_i$ as the event that 
\[
	\event_i := \mathds{1}(\bias_{\pi}(i,B) > 2\sqrt{\delta}),
\]
namely, all blackboard contents $B$ that lead to the event $\event_\pi(i,B)$. Notice that 
\begin{align}
	\Pr_{B \sim \dist^*_i}(\event_i) = 1, \label{eq:1-pass-ideal}
\end{align}
by the definition of $\dist^*_i$ having a conditioning on $W$, since $W$ implies $\event_\pi(i,B)$ for all $i \in S$. The following claim relates the probability of success of $A$ to that of $\event_i$ but under the distribution $\dist_{A,i}$. 

\begin{claim}[``success probability of $A$'']\label{clm:xor-1-success}
\[	
\Pr_{\mu}\paren{\text{$A$ is correct on $\sigma$}} \geq \dfrac12 + {\sqrt{\delta}} \cdot \Exp_{i \in S}\bracket{\Pr_{\dist_{A,i}}(\event_i)}.
\] 
\end{claim}
\begin{proof}
	Fix any $i \in S$ and any transcript $B$ output by $A$ in the actual distribution $\dist_{A,i}$. Firstly notice that by the definition of $\mu_{W,i}$ (used in algorithm $A$), 
	any blackboard content $B$ generated by $A$ is also in the support of the original protocol $\pi$ in $\mu^{\ell}$. Thus, conditioning on $\rB=B$ is well defined for both
	$\dist_{A,i}$ and $\mu^{\ell}$ for any transcript $B$ fixed this way. For any $\sigma$, we claim that 
	\begin{align}
		\Pr_{\dist_{A,i}}\paren{\rX_i = \sigma \mid \rB=B} = \Pr_{\mu^\ell}\paren{\rX_i = \sigma \mid \rB=B}; \label{eq:xor-1-posterior} 
	\end{align}
	we emphasize that in this equation, we are changing the distribution $\dist_{A,i}$ to $\mu^{\ell}$ (and not the ``ideal'' distribution $\dist^*_i$). 
	We will prove this shortly, but assuming it for now, we have that for any transcript $B$ output by $A$ and $\theta \in \set{0,1}$,  
	\[
		\Pr_{\dist_{A,i}}\paren{f(\sigma_i) = \theta \mid \rB=B} = \Pr_{\mu^\ell}\paren{f(\sigma_i)=\theta\mid \rB=B},
	\]
	which implies that 
	\[
		\Pr_{\mu}\paren{\text{$A$ is correct on $\sigma$} \mid i} = \frac12 + \frac12 \cdot \bias_{\pi}(i,B). 
	\]
	As such, 
	\[
		\Pr_{\mu}\paren{\text{$A$ is correct on $\sigma$}} = \frac12 + \Exp_{i \in S}\bracket{\bias_{\pi}(i,\rB)} \geq \frac12 + \frac12 \cdot \Exp_{i \in S}\bracket{\Pr_{\dist_{A,i}}(\event_i)} \cdot 2\sqrt{\delta},
	\]
	as $\bias_{\pi}(i,B)$ is non-negative and it is at least $2\sqrt{\delta}$ whenever $\event_i$ happens. This implies the claim immediately. 
	
	It thus remains to prove~\Cref{eq:xor-1-posterior}. Fix any $\sigma$ and any $B=(M_{<i},M_i,M_{>i})$ output by $A$ in $\dist_{A,i}$. 
	In the following, for simplicity of exposition, we write 
	\[
		\Pr_{\dist_{A,i}}(\cdot) = \dist_{A,i}(\cdot) \qquad \text{and} \qquad \Pr_{\mu^{\ell}}(\cdot) = \mu^{\ell}(\cdot),
	\]
	and similarly for other relevant distributions. We also may abbreviate $\rM_i = M_i$ with $M_i$ when clear from the context (and similarly for other variables). 
	We have, 
	\begin{align*}
	&\dist_{A,i}\paren{\rX_i = \sigma \mid \rB=B} \\
	&\hspace{10pt}= \frac{\dist_{A,i}\paren{B \mid \rX_i = \sigma} \cdot \dist_{A,i}(\rX_i=\sigma)}{\dist_{A,i}\paren{B}} \tag{by Bayes' rule} \\
	&\hspace{10pt}= \frac{\dist_{A,i}\paren{M_{<i} \mid \rX_i = \sigma} \cdot \dist_{A,i}\paren{M_{i} \mid M_{<i},\rX_i = \sigma} \cdot \dist_{A,i}\paren{M_{>i} \mid M_{\leq i}, \rX_i = \sigma} \cdot \dist_{A,i}(\rX_i=\sigma)}{\dist_{A,i}\paren{M_{<i}} \cdot \dist_{A,i}\paren{M_i \mid M_{<i}} \cdot \dist_{A,i}\paren{M_{>i}\mid M_{\leq i}}} \tag{by chain rule for both nominator and denominator}\\
	&\hspace{10pt}= \frac{\dist_{A,i}\paren{M_{<i}} \cdot \mathds{1}(\text{$M_{i}$ is produced by $(M_{<i},\sigma)$}) \cdot \mu_{W,i}(M_{>i} \mid M_{\leq i}) \cdot \mu(\rX_i=\sigma)}{\dist_{A,i}\paren{M_{<i}} \cdot \dist_{A,i}\paren{M_i \mid M_{<i}} \cdot \dist_{A,i}\paren{M_{>i}\mid M_{\leq i}}},
	\intertext{where the equalities hold term by term because: (1) in $\dist_{A,i}$, $\rX_i$ is independent of sampled messages $M_{<i}$, (2) $M_{i}$ is deterministically fixed by $(M_{<i},\sigma)$, (3) distribution of $M_{> i}$ is determined
	by $M_{\leq i}$ alone, and $(4)$ distribution of $\rX_i$ is $\mu$. Given $\dist_{A,i}\paren{M_{>i}\mid M_{\leq i}} = \mu_{W,i}(M_{>i} \mid M_{\leq i})$ also by definition, and by cancelling the same terms, we can continue to}
	&\hspace{10pt}= \frac{\mathds{1}(\text{$M_{i}$ is produced by $(M_{<i},\sigma)$}) \cdot \mu(\rX_i=\sigma)}{\dist_{A,i}\paren{M_i \mid M_{<i}}}.
	\end{align*}
	
	Applying the same chain of equations to $\mu^{\ell}$ (and replacing $\mu_{W,i}$ by the corresponding distribution in $\mu^{\ell}$ and using the rectangle property of protocols to drop conditioning on $\rX_i=\sigma$), we obtain 
	\[
		\mu^{\ell}\paren{\rX_i = \sigma \mid \rB=B} = \frac{\mathds{1}(\text{$M_{i}$ is produced by $(M_{<i},\sigma)$}) \cdot \mu(\rX_i=\sigma)}{\mu^{\ell}\paren{M_i \mid M_{<i}}}.
	\]
	Finally, we claim that $\mu^{\ell}(M_i \mid M_{<i}) = \dist_{A,i}(M_i \mid M_{<i})$ which will conclude the proof. This equality holds because in both distributions $\sigma$ is chosen the same distribution $\mu$ independent of $M_{<i}$ and the event $\rM_i = M_i$ happens 
	whenever $(M_{<i},\sigma)$ generate $M_i$ which is the same process in both cases. This finalizes the proof. 
\end{proof}

Given~\Cref{clm:xor-1-success} and~\Cref{eq:1-pass-ideal}, we only need to bound the distance between the distributions $\dist_{A,i}$ and $\dist^*_i$ (for a random $i \in S$) and then apply~\Cref{fact:tvd-small} to conclude the proof. 
\begin{lemma}[``Distributions $\dist_{A,i}$ and $\dist^*_i$ are not too far'']\label{lem:xor-1-tvd}
	\[
		\Exp_{i \in S}\tvd{\dist^*_i}{\dist_{A,i}} < 1-{\sqrt{\delta}}.
	\]
\end{lemma}
\begin{proof}
	We will bound the KL-divergence between the two distributions and then apply \Cref{prop:pinsker} to conclude the proof. 
	Throughout the proof, all random variables $\rX,\rM$, etc. are distributed according to $\mu^{\ell}$ unless their distributions are stated explicitly otherwise. 
	
	For any fixed $i\in S$, using the chain rule of KL-divergence
(\Cref{fact:kl-chain-rule}), we have that
\begin{align}
	\kl{\dist^*_i}{\dist_{A,i}} &= \kl{\rX_{<i} \mid W}{\rX_{<i} \mid W} + \Exp_{\sigma_{<i} \sim \rX_{<i}\mid W} \kl{\rX_i \mid \rX_{<i} = \sigma_{<i}, W}{\rX_i} \notag \\
	&\hspace{30pt} + \Exp_{\sigma_{\leq i} \sim \rX_{\leq i} \mid W} \kl{\rM_{> i} \mid \rX_{\leq i} = \sigma_{\leq i}, W}{\mu_{W,i}(\rM_{> i} \mid M_{\leq i})}, \label{eq:xor-1-plugin}
\end{align}
where $M_{\leq i}$ is well defined in the last term as it is a deterministic function of $\sigma_{\leq i}$ (messages are specified uniquely by the inputs).  

The first term is clearly zero as the distributions are identical. In the following two claims, we will bound the remaining two
terms in expectation over $i\in S$.

\begin{claim}[``Distribution of $\sigma_i$ does not change by much'']\label{clm:xor-1-dist-sigma_i}
	\[
		\Exp_{i \in S} \Exp_{\sigma_{<i} \sim \rX_{<i}\mid W} \kl{\rX_i \mid \rX_{<i} = \sigma_{<i}, W}{\rX_i} < \log\paren{\frac1{2\sqrt{\delta}}}. 
	\]
\end{claim}
\begin{proof}
	We have
	\begin{align*}
		\Exp_{i \in S} \Exp_{\sigma_{<i} \sim \rX_{<i}\mid W} \kl{\rX_i \mid \rX_{<i} = \sigma_{<i}, W}{\rX_i} 
		&= \frac1{\card{S}} \cdot \sum_{i \in S} \Exp_{\sigma_{<i} \sim \rX_{<i}\mid W} \kl{\rX_i \mid \rX_{<i} = \sigma_{<i}, W}{\rX_i} \\
		&\leq \frac1{\card{S}} \cdot \sum_{i=1}^{\ell} \Exp_{\sigma_{<i} \sim \rX_{<i}\mid W} \kl{\rX_i \mid \rX_{<i} = \sigma_{<i}, W}{\rX_i}
		\tag{by including $[\ell] \setminus S$ indices and non-negativity of KL-divergence} \\
		&= \frac1{\card{S}} \cdot  \kl{\rX_1,\ldots,\rX_{\ell} \mid W}{\rX_1,\ldots,\rX_{\ell}} \tag{by the chain rule of KL divergence (\Cref{fact:kl-chain-rule})} \\
		&=  \frac1{\card{S}} \cdot \log{\paren{\frac{1}{\Pr(W)}}} \tag{by~\Cref{fact:kl-event} as $\rX_1,\ldots,\rX_{\ell}$ fixes $B$ and in turn the event $W$} \\
		&< \log{\paren{\frac{1}{2\sqrt{\delta}}}} \tag{given the bound on probability of $W$ in the lemma statement},
	\end{align*}
	concluding the proof.
	\Qed{clm:xor-1-dist-sigma_i}
	
\end{proof}

\begin{claim}[``Distribution of $M_{> i}$ does not change at all'']\label{clm:xor-1-dist-M_i}
	\[
		\Exp_{\sigma_{\leq i} \sim \rX_{\leq i} \mid W} \kl{\rM_{> i} \mid \rX_{\leq i} = \sigma_{\leq i}, W}{\mu_{W,i}(\rM_{> i} \mid M_{\leq i})}=0. 
	\]
\end{claim}
\begin{proof}
	We show that both distributions are identical for any choice of $\sigma_{\leq i}$ \emph{in the support} of $\rX_{\leq i} \mid W$ (this is \emph{not} necessarily true for all $\sigma_{\leq i}$). 
	For a fixed choice of $\sigma_{\leq i}$, there will be a unique choice of message 
	$M_{\leq i}$ as well (given that the protocol is deterministic). For any choice of $M_{> i}$, we have, 
	\begin{align*}
		&\Pr\paren{\rM_{> i}=M_{> i} \mid \rX_{\leq i} = \sigma_{\leq i}, W} \\
		&\hspace{30pt}= \Pr\paren{\rM_{> i}=M_{> i} \mid \rX_{\leq i} = \sigma_{\leq i}, \rM_{\leq i} = M_{\leq i},  W} \tag{given $\sigma_{\leq i}$ uniquely determines $M_{\leq i}$} \\
		&\hspace{30pt}=\frac{\Pr\paren{\rM_{> i}=M_{> i} \mid \rX_{\leq i} = \sigma_{\leq i}, \rM_{\leq i} = M_{\leq i}} \cdot \Pr\paren{W \mid \rX_{\leq i} = \sigma_{\leq i}, \rM_{\leq i} = M_{\leq i},\rM_{\geq i}=M_{\geq i}}}
		{\Pr\paren{W \mid \rX_{\leq i} = \sigma_{\leq i},\rM_{\leq i}=M_{\leq i}}} \tag{by Bayes' rule} \\
		&\hspace{30pt}=\frac{\Pr\paren{\rM_{> i}=M_{> i} \mid \rM_{\leq i} = M_{\leq i}} \cdot \mathds{1}(\text{$(M_{\leq i},M_{> i})$ is consistent with $W$})}{\Pr\paren{W \mid \rM_{\leq i}=M_{\leq i}}}\intertext{
	for the nominator: in the protocol $\prot$, $\rM_{> i} \perp \rX_{\leq i} \mid \rM_{\leq i}$ by rectangle property and because $W$ is uniquely determined by $B=M_{\leq i},M_{> i}$; for the denominator: again, by the rectangle property, 
	$\rM_{\geq i} \perp \rX_{\leq i} \mid \rM_{\leq i}$ and $W$ is only a function of $\rM_{\leq i}$ and $\rM_{> i}$ and thus $W \perp \rX_{\leq i} \mid \rM_{\leq i}$. Continuing,} 
	&\hspace{30pt}=\Pr\paren{\rM_{> i}=M_{> i} \mid \rM_{\leq i} = M_{\leq i},W}\tag{by definition of conditional probability},
	\end{align*}
	which is exactly the same probability assigned to $M_{>i}$ in $\mu_{W,i}(\rM_{>i} \mid M_{\leq i})$ given that for any choices of $M_{\leq i}$ originating from $\sigma_{\leq i}$'s in the expectation, the event $W$ can happen and thus we
	are in the first part of the definition of $\mu_{W,i}$. This concludes the proof. 
	\Qed{clm:xor-1-dist-M_i}
	
\end{proof}
	Plugging all prior bounds in~\Cref{eq:xor-1-plugin}, we have
	\[
		\Exp_{i \in S} \kl{\dist^*_i}{\dist_{A,i}} < \log{\paren{\frac{1}{2\sqrt{\delta}}}}, 
	\]
	which, by~\Cref{prop:pinsker} and Jensen's inequality implies,
	\[
		\Exp_{i \in S}\tvd{\dist^*_i}{\dist_{A,i}} \leq 1-\frac12 \cdot \Exp_{i \in S}\bracket{\exp\paren{-\kl{\dist^*_i}{\dist_{A,i}}}} \leq 1-\frac12 \cdot \bracket{\exp\paren{-\Exp_{i \in S}\kl{\dist^*_i}{\dist_{A,i}}}} \leq 1-\sqrt{\delta},
	\]
	finalizing the proof of~\Cref{lem:xor-1-tvd}.  
\end{proof}

Finally, combining~\Cref{clm:xor-1-success} and~\Cref{lem:xor-1-tvd}, together with~\Cref{eq:1-pass-ideal} and~\Cref{fact:tvd-small}, we get, 
\begin{align*}
	\Pr_{\mu}\paren{\text{$A$ is correct on $\sigma$}} &\geq \dfrac12 + {\sqrt{\delta}} \cdot \Exp_{i \in S}\bracket{\Pr_{\dist_{A,i}}(\event_i)} \tag{by~\Cref{clm:xor-1-success}} \\
	&\geq \dfrac12 + {\sqrt{\delta}} \cdot \Exp_{i \in S}\paren{1-\tvd{\dist^*_i}{\dist_{A,i}}} \tag{by~\Cref{eq:1-pass-ideal} and~\Cref{fact:tvd-small}} \\
	&> \dfrac12 + {\sqrt{\delta}} \cdot \paren{{\sqrt{\delta}}} \tag{by~\Cref{lem:xor-1-tvd}} \\
	&= \frac12 + \delta.
\end{align*}
This concludes the proof of~\Cref{lem:xor-1-copy} for single-pass algorithms.

\subsection{Proof of~\Cref{lem:xor-1-copy}}

We now modify the previous proof to establish~\Cref{lem:xor-1-copy} in its full generality. 

\begin{lemma*}[Restatement of \Cref{lem:xor-1-copy}]
	Consider any $S \subseteq [\ell]$ and suppose that 
	\[
	\Pr\paren{\event_\pi(S,B)} > (2\sqrt{\delta})^{\card{S}}.
	\]
	Then, there is a $p$-pass $s$-space streaming algorithm for computing $f(\sigma)$ on $\mu$ 
	with probability of success more than $1/2 + \delta$. 
\end{lemma*}

The strategy is exactly as before by turning $\pi$ into a $s$-space $p$-pass streaming algorithm $A$ for computing
$f(\sigma_i)$ on the stream $\sigma_i\sim\mu$ (and not the entire input). We will use the same notation as in the previous subsection. In particular, let $W:=\event_{\pi}(S,B)$. 

As in the proof of~\Cref{lem:xor-independence}, flatten the messages into $M_1,\ldots,M_d$ for $d = p \cdot \ell$ in the order they are being communicated. 
For any $i \in [d]$ and messages $M_{< i} = (M_1,\ldots,M_{i-1})$, define the following distribution for the \emph{next} immediate message $M_{i}$:  
\begin{align}
	\mu_{W,i}(\rM_{i} \mid M_{< i}) := 
	\begin{cases} 
		\mu^\ell(\rM_{i} \mid \rM_{< i} = M_{< i} , W) \quad &\Pr_{\mu^\ell}\paren{W \mid \rM_{< i} = M_{< i}} > 0 \\
		\mu^\ell(\rM_{i} \mid \rM_{< i} = M_{< i}) \quad &\Pr_{\mu^\ell}\paren{W \mid \rM_{< i} = M_{< i}} = 0.
	\end{cases}
	\label{eq:p-dist-mu-W}
\end{align}
This is effectively the same as~\Cref{eq:dist-mu-W} for single-pass algorithms, although now applied one message at a time, and aims to capture the ``distribution'' $\mu^\ell(\rM_{i} \mid \rM_{< i} = M_{< i}, W)$ (which, as before, is not an actual distribution, and thus requires some handling of corner cases). 

We now present the main algorithm $A$ for processing the $\sigma\sim\mu$: 

\begin{tbox}
\begin{enumerate}[label=$(\roman*)$,leftmargin=25pt]
	\item Sample $i \in S$ uniformly at random ($S \neq \emptyset$ otherwise~\Cref{lem:xor-1-copy} holds vacuously). 
	\item Player $P_0$ of $A$ generates messages $M_1,\ldots,M_d$ in order as follows for any $j \in [d]$: 
	\begin{enumerate}
		\item If $M_j$ is a function of the $\sigma_i$-part of the input (i.e., is the message of $i$-th player of $\prot$), players $P_0,\ldots,P_n$ run the streaming algorithm of $\prot$ on their input $\sigma_i=\sigma$ (with $P_0$ knowing the prior messages $M_{<j}$ on the board)
		to generate the message $M_j$ which $P_0$ now knows. 
		\item Otherwise, $P_0$ simply samples $M_j \sim \mu_{W,j}(\rM_j \mid M_{<j})$ using~\Cref{eq:p-dist-mu-W}. 
	\end{enumerate}
	\item At the end, $P_0$ has the blackboard content $B=(M_1,\ldots,M_d)$ and computes	
	\[
		g_i(B):=\underset{\theta\in\{0,1\}}{\arg\max}\; \Pr_{\mu^\ell}\paren{f(\sigma_i)=\theta\mid \rB=B}.
	\]
	(breaking the ties arbitrarily) and outputs the same answer for $A$ on $\sigma$. 
\end{enumerate}
\end{tbox}

It is easy to see that in $A$ only scans the real stream $\sigma=\sigma_i$ in $p$ passes using an $s$-space algorithm. 
Thus, $A$ is a $p$-pass $s$-space algorithm for $f(\sigma)$ sampled from $\mu$. It remains to analyze its probability of success. 

For any $i \in S$, define the following two distributions: 
\begin{itemize}
	\item $\dist_{A,i}$: the joint distribution of $(\rX_i,\rB)$ generated by $A$ on an input $\rX_i \sim \mu$ (namely, sampling $\sigma$ and embedding it in $\sigma_i$ of $\prot$). 
	\item $\dist^*_i$: the joint distribution of $(\rX_i,\rB)$ sampled from $\mu^\ell \mid W$. 
\end{itemize}
Moreover, for $i \in S$, define $\event_i$ as the event that 
\[
	\event_i := \mathds{1}(\bias_{\pi}(i,B) > 2\sqrt{\delta}),
\]
namely, all blackboard contents $B$ that lead to the event $\event_\pi(i,B)$. As before, we have
\begin{align}
	\Pr_{B \sim \dist^*_i}(\event_i) = 1, \label{eq:p-pass-ideal}
\end{align}
by the definition of $\dist^*_i$ having a conditioning on $W$, since $W$ implies $\event_\pi(i,B)$ for all $i \in S$. The following claim (analogue of~\Cref{clm:xor-1-success}) 
relates the probability of success of $A$ to that of $\event_i$ but under the distribution $\dist_{A,i}$. The proof is identical to~\Cref{clm:xor-1-success} and we only sketch the differences. 

\begin{claim}[``success probability of $A$'']\label{clm:xor-p-success}
\[	
\Pr_{\mu}\paren{\text{$A$ is correct on $\sigma$}} \geq \dfrac12 + {\sqrt{\delta}} \cdot \Exp_{i \in S}\bracket{\Pr_{\dist_{A,i}}(\event_i)}.
\] 
\end{claim}
\begin{proof}
	Fix any $i \in S$ and any transcript $B$ output by $A$ in the actual distribution $\dist_{A,i}$. Exactly as before, 
	we claim that for any $\sigma$, 
	\[
		\Pr_{\dist_{A,i}}\paren{\rX_i = \sigma \mid \rB=B} = \Pr_{\mu^\ell}\paren{\rX_i = \sigma \mid \rB=B}; 
	\]
	Concluding the proof from this equation is exactly as in~\Cref{clm:xor-1-success}. It thus only remains to prove this equation. 
	In the following, for an integer $j \in [d]$, we write $P(j) \in [\ell]$ to denote the player in $\prot$ that sends the message $M_j$ in the flattened list of messages $M_1,\ldots,M_d$. 
	For any fixed $B$, define: 
	\[
		\mathcal{X}_i(B) := \set{\sigma_i \in \Sigma_n : \text{$M_j$ is the message of $(\sigma_i,M_{<j})$ for every $j$ with $P(j)=i$}};
	\]
	in other words, $\mathcal{X}_i(B)$ is the set of all inputs to the $i$-th player of $\pi$ that are consistent with the transcript $B$. 
	Exactly as in~\Cref{clm:xor-1-success}, we have, 
	\begin{align}
		\dist_{A,i}(\sigma,B) = \mu(\sigma) \cdot \mathds{1}\paren{\sigma \in \mathcal{X}_i(B)} \cdot \prod_{j : P(j) \neq i} \mu_{W,j}(M_j \mid M_{<j}), \label{eq:joint-dist_Ai}
	\end{align}
	and by the rectangle property
	\begin{align}
		\mu^{\ell}(\sigma,B) = \mu(\sigma) \cdot \mathds{1}\paren{\sigma \in \mathcal{X}_i(B)} \cdot \prod_{j : P(j) \neq i} \mu^{\ell}(M_j \mid M_{<j}). \label{eq:joint-mu^ell}
	\end{align}
	Finally, we have, 
	\begin{align*}
		\dist_{A,i}(\rX_i = \sigma \mid \rB=B) &= \frac{\dist_{A,i}(\rX_i = \sigma , \rB=B)}{\dist_{A,i}(\rB=B)} \\
		&=\frac{ \mu(\sigma) \cdot \mathds{1}\paren{\sigma \in \mathcal{X}_i(B)}}{\sum_{\sigma'}\mu(\sigma') \cdot \mathds{1}\paren{\sigma' \in \mathcal{X}_i(B)})} \tag{using~\Cref{eq:joint-dist_Ai} as the third terms have no dependence on $\sigma$} \\
		&= \frac{\mu^{\ell}(\rX_i = \sigma , \rB=B)}{\mu^{\ell}(\rB=B)} = \mu^{\ell}(\rX_i = \sigma \mid \rB=B), \tag{now using~\Cref{eq:joint-mu^ell} the same exact way} 
	\end{align*}
	concluding the proof. 
\end{proof}

We now bound the distance between the distributions $\dist^*_i$ and $\dist_{A,i}$ (analogue of~\Cref{lem:xor-1-tvd}). 

\begin{lemma}[``Distributions $\dist_{A,i}$ and $\dist^*_i$ are not too far'']\label{lem:xor-p-tvd}
	\[
		\Exp_{i \in S}\tvd{\dist^*_i}{\dist_{A,i}} < 1-{\sqrt{\delta}}.
	\]
\end{lemma}
\begin{proof}
	We will bound the KL-divergence between the two distributions and apply \Cref{prop:pinsker} to conclude the proof. 
		
	For any fixed $i \in S$: 
	\[
		\kl{\dist^*_i}{\dist_{A,i}} = \Exp_{(\sigma_i,B) \sim \dist^*_i(\rX_i,\rB)}\bracket{\log{\frac{\dist^*_i(\sigma_i,B)}{\dist_{A,i}(\sigma_i,B)}}} 
	\]
	by the definition of KL-divergence and since $(\sigma_i,B)$ in support of $\dist^*_i$ has non-zero probability in $\dist_{A,i}$ (any transcript generated in $\mu^{\ell} \mid W$ is also accounted for in $\mu_{W,i}$ in $\dist_{A,i}$). 
	For any $(\sigma_i,B)$ (each single term in the KL-divergence above), 
	\begin{align*}
		\log{\frac{\dist^*_i(\sigma_i,B)}{\dist_{A,i}(\sigma_i,B)}} &= \log{\frac{\mu^{\ell}(\sigma_i,B \mid W)}{\dist_{A,i}(\sigma_i,B)}} \tag{by the definition of $\dist^*_i$} \\
		&= \log{\frac{\mu^{\ell}(\sigma_i,B)}{\mu^{\ell}(W) \cdot \dist_{A,i}(\sigma_i,B)}} \tag{by conditional probability and since any $B$ considered belongs to the event $W$} \\
		&= \log\frac1{\mu^{\ell}(W)} + \log\frac{\prod_{j : P(j) \neq i} \mu^{\ell}(M_j \mid M_{<j})}{\prod_{j : P(j) \neq i} \mu_{W,j}(M_j \mid M_{<j})} \tag{by~\Cref{eq:joint-dist_Ai} and~\Cref{eq:joint-mu^ell}} \\
		&= \log\frac1{\mu^{\ell}(W)} + \sum_{j : P(j) \neq i} \log\frac{\mu^{\ell}(M_j \mid M_{<j})}{\mu^{\ell}(M_j \mid M_{<j},W)} \tag{by taking the products out and since all $B$ here belong to term one of definition of $\mu_{W,i}$}.
	\end{align*}
	Define, for any $j \in [d]$,  
	\[
		\kappa_j := \Exp_{M_j \sim \dist^*_i(\rM_{<j})}\kl{\mu^{\ell}(\rM_j \mid M_{<j},W)}{\mu^{\ell}(\rM_j \mid M_{<j})}, 
	\]
	and note that 
	\[
		\sum_{j=1}^{d} \kappa_j = \Exp_{M_j \sim \dist^*_i(\rM_{<j})}\kl{\mu^{\ell}(\rM_j \mid M_{<j},W)}{\mu^{\ell}(\rM_j \mid M_{<j})} = \kl{\mu^{\ell}(\rB \mid W)}{\mu^{\ell}(\rB)} = \log\paren{\frac{1}{\mu^{\ell}(W)}},
	\]
	by the chain rule of KL-divergence (\Cref{fact:kl-chain-rule}) for the first equality and~\Cref{fact:kl-event} for the second (since $\rB$ uniquely determines if the event $W$ happens or not). 
	Going back to the first equation for KL-divergence, and using our earlier equations here, we get 
	\begin{align*}
		\kl{\dist^*_i}{\dist_{A,i}} &= \log\frac1{\mu^{\ell}(W)} - \Exp_{(\sigma_i,B) \sim \dist^*_i(\rX_i,\rB)}\bracket{\sum_{j : P(j) \neq i} \log\frac{\mu^{\ell}(M_j \mid M_{<j} , W)}{\mu^{\ell}(M_j \mid M_{<j})}} \\
		&=  \log\frac1{\mu^{\ell}(W)} - \sum_{j : P(j) \neq i} \kappa_j \tag{by the definition of $\kappa_j$}  \\
		&=  \log\frac1{\mu^{\ell}(W)} - \sum_{j =1}^{d} \kappa_j  + \sum_{j : P(j) = i} \kappa_j \tag{by adding and subtracting the remaining terms} \\
		&= \sum_{j : P(j) = i} \kappa_j \tag{by the previous equation on the sum of $\kappa_j$'s}. 
	\end{align*}
	Finally, we can conclude the lemma because
	\begin{align*}
		\Exp_{i \in S}\kl{\dist^*_i}{\dist_{A,i}} &= \frac{1}{\card{S}} \cdot \sum_{i \in S} \sum_{j : P(j) = i} \kappa_j \leq \frac{1}{\card{S}} \cdot \sum_{j=1}^{d} \kappa_j = \frac{1}{\card{S}} \cdot \log\paren{\frac{1}{\mu^{\ell}(W)}} < \log\paren{\frac{1}{2\sqrt{\delta}}},
	\end{align*}
	where the first inequality is because $\set{j:P(j)=i}$-sets are disjoint for $i \in S$ and are subsets of $[d]$ (and $\kappa_j$'s are non-negative KL-divergence terms), and the last inequality is by the assumption in the lemma 
	statement on the probability of the event $W$. 
	
	Plugging this bound into~\Cref{prop:pinsker} and Jensen's inequality implies,
	\[
		\Exp_{i \in S}\tvd{\dist^*_i}{\dist_{A,i}} \leq 1-\frac12 \cdot \Exp_{i \in S}\bracket{\exp\paren{-\kl{\dist^*_i}{\dist_{A,i}}}} \leq 1-\frac12 \cdot \bracket{\exp\paren{-\Exp_{i \in S}\kl{\dist^*_i}{\dist_{A,i}}}} < 1-\sqrt{\delta},
	\]
	finalizing the proof of~\Cref{lem:xor-p-tvd}.  
\end{proof}
The rest of the proof is exactly as in the single-pass algorithms case, thus concluding the entire proof of~\Cref{lem:xor-1-copy}. 

%% file: figs/game.tex

\begin{tikzpicture}

\tikzset{Qp/.style={rectangle, draw, minimum width=100pt, minimum height=40pt, inner sep=5pt]}}

\tikzset{Pp/.style={rectangle, draw, minimum width=25pt, minimum height=15pt, inner sep=0pt]}}

\node[Pp] (P10){$P_0$}; 
\node[Pp] (P11)[right=10pt of P10]{$P_1$}; 
\node[Pp] (P12)[right=10pt of P11]{$P_2$}; 
\node[Pp] (P13)[right=10pt of P12]{$P_3$}; 

\node[Qp] (Q1) [fit=(P10) (P13)]{};
\node[above=5pt of Q1] {\text{player $Q_1$ with  stream $\sigma_1$}};

\node[Pp] (P20)[right=30pt of Q1]{$P_0$}; 
\node[Pp] (P21)[right=10pt of P20]{$P_1$}; 
\node[Pp] (P22)[right=10pt of P21]{$P_2$}; 
\node[Pp] (P23)[right=10pt of P22]{$P_3$}; 	

\node[Qp] (Q2) [fit=(P20) (P23)]{};
\node[above=5pt of Q2] {\text{player $Q_2$ with  stream $\sigma_2$}};

\node[Pp] (P30)[right=30pt of Q2]{$P_0$}; 
\node[Pp] (P31)[right=10pt of P30]{$P_1$}; 
\node[Pp] (P32)[right=10pt of P31]{$P_2$}; 
\node[Pp] (P33)[right=10pt of P32]{$P_3$}; 	

\node[Qp] (Q3) [fit=(P30) (P33)]{};
\node[above=5pt of Q3] {\text{player $Q_3$ with  stream $\sigma_3$}};

\node[rounded rectangle, draw, fill=gray!10, minimum width=200pt, minimum height=40pt] (bb) [below=15pt of Q2] {Black Board};

\foreach \i in {1,2,3}
{
	\foreach \j in {0,1,2}
	{
        \pgfmathtruncatemacro{\jp}{\j+1};
        \draw[line width=1pt, red, ->]
		(P\i\j) to (P\i\jp);
	}
	\draw[line width=1pt, red, ->]
	(P\i3) to ($(P\i3)+(0pt,15pt)$) to ($(P\i0)+(0pt,15pt)$) to (P\i0);
	
}

	\draw[line width=2pt, blue, dashed]
	(P10) to ($(P10) + (0pt,-55pt)$) to (bb);

	\draw[line width=2pt, blue, dashed]
	(P20) to ($(P20) + (0pt,-35pt)$);
	
	\draw[line width=2pt, blue, dashed]
	(P30) to ($(P30) + (0pt,-55pt)$) to (bb);
\end{tikzpicture}

%% file: 4cycle.tex
\newcommand{\rProt}{\rv{\Pi}}
\newcommand{\rP}{\rv{P}}
\newcommand{\unif}{\mathcal{U}}

\newcommand{\width}{w}
\newcommand{\depth}{d}

\newcommand{\pcw}{m}
\newcommand{\pcd}{b}

\newcommand{\istar}{\ensuremath{i^*}}

\newcommand{\hM}{\hat{M}}

\newcommand{\hpcw}{\hat{\pcw}}
\newcommand{\hpcd}{\hat{\pcd}}
\newcommand{\hG}{\hat{G}}

\section{The Lower Bound for the Noisy Gap Cycle Counting Problem}\label{sec:cycle}

We prove our main streaming lower bound for the noisy gap cycle counting problem in this section, defined formally as follows (see also~\Cref{fig:k-2k-cycle} in the Introduction for an illustration). 

\begin{definition}[\textbf{Noisy Gap Cycle Counting Problem (NGC)}]\label{def:NGCdef} Let $k,t \in \IN^+$ and $n=6 t \cdot k$. In $\NGC_{n,k}$, we have an $n$-vertex graph $G$ 
 with the promise that $G$ either contains $(i)$  $2t$ vertex-disjoint $k$-cycles, or $(ii)$ $t$ vertex-disjoint $(2k)$-cycles; in both cases, the remaining vertices of $G$ are partitioned into 
$4t$ vertex-disjoint paths of length $k-1$ (the ``noise'' part of the graph). The goal is to distinguish between these two cases. 
\end{definition}

We prove the following lower bound for this problem that formalizes~\Cref{res:ngc}. 

\begin{theorem}\label{thm:ngc}
	For every  $k \in \IN^+$, any $p$-pass streaming algorithm for Noisy Gap Cycle Counting $\NGC_{n,k}$ with probability of success at least $2/3$ requires $\Omega\paren{\frac{1}{p^5} \cdot (n/k)^{1-O(p/k)}}$ space. 
\end{theorem}
Note that for this lower bound to be non-trivial, both $k$ needs to be at least some large constant, and $p$ should be smaller than $k$ by a similar factor.  
The rest of this section is organized as follows. We first design a hard input distribution for $\NGC_{n,k}$ and prove its useful properties for our purpose. We then give a high level plan of the lower bound for this distribution, followed
by the details of each step, and  the proof of the theorem. We start out with some extra definitions as we need some structure on the family of graphs we work with in order to describe our hard distribution.

\begin{definition}[\textbf{Layered Graph}]\label{def:lg}
	For any integers $\width,\depth \geq 1$, we define a \emph{${(\width,\depth)}$-layered graph}, with \emph{width} $\width$ and \emph{depth} $\depth$, as any graph $G=(V,E)$ with the following properties: 
	\begin{enumerate}[label=$(\roman*)$]
		\item Vertices $V$ consist of $\depth+1$ \emph{layers} of vertices $V_1,\ldots,V_{\depth+1}$, each of size $\width$.  
		\item Edges $E$ consist of $\depth$ \emph{matchings} $M_1,\ldots,M_{\depth}$ where $M_i$ is a perfect matching between $M_i,\,M_{i+1}$. 
	\end{enumerate}
	For any vertex $v \in V_1$, we use $P(v)$ to denote the \emph{unique} vertex reachable from $v$ in $V_{\depth+1}$. 
	
	Moreover, by a \emph{random} layered graph, we mean a layered graph whose matchings are chosen \emph{uniformly at random} and \emph{independently} but the partitioning of vertices into the layers is fixed. 
\end{definition}

In our proof, we also work the \emph{pointer chasing} problem (although with several non-standard aspects). We define this problem as follows:

\begin{definition}[\textbf{Pointer Chasing (PC)}]\label{def:PC} Let $\pcw,\pcd \in \IN^+$. In $\PC_{\pcw,\pcd}$, we have a $(\pcw,\pcd)$-layered graph on layers $V_1,\ldots,V_{\pcd+1}$, an arbitrary vertex $s \in V_1$, and 
an arbitrary equipartition $X,Y$ of $V_{\pcd+1}$. The goal is to decide whether $P(s) \in V_{\pcd+1}$ belongs to $X$ (a $X$-instance) or $Y$  (a $Y$-instance). 
\end{definition}

\Cref{fig:pc} gives an illustration of this problem. 

\bigskip

\begin{figure}[h!]	
\centering
\subcaptionbox{\footnotesize $X$-instance of $\PC_{6,4}$}%
  [.45\linewidth]{\input{figs/pcX}}
  \hspace{1cm}
\subcaptionbox{\footnotesize $Y$-instance of $\PC_{6,4}$}%
  [.45\linewidth]{\input{figs/pcY}}
\caption{An illustration of $\PC_{\pcw,\pcd}$ for $\pcw=6$ and $\pcd=4$. The edges are perfect matchings that go between consecutive vertex layers. The start vertex $s$ is depicted in the vertex layer $V_1$, and the sets $X,Y$ are marked in the vertex layer $V_5$. $P(s)$, the unique vertex in $V_5$ reachable from $s$ is also shown.} 
\label{fig:pc}
\end{figure}
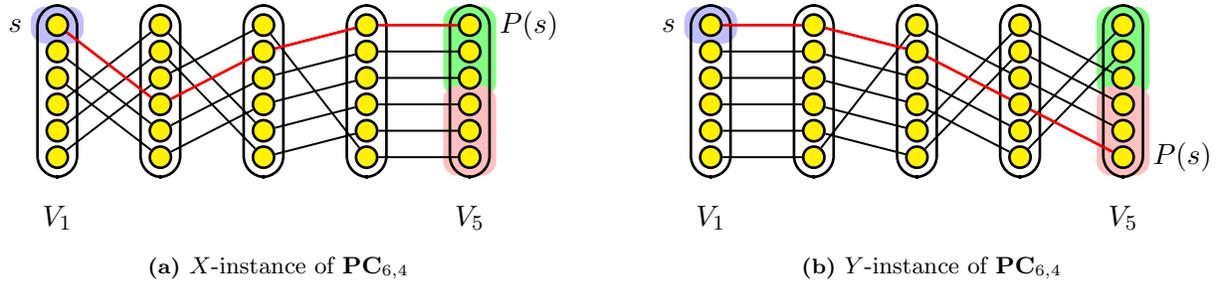

\subsection{A Hard Distribution for \NGC}\label{sec:ngc-dist}

We can now define our hard input distribution. We sample a random $(\width,\depth)$-layered graph $G_0$ for parameter $\width=3t$ and $\depth=\frac{k-2}{2}$ for $\NGC_{n,k}$ \emph{conditioned on the following event}: 
\begin{itemize}
\item Let $S \subseteq V_1$ be a fixed subset of size $t$, and $X,Y$ be a fixed equipartition of 
$V_{\depth+1}$ (say, both are the lexicographically-first option); then $\set{P(v) \mid v \in S}$ is either a subset of $X$ or $Y$. 
\end{itemize}
We construct the final graph $G$ from four \emph{identical} copies of $G_0$ (on disjoint sets of vertices), plus some fixed gadget so that it satisfies the following property: when all vertices $v \in S$ have $P(v) \in X$, the resulting graph $G$ has 
$2t$ cycles of length $k$ each; otherwise, it has $t$ cycles of length $2k$ instead; in both cases, the graph $G$ also has $4t$ paths of length $k-1$. 
We now present the formal description of the distribution (see~\Cref{fig:ngc} for an illustration).

\begin{Distribution}\label{distngc}
The distribution $\distNGC$ for $\NGC_{n,k}$ for given parameters $n,k$ and $t = n/6k$.

\begin{enumerate}[label=$(\roman*)$]
	\item Let $\depth :=\frac{(k-2)}{2}$ and sample a random $(3t,\depth)$-layered graph $G_0$ on vertices $V_1,\ldots,V_{\depth+1}$ and matchings $M_1,\ldots,M_\depth$ \emph{conditioned} on the following event:
	\begin{itemize}[leftmargin=10pt]
		\item Let $S$ be a fixed $t$-subset of $V_1$ and $X,Y$ be a fixed equipartition of $V_{\depth+1}$.  
		Then, $\set{P(v) \mid v \in S}$ is entirely a subset of $X$ (a \emph{$X$-instance}) or a subset of $Y$ (a \emph{$Y$-instance}).
	\end{itemize}
	\item Create the following graph $G=(V,E)$ on groups of vertices $V^{i}_j$ for $i \in [4]$ and $j \in [\depth+1]$ using four identical copies of the graph sampled $G_0$ above:
	\begin{enumerate}[leftmargin=10pt]
		\item For every $j \in [\depth+1]$, let $V^i_j$ be the copies of $V_j$ in $G_0$ and define $S^i, X^i,Y^i$ as copies of $S,X,Y$, respectively (the same for all $i \in [4]$) -- for any vertex $v \in G_0$ and $i \in [4]$, we use $copy(v,i)$ to denote the copy of $v$ in $V^i$.  
		\item\label{line:rand} Connect $V^i_j$ to $V^i_{j+1}$ for any $i,j$ by a matching $M^i_j$ corresponding to $M_j$ of $G_0$.
		\item\label{line:det} Connect $S^1$ to $S^3$, and $S^2$ to $S^4$ using identity perfect matchings, respectively. Similarly, connect $X^1$ to $X^3$ and $X^2$ to $X^4$, and $Y^1$ to $Y^4$ and $Y^2$ to $Y^3$ using identity perfect matchings,
		respectively (note the crucial change between the treatment of $X^i$ and $Y^i$). 
	\end{enumerate}
	\item The input stream consists of edges inserted in~\eqref{line:det} in some arbitrary order, followed by $M^i_j$ in \emph{decreasing} order of $j$ and increasing order of $i$ (the order inside each $M^i_j$ is arbitrary), i.e.,
	this part of the stream is $M^1_{\depth} \conc \cdots \conc M^1_1 \conc \cdots \conc M^4_{\depth} \conc \cdots \conc M^4_1$. 
\end{enumerate}
\end{Distribution}

  \begin{figure}[t!]	
\centering
\input{figs/ngc}
\caption{An illustration of~\Cref{distngc} for $t=2$ and $k=10$. The thin (black) edges are formed by the random input matchings while thick (red) edges are input-independent. The final graph is obtained via four identical copies of 
a $(6,4)$-layered graph and the sets $S,X,Y$ are marked in each one. The input stream consists of the edges these matchings ordered from the inner matchings to the outer ones. The input drawn shows a $k$-cycle instance.} 
\label{fig:ngc}
\end{figure}
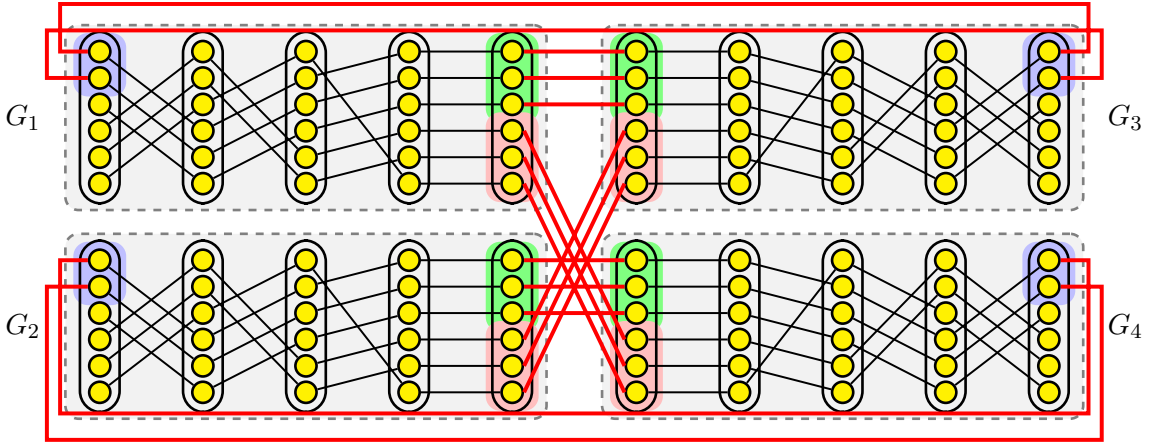

 A careful reader may have noticed that~\Cref{distngc} already impose some constraints on the choice of $k$, e.g., $k$ needs to be even and more constraints will also follow (we also need $t$ to be an even number so $V_{\depth+1}$ admits an equipartition). Yet~\Cref{thm:ngc} is supposed to work 
for all choices of $k$, in particular, we need odd values of $k$ for some of our reductions. We will however fix this later by some simple padding argument  when concluding the proof of~\Cref{thm:ngc}.

We start by proving that this distribution indeed outputs valid instances of $\NGC_{n,k}$. 

\begin{lemma}\label{lem:distngc-valid}
	Let $G$ be a graph sampled from~\Cref{distngc}. Then, 
	\begin{enumerate}[label=$(\roman*)$]
		\item if $G_0$ is a $X$-instance, then $G$ consists of $2t$ cycles of length $k$ and $4t$ paths of length $k-1$;
		\item if $G_0$ is a $Y$-instance, then $G$ consists of $t$ cycles of length $2k$ and $4t$ paths of length $k-1$.
	\end{enumerate}
\end{lemma}

\begin{proof}

Suppose first that $G_0$ is a $X$-instance. Fix a vertex $v \in S$ of $G_0$ and consider $copy(v,1) \in S^1$. 
We claim that $copy(v,1)$  belongs to the following $k$-cycle in $G$: 
\[
	copy(v,1) \in S^1 \leadsto copy(P(v),1) \in X^1 \rightarrow copy(P(v),3) \in X^3 \leadsto copy(v,3) \in S^3 \rightarrow copy(v,1) \in S^1;
\]
the first path exists by the definition of $P(v) \in X$ in $G_0$, the next edge by the identity perfect matching between $X^1$ and $X^3$, the next path by reversing the identical path $P(v) \in X$ to $v$ in $G_0$, and the final edge by the identity perfect matching
between $S^1$ and $S^3$. Note that this cycle in particular also contains $copy(v,3) \in S^3$ and its length is $\depth+1+\depth+1 = 2 \cdot \frac{(k-2)}{2} + 2 = k$.  

As a result, there are $\card{S}=t$ disjoint $k$-cycles, each including a pair of vertices from $S^1$ and $S^3$. By symmetry, this is also happening for another $t$ disjoint $k$-cycles for vertices of $S^2$ and $S^4$, which gives us $2t$ disjoint $k$-cycles as desired.
Moreover, as these $k$-cycles cover, for every $i \in [4]$, all vertices of $S^i$ and all other remaining vertices in $V^i$ have degree one, there is no other cycle; instead, vertices of $V^i \setminus S^i$ belong to a path of length $k-1$, giving 
us $4t$ disjoint $(k-1)$-paths. 

Now suppose that $G_0$ is a $Y$-instance.  Again, fix a vertex $v \in S$ of $G_0$ and consider $copy(v,1) \in S^1$. 
We now claim that $copy(v,1)$ belongs to the following $2k$-cycle in $G$: 
\begin{align*}
	copy(v,1) \in S^1 &\leadsto copy(P(v),1) \in Y^1 \rightarrow copy(P(v),4) \in Y^4 \leadsto copy(v,4) \in S^4 \rightarrow copy(v,2) \in S^2, \\
	copy(v,2) \in S^2 &\leadsto copy(P(v),2) \in Y^2 \rightarrow copy(P(v),3) \in Y^3 \leadsto copy(v,3) \in S^3 \rightarrow copy(v,1) \in S^1;
\end{align*}
this is exactly as in case $(i)$ except that the switch in connecting $Y^1$ to $Y^4$ and $Y^2$ to $Y^3$ instead, results in $copy(v,1)$ reaching $copy(v,2)$  and thus it takes another ``round'' for $copy(v,2)$ to reach $copy(v,1)$ also, 
doubling the length of the cycle. The rest of the argument is exactly as before and we omit the details (but, notice that this doubling effect will not happen for the $(k-1)$-paths as their endpoints have degree one only).
This concludes the proof. 
\end{proof}

In addition to proving the validity of distribution \Cref{distngc}, \Cref{lem:distngc-valid} also has a  message for our lower bound: any algorithm that can solve $\NGC_{n,k}$ over the distribution $\distNGC$, 
can also decided whether a given graph $G_0$ sampled in $\distNGC$ is a $X$-instance or a $Y$-instance (since the instance is defined deterministically given $G_0$, and the remaining edges are input-independent, we can just
do a straightforward reduction between the two problems). This latter problem is what we prove our lower bound for.

Before moving on, let us also make the following important remark about the source of hardness of $\NGC$ (over $\distNGC$) that will be used in several of our reductions in the subsequent section. 

\begin{remark}\label{rem:distngc-known}
	The following information is known  by any streaming algorithm that solves $\NGC$ over the distribution $\distNGC$: 
	\begin{enumerate}[label=$(\roman*)$]
		\item One endpoint of every noise path in the graph $G$;
		\item A set of $t$ four-tuples of vertices $(u_1,u_2,u_3,u_4)$ such that in the $k$-cycle case $u_1,u_2$ and $u_3,u_4$ belong to two disjoint $k$-cycles each, while in the $2k$-cycle case, all  belong to the same $2k$-cycle. 
	\end{enumerate}
	Notice that these vertices cover all paths and cycles of the graph in both cases. 
\end{remark}

The first point is because the set $S$ in $G_0$ is fixed and we can take $copy(v,i)$ for $v \notin S$ and $i \in \set{1,2}$ as endpoints of every path. The second point is similarly because 
we can take the tuples $(u_1,u_2,u_3,u_4)$ to be $u_1=copy(v,1),u_2=copy(v,3),u_3=copy(v,2),u_4=copy(v,4)$ and use the proof of~\Cref{lem:distngc-valid} to see that these vertices belong to desired cycles. 

\subsection{The High Level Plan} \label{sec:plan}

Our plan for proving~\Cref{thm:ngc} involves the following three steps (the reader is encouraged to check the algorithm for $\NGC$ in~\Cref{sec:alg-cycle} as it helps with the intuition of this proof).  

\paragraph{Step one: decorrelating the distribution.} Recall that by our previous discussion, our task at this point is to prove a lower bound for the following problem: 
Given a $(3t,\depth)$-layered graph $G_0$ and a set $S$ of $t$ vertices in the first layer, decide whether following edges of $G_0$ takes  these vertices to $X$ or $Y$ in the last layer. If we look at a single vertex $v \in S$, this problem 
is a pointer chasing problem along the edges of the $\depth$ matchings of $G_0$.  The challenge here is that we are not solving any one pointer chasing problem though, but rather a collection of $t$ \emph{correlated} ones. This problem
is quite simpler (algorithmically) than original pointer chasing as we only need to get ``lucky enough'' to chase one of them. Concretely, an algorithm that samples $\approx t^{1-1/\depth}$ edges of the graph has a constant probability of 
finding one complete path and solves the problem. 

To bypass this challenge, we consider a generalized version of the distribution $\distNGC$ wherein every vertex in $S$ has \emph{almost} the same probability of ending
up in the set $X$ or $Y$, independent of the choice of other starting vertices. We prove that even though these distributions do not correspond to valid instances of $\NGC$, still, if we run any algorithm for $\NGC$ over these inputs, it has to do
some ``non-trivial work'': informally speaking, it will be able to solve the pointer chasing instance corresponding to one of the starting vertices with a probability of $1/2+\Omega(1/t)$ -- this time however, this instance is 
independent of the choice of remaining vertices in $S$ (owing to the introduction of noise). This hybrid argument allows us to reduce the problem to a low probability [of success] pointer chasing problem, which we tackle in the next step. 

It is worth pointing out that this step matches the intuition that to solve $\NGC$, we need to ``find'' at least one $k$-cycle or a $(k+1)$-path in the graph. 

\paragraph{Step two: applying the streaming XOR Lemma.} The pointer chasing problem we now need to prove a lower bound for requires a really low probability of success, which is way below the threshold for any of 
the standard lower bounds to kick in. Our next step is then to apply our hardness amplification result in~\Cref{thm:xor-lemma} to reduce this to a more standard pointer chasing problem with higher probability of success. This requires us to 
cast our pointer chasing instance as an XOR of several other \emph{independent} pointer chasing instances. 

To do this for $p$-pass algorithms, we ``chop'' the layered graph into $\approx k/p$ \emph{consecutive} groups of $\approx p$ layers. we show how one can carefully connect these groups together
 to get $\approx k/p$ \emph{independent} instances of pointer chasing in each group, so that the XOR of their answers determine the answer to the original problem. This step uses similar
ideas as definition of~\Cref{distngc} and some further randomization tricks. 
We can now apply our streaming XOR Lemma (\Cref{thm:xor-lemma}) and reduce the problem to proving a lower bound for pointer chasing on depth $\approx p$ layered graphs with probability of success $1/2+1/t^{\Theta(p/k)}$, which is 
the content of the next step. 

This step also matches the intuition that to solve $\NGC$ in $p$ passes, we need to be able to ``create'' roughly $k/p$  paths of length $p$ {inside} the \emph{same} $k$-cycle or $(k+1)$-path. 

\paragraph{Step three: a lower bound for the single-copy problem.} We are now in the familiar territory in which the goal is to prove a lower bound for a depth $\approx p$ pointer chasing problem with probability of success 
$1/2+1/t^{\Theta(p/k)}$. The main difference is that our distribution do not match that of standard lower bounds, say~\cite{GuhaM09,NisanW91,PonzioRV99,Yehudayoff16}, which can be made to work with layered graphs but need random degree-one graphs instead of 
random matchings (so higher entropy inputs). Nevertheless, we show that this can be managed with some further crucial modifications. 

All in all, this step allows us to prove 
that solving pointer chasing on depth $(p+1)$ layered graphs in $p$ passes and $n^{1-o(p/k)}$ space does not lead to success probability of $1/2+1/t^{\Theta(p/k)}$. Tracing back these steps and plugging in the parameters, concludes the proof of the theorem. 

Finally, this step also matches the standard intuition that ``finding'' a path of length $>p$ in $p$ passes, with large probability of success, is not possible in much less than near-linear space.  

\subsection{Step One: Decorrelating~\Cref{distngc}}\label{sec:ngc-decor} 

Our goal in this section is to reduce $\NGC$ over $\distNGC$ to solving $\PC$ (over a slightly smaller graph) albeit with a much lower probability of success. Consider the following distribution for $\PC$: 
\begin{Distribution}\label{distpc}
The distribution $\distPC$ for $\PC_{\pcw,\pcd}$ for given width and depth parameters $\pcw,\pcd$. 

\begin{enumerate}[label=$(\roman*)$]
	\item Sample a random $(\pcw,\pcd)$-layered graph with an arbitrary vertex $s \in V_1$ and an arbitrary equipartition $X,Y$ of $V_{d+1}$ (say, both are lexicographically-first options). 
	\item Let the input stream be $M_{\pcd} \conc \cdots \conc M_1$ (with arbitrary orderings in each matching). 
\end{enumerate}
\end{Distribution}

We prove the following lemma in this section. 
\begin{lemma}\label{lem:ngc-decor}
	Suppose there is a $p$-pass $s$-space streaming algorithm $A$ for $\NGC_{n,k}$ on $\distNGC$  that succeeds with probability at least $2/3$. Then, there is a $p$-pass $s$-space streaming algorithm for $\PC_{\pcw,\pcd}$ on $\distPC$
	 for some \underline{even} $\pcw := \Theta(n/k)$ and $\pcd := \frac{k-2}{2}$ with probability of success at least $\frac12 + \frac1{6\pcw}$. 
\end{lemma}

For the rest of this section, we fix the algorithm $A$ in~\Cref{lem:ngc-decor} to use it in a reduction. Our reduction uses a hybrid argument and thus is going to be \emph{algorithm-dependent}, i.e., use $A$ 
in a non-black-box way. To do so, we first need to define a family of hybrid distributions. 

Recall the parameters $t=(n/6k)$ and $\depth=\frac{(k-2)}{2}$ in the definition of $\distNGC$. For any vector $f = (f_1,\ldots,f_t) \in \set{0,1}^t$, we define the distribution $\mu(f)$ as follows: 

\begin{itemize}
	\item \textbf{Hybrid distribution} $\bm{\mu(f)}$: Sample a random $(3t,\depth)$-layered graph $G_0$ (with fixed $S \subseteq V_1$ and equipartition $X,Y$ of $V_{\depth+1}$) conditioned on the following event: ``for any vertex $v_i \in S$, 
	$P(v_i)$ belongs to $X$ if $f_i = 0$ and belongs to $Y$ if $f_i=1$''. Plug this graph $G_0$ in~\Cref{distngc} instead and return the resulting stream for the created graph $G$. 
\end{itemize}

With this definition, we have that $\distNGC = \frac12 \cdot \mu(0^{t}) + \frac12 \cdot \mu(1^t)$. The problem of working with $\mu(0^t)$ and $\mu(1^t)$ directly is that their $\PC$ instances are highly correlated (all vertices in $S$ either go to $X$ or to $Y$). Thus, it is unclear which instance is actually ``solved''. On the other hand, remaining distributions $\mu(f)$ may generate graphs that are not in the support of $\distNGC$ or even 
well-defined for $\NGC$. Nevertheless, we will show that $A$ still needs to do something non-trivial over these distributions: there is a pair of neighboring vectors $g,h$ that differ in {exactly one} coordinate such that 
$A$ is still able to distinguish between them, namely, ``solve'' the pointer chasing instance on their differing index (although with a much lower probability). We now formalize this. 

Let $mem(A)$ denote the final content of the memory of $A$. Let $\mu(g)$ and $\mu(h)$ be any two distributions in the family above. With a slight abuse of notation, we say that $A$ \emph{distinguishes} between $\mu(g)$ and $\mu(h)$ with
probability $p > 0$, if given a sample from either $\mu(g)$ or $\mu(h)$, we can run $A$ over the sample and use maximum likelihood estimation of $mem(A)$ to decide which distribution it was sampled from with probability at least $p$.  
Define the following $t+1$ vectors: 
	\[
	f^0 = (0,\ldots,0), ~~ f^1 = (1, 0,\ldots,0), \quad \cdots \quad f^i = (\underbrace{1,\ldots, 1}_i,  0,\ldots,  0), \quad \cdots \quad f^{t} = (1,\ldots,1).
	\]
We prove that $A$ distinguishes between two consecutive distributions in this sequence. 

\begin{claim}\label{clm:informative-index}
	There is an index $\istar \in [t]$ such that $A$ distinguishes between $\mu(f^{\istar-1})$ and $\mu(f^{\istar})$ with probability at least $ \frac12 + \frac{1}{6t}$. 
\end{claim} 
\begin{proof}
	Since $A$ can solve $\NGC$ on instances drawn from $\distNGC = \frac12 \cdot \mu(0^{t}) + \frac12 \cdot \mu(1^t)$ with probability of success
	 at least $2/3$, we have that (see~\Cref{fact:tvd-sample}),
	\begin{align}
		\tvd{\paren{mem(A) \mid \mu(f^0)}}{\paren{mem(A) \mid \mu(f^t)}} \geq 1/3, \label{eq:mempi-hybrid}
	\end{align}
	as the algorithm uses only $mem(A)$ at the end to output the answer (here $(mem(A) \mid \mu)$ denotes the distribution of $mem(A)$ conditioned on the input sampled from $\mu$).

	Suppose we run algorithm $A$ on each of the distributions $\mu(f^i)$ (even though beside $f^0,f^{t}$, neither of them correspond to an $\NGC$ instance). Then, by triangle inequality, 
	\[
		\tvd{\paren{mem(A) \mid \mu(f^{0})}}{\paren{mem(A) \mid \mu(f^t)}} \leq \sum_{i=1}^{t} \tvd{\paren{mem(A) \mid \mu(f^{i-1})}}{\paren{mem(A) \mid \mu(f^{i})}},
	\]
	which, together with~\Cref{eq:mempi-hybrid}, implies that there is an index $i^* \in [t]$ such that 
	\begin{align}
		\tvd{\paren{mem(A) \mid \mu(f^{i^{*}-1})}}
	{\paren{mem(A) \mid \mu(f^{i^{*})}}} \geq \frac{1}{3t}. 
	\end{align}
Thus, $A$ distinguishes between $\mu(f^{i^{*}-1})$ and $\mu(f^{i^*})$ with probability $\geq \frac12+\frac{1}{6t}$ 
by~\Cref{fact:tvd-sample}.
\end{proof}

Let us define our final distribution $\mu^* := \frac12 \cdot \mu(f^{\istar-1}) + \frac12 \cdot \mu(f^{\istar})$. 
\Cref{clm:informative-index} suggests a way of solving instances of $\PC$ by embedding them in the index $\istar$ of $\mu^*$ and running $A$ over the resulting input. 
We now give a process for sampling from $\mu^*$ which is crucial for this embedding. 	

\renewcommand{\PP}{\mathcal{P}}

\begin{claim}\label{clm:sprocess}
The following process samples a $(3t,\depth)$-layered graph $G_0$ from the distribution of $\mu^*$: 

\begin{enumerate}[label=$(\arabic*)$]
\item\label{s1} Sample $(t-1)$ vertex-disjoint paths from vertices in $S \setminus v_{\istar}$ to vertices in $V_{\depth+1}$  conditioned on the event that ``for any vertex $v_i \in S \setminus \set{v_{\istar}}$, 
$P(v_i)$ is in $X$ if $f^{\istar}_i = 0$ and in $Y$  if $f^{\istar}_i = 1$''.

\item\label{s2} Let $c := \card{(\istar-1) - (t-\istar)}$, the discrepancy in the number of $0$'s and $1$'s in both vectors $f^{\istar-1},f^{\istar}$ when we ignore index $\istar$. 
Sample $c$ random vertex-disjoint path starting from $V_1 \setminus S$ to the remaining vertices of $X$ in $V_{\depth+1}$ if $0$'s of $f^{\istar}$ are fewer that $1$'s, and to $Y$ otherwise. 

\item\label{s3} Sample a random $(2t+1-c,\depth)$-layered graph on the remaining vertices. 

\end{enumerate}
\end{claim}
\begin{proof}
For any vertex $v \in V_1$, define $\PP(v)$ as the path starting from $v$ and ending in $V_{\depth+1}$. Considering $G_0$ is a $(3t,\depth)$-layered graph consisting of perfect matchings between the layers, the paths $\PP(v)$ are 
vertex-disjoint and of length $\depth+1$. As such, we can think of the process of sampling $G_0$ in $\mu^*$ as sampling these vertex-disjoint paths. 

In step $(1)$, we are sampling $\PP(v)$ for all $v \in S \setminus v_{\istar}$. As $f^{\istar-1}_j=f^{\istar}_j$ for all $j \neq \istar$, this step can just sample these paths uniformly at random conditioned on an appropriate endpoint in $X$ and $Y$ for them. 
Thus far, the sampling process is the same as $\mu^*$. 

Let us now examine what happens to the choice of $\PP(v_{\istar})$ at this point. Since $\mu^*$ is a uniform mixture of $f^{\istar-1},f^{\istar}$, the path $\PP(v_{\istar})$ should end up at either $X$ or $Y$ with the same probability. However, 
considering we already conditioned on $\PP(v_j)$ for $j \neq \istar$, the number of remaining $X$ and $Y$ vertices are not equal. This means that $\PP(v_{\istar})$ is \emph{not} a uniformly random path in the rest of the graph. 
The goal of step $(2)$ is to fix this\footnote{A simple analogy may help here: suppose we have four red balls and two green balls and we want to sample a ball uniformly so that its color is red or green with the same probability. We can first
sample two red balls uniformly and throw them out and then sample a ball uniformly  from the rest.}. We can first sample $\PP(v)$ for $c$ vertices in $V_1 \setminus S$ so that they all end up in an $X$ or $Y$ vertex, depending on which
of the sets has more remaining vertices. This equalizes the size of the remaining $X$ and $Y$ vertices, while keeping the distribution intact, using the randomness in choices of these $c$ vertices. 

Finally, at this point, we need to sample $\PP(v)$ for remaining vertices conditioned on $\PP(v_{\istar})$ having the same probability of landing in $X$ or $Y$. Considering sizes of remainder of $X$ and $Y$ are equal, this can be done by sampling a uniform set of vertex-disjoint paths, 
or alternatively, a random layered graph on the remaining vertices which are $3t-(t-1)-c = 2t+1-c$. This is precisely what is done in step $(3)$, concluding the proof.  
\end{proof}

\newcommand{\rR}{\rv{R}}
\newcommand{\rH}{\rv{H}}

A simple corollary of the process in~\Cref{clm:sprocess} is the following conditional independence: let $\rR_1,\rR_2,\rR_3$  denote the random variables for choices in steps $(1),(2),(3)$ of this process; then, conditioned on
any choice $R_1,R_2$ of $\rR_1,\rR_2$, the variable $\rR_3$ is distributed as a random $(2t+1-c,\depth)$-layered graph on the remaining vertices, with independent choice of edges, now that we conditioned on its vertices. Let $H_0 \subseteq G_0$, denote this subgraph. By~\Cref{clm:informative-index} and an averaging argument, there is a choice of $R_1,R_2$ such that, 
 \begin{align}
	\Pr_{H_0 \sim \rR_3} \paren{\text{$A$ distinguishes between $\mu(f^{\istar-1}),\mu(f^{\istar})$} \mid R_1,R_2} \geq \frac12 + \frac1{6t}. \label{eq:avg-cond}
\end{align} 
Distinguishing between $\mu(f^{\istar-1}),\mu(f^{\istar})$ is to decide whether $v_{\istar} \in V_1$, has $P(v_{\istar})$ in $X$ or $Y$ -- this is equivalent to solving $\PC$ over the graph $H_0$ for $s=v_{\istar}$. 
 We now use this to finalize our reduction and prove~\Cref{lem:ngc-decor}. 

\begin{proof}[Proof of \Cref{lem:ngc-decor}]
	Let $c$ be the parameter in~\Cref{clm:sprocess} and note that since $t$ is even (by construction of $\distNGC$), $c$ should be odd (as $c=\card{t-2\istar+1}$).
	Let $\pcw := 2t+1-c$ which is an even number as desired and $\pcd := \depth = \frac{(k-2)}{2}$; moreover note that since $c \leq t-1$, $\pcw \geq t+2 = \Theta(n/k)$. 
	We design a streaming algorithm $B$ from $A$ for $\PC_{\pcw,\pcd}$ for over the distribution $\distPC$. 
	
	Given $G \sim \distPC$, algorithm $B$ uses~\Cref{clm:sprocess} to create the graph 
	\[
	G_0 \sim \mu^* \mid \rR_1=R_1,\rR_2=R_2, \rH_0 = G,
	\]
	 where $R_1,R_2$ are the choices in~\Cref{eq:avg-cond}. To be precise, by setting $H_0 = G$, we mean that the players of $B$ pick a canonical mapping between vertices of $G$ and $H_0$ such that $s=v_{\istar}$, the
	  $X$-vertices (resp. $Y$-vertices) of $G$ are mapped to $X$-vertices ($Y$-vertices), and each player in $A$ with $e \in H_0$ has a unique player in $B$ that simulates it. 
	  The players  then run $A$ over $G_0$ to distinguish $\mu(f^{\istar-1})$ from $\mu(f^{\istar})$ and output $P(s) \in X$ if the answer of $A$ was $\mu(f^{\istar-1})$ and otherwise output $P(s) \in Y$. 
	
	The algorithm $B$ is still a $p$-pass $s$-space algorithm (recall that in~\Cref{def:str}, the players are computationally unbounded and so can do their part of creating the graph $G_0$ without any communication). 
	By the independence property argued for~\Cref{eq:avg-cond}, the distribution of graphs $G_0$ above matches that of this equation. As such, $B$ outputs the correct answer with probability $\frac12+\frac1{6t} \geq \frac12+\frac{1}{6m}$, 
	finalizing the proof. 
\end{proof}


\subsection{Step Two: Applying the Streaming XOR Lemma}\label{sec:ngc-xor} 

By~\Cref{lem:ngc-decor}, our task is reduced to proving a low-probability lower bound for $\PC_{\pcw,\pcd}$ over the distribution $\distPC$. Our goal in this step, is to use our streaming XOR Lemma in~\Cref{thm:xor-lemma}, to reduce this problem
to another $\PC_{\hat{\pcw},\hat{\pcd}}$ problem over distribution $\distPC$ (for  choices of $\hat{\pcw},\hat{\pcd}$ as functions of $\pcw,\pcd$). We prove the following lemma in this section, which realizes our goal.

\begin{lemma}\label{lem:ngc-xor}
	For every $\pcw,\pcd,\ell \in \IN^+$ such that $\pcw$ is even and  $2\ell$ divides $\pcd-1$, the following holds.  
	Suppose there is a $p$-pass $s$-space algorithm $A$ for $\PC_{\pcw,\pcd}$ that succeeds with probability at least $1/2+\delta$ on $\distPC$. 
	Then, there is a $p$-pass $s$-space algorithm for $\PC_{\hpcw,\hpcd}$ over $\distPC$, for some $\hpcw = \frac{\pcw}{2}$ and $\hpcd = \frac{\pcd-1}{2\ell}-1$, that 
	succeeds with probability at least $\frac12 \cdot \delta^{7/\ell}$. 
\end{lemma}

The key to the proof of~\Cref{lem:ngc-xor} is the streaming XOR Lemma that we already established; however, 
to be able to apply the XOR Lemma, we first need to cast $\PC_{\pcw,\pcd}$ as an XOR problem, which we do in the following, using a simple graph product.  

Let us start by defining a simple graph product, which we call the \emph{XOR product}: Given 
$\ell$ layered graphs $G_1,\ldots,G_\ell$ as instances of $\PC$ problem, this product generates a graph $H := \oplus_{i=1}^{\ell} G_i$ such that the answer to a $\PC$ problem on $H$ is equal to XOR of answers to $\PC$ on $G_1,\ldots,G_\ell$.  
We define this product formally as follows. 

\paragraph{XOR product.} Suppose we have a set of $V:= (V_1,\ldots,V_{\depth+1})$ of vertices, each of size $\width$, and an equipartition $X,Y$ of $V_{\depth+1}$. 
Consider $\ell$ different $(\width,\depth)$-layered graphs $G_1,\ldots,G_\ell$ on these sets of vertices. The XOR product graph $H := \oplus_{i=1}^{\ell} G_i$ is the following graph (see also~\Cref{fig:xor-product}): 
\begin{itemize}[label=$-$]
	\item \textbf{Vertex-set:} Create vertex-sets $U^{r,i}_j$ for $r \in [\ell]$, $i \in [4]$,  and $j \in [\depth+1]$, such that for every choice of $r,i$, $U^{r,i}_j$ is a copy of $V_j$. 
	For any  $v \in V_j$, ${copy}(v,r,i)$ denotes the copy of $v$ in $U^{r,i}_j$. Additionally, we define $X^{r,i},Y^{r,i}$ as copies of $X$ and $Y$ in $U^{r,i}_{\depth+1}$. 
	
	\item \textbf{Edge-set:} The first part creates four identical copies of each $G_r$ on $U^{r,i}$-vertices for $i \in [4]$. 
	More formally, for any edge $(u,v)$ in $G_r$, we connect ${copy}(u,r,i)$ to ${copy}(v,r,i)$ for all $i \in [4]$. 
	
	The second part connects these  separate graphs. For every $r \in [\ell]$, connect $X^{r,1}$ to $X^{r,3}$ and $X^{r,2}$ to $X^{r,4}$ using identity perfect matchings. Conversely, 
	connect $Y^{r,1}$ to $Y^{r,4}$ and $Y^{r,2}$ to $Y^{r,3}$ using identity perfect matchings. Finally, for every $r \in [\ell-1]$, connect $U^{r,3}_1$ to $U^{r+1,1}_1$ and $U^{r,4}_1$ to $U^{r+1,2}_1$ 
	using identity perfect matchings. 
\end{itemize}
The outcome of this product is another layered graph with width $2\width$ and depth $\ell \cdot (2\cdot\depth+1) + \ell - 1$. This concludes the description of the XOR product. 
 
 We now state the main property of this product. In the following, for an instance $G$ of the $\PC$ problem, we write $\PC(G) \in \set{0,1}$ to denote the answer of $\PC$ on $G$ which is $0$ if $G$ is a $X$-instance and is $1$ if 
 $G$ is a $Y$-instance. Suppose we have $\ell$ different instances of $\PC_{\hpcw,\hpcd}$ as graphs $G_1,\ldots,G_{\ell}$ on the same set of vertices (and same $s_r,X_r,Y_r$ across all $r \in [\ell]$). 
Consider the $(\pcw,\pcd)$-layered graph $H := \oplus_{r=1}^{\ell} G_i$ and define the following $\PC_{\pcw,\pcd}$ instance over $H$. For 
\[
s = copy(v,1,1), \qquad X = U_1^{\ell,3}, \qquad Y = U_1^{\ell,4},
\]
 decide whether $P(s)$ (in $H$) belongs to $X$ or $Y$ (it is worth pointing out that the sets $X,Y$ defined for $H$ are completely unrelated to sets $X_r,Y_r$ for $r \in [\ell]$). 
 We now have the following claim (a ``proof by picture'' is illustrated in~\Cref{fig:xor-product}). 

 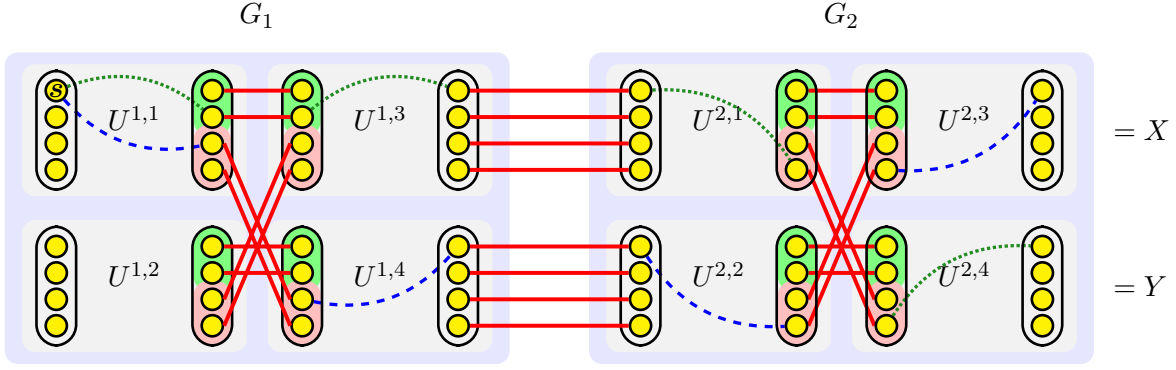
\begin{figure}[t!]	
\centering
\input{figs/xor-product}
\caption{An illustration of the XOR product $H$ of two graphs $G_1,G_2$. Here, the $X^{r,i},Y^{r,i}$ sets are specified for each graph -- these sets are entirely unrelated to equipartition $X,Y$ of $H$ drawn on the right. 
Moreover two potential paths out of $s$ are drawn: $(i)$ the dotted (green) one corresponds to $\PC(G_1)=0,\PC(G_2)=1$ and so $\PC(G_1 \oplus G_2) = \PC(G_1) \oplus \PC(G_2) = 1$ which is true as $s$ reaches $Y$ in $H$; 
$(ii)$ the dashed (blue) one corresponds to $\PC(G_1)=\PC(G_2) = 1$ and so $\PC(G_1 \oplus G_2) = \PC(G_1) \oplus \PC(G_2) = 0$ which is true as $s$ reaches $X$ in $H$.}
\label{fig:xor-product}
\end{figure}

\begin{claim}\label{clm:xor-product}
	For any integer $\ell$ and $H := \oplus_{r=1}^{\ell} G_r$, we have $\PC_{\pcw,\pcd} (H) = \oplus_{r=1}^{\ell} \PC_{\hpcw,\hpcd}(G_r)$. 
\end{claim}
\begin{proof}
	Recall that in any layered graph, each vertex is part of a unique path from the first layer to the last one; for any $v \in H$, we denote this path by $\PP(v)$. 
	Consider any graph $G_r$ for $r \in [\ell]$ and its starting vertex $s_r$. Notice that there are four copies of $s_r$ across subgraphs $U^{r,i}$ for $i \in [4]$. Let us denote these copies by $s^{r,i} = copy(s,r,i)$ for $i \in [4]$. 
	
	If $\PC(G_r) = 0$, then $P(s_r) \in X_r$ by definition. We claim that in this case, $\PP(s^{r,1})$ contains $\PP(s^{r,3})$ because of the following path in $H$: 
	\begin{align*}
		s^{r,1} = copy(s_r,r,1) \leadsto copy(P(s_r),r,1) \in X^{r,1} \rightarrow copy(P(s_r),r,3) \in X^{r,3} \leadsto copy(s_r,r,3) = s^{r,3};
	\end{align*}
	the first path exists because $U^{r,1}_1$ to $U^{r,1}_{\hpcd}$ in $H$ are connected the same as $G_r$, the edge exists by the definition of $H$, and again $U^{r,3}_{\hpcd}$ goes to $U^{r,3}_1$ the same as $G_r$. 
	By the same exact reason, $\PP(s^{r,2})$ contains $\PP(s^{r,4})$. 
	
	Conversely, if $\PC(G_r) = 1$, then $P(s_r) \in Y_r$ by definition. We claim that in this case, $\PP(s^{r,1})$ contains $\PP(s^{r,4})$ instead because of the following path in $H$: 
	\begin{align*}
		s^{r,1} = copy(s_r,r,1) \leadsto copy(P(s_r),r,1) \in Y^{r,1} \rightarrow copy(P(s_r),r,4) \in Y^{r,4} \leadsto copy(s_r,r,4) = s^{r,4};
	\end{align*}
	this is exactly as before except for the fact that $Y^{r,1}$ is instead connected to $Y^{r,4}$ by a perfect matching. Again, we also have that $\PP(s^{r,2})$ contains $\PP(s^{r,3})$ in this case. 
	
	Now, consider the path $\PP(s)$ in $H$. Recall that $s = s^{1,1}$ by definition. By the above argument, the path $\PP(s^{1,1})$ then either contains $s^{1,3}$ in $U^{1,3}_1$ or $s^{1,4}$ in $U^{1,4}_1$. In the first case, 
	$\PP(s^{1,1})$ will then contain $\PP(s^{2,1})$ and in the second case, it will next contain $\PP(s^{2,2})$ by the construction of the last set of edges added to $H$ in its definition. Continuing this inductively from 
	$s^{2,1}$ and $s^{2,2}$ and their corresponding paths $\PP(s^{2,1})$ and $\PP(s^{2,2})$, we get that $\PP(s)$ goes through a collection of vertices $s^{r,i_r}$ for \emph{every} $r \in [\ell]$ and \emph{exactly one} $i_r \in [2]$ until 
	it eventually ends up at either $s^{\ell,3} \in X$ or $s^{\ell,4} \in Y$ (which, we can think of as $s^{\ell+1,1}$ and $s^{\ell+1,2}$, respectively, for the ease of notation). 
	
	Next, consider any pair $s^{r,i_r}$ and $s^{r+1,i_{r+1}}$ on $\PP(s)$. By the argument above, if $\PC(G_r) = 0$, then $i_{r+1} = i_r$ while if $\PC(G_r) = 1$, then $i_{r+1} = 3-i_r$, i.e., it ``flips''. Thus, starting from $s = s^{1,i_1} = s^{1,1}$,
	$\PP(s)$ ends up at $s^{\ell+1,1} = s^{\ell,3} \in X$ if the number of flips is even and at $s^{\ell+1,2} = s^{\ell,4} \in Y$ if it is odd. This means that $\PC_{\pcw,\pcd} (H) = \oplus_{r=1}^{\ell} \PC_{\hpcw,\hpcd}(G_r)$, proving the claim. 
\end{proof}

Equipped with the XOR product and~\Cref{clm:xor-product}, we can now prove~\Cref{lem:ngc-xor}. 

\begin{proof}[Proof of~\Cref{lem:ngc-xor}]
	Consider $\ell$ \emph{independent} instances $G_1,\ldots,G_{\ell}$ of $\PC_{\hpcw,\hpcd}$ sampled from ${\distPC}^{\ell}$. Define $H := \oplus_{r=1}^{\ell} G_r$, which is a 
	$(\pcw,\pcd)$-layered graph for $\pcw = 2\hpcw$ and $\pcd = \ell \cdot (2\cdot\hpcd+1) + \ell - 1$ (these parameters match those of~\Cref{lem:ngc-xor}). By~\Cref{clm:xor-product}, 
	we have $\PC(H) = \oplus_{r=1}^{\ell} \PC(G_r)$. 
	
	We can now apply (the contrapositive of) our streaming XOR Lemma (\Cref{thm:xor-lemma}) to obtain: if we have a $p$-pass $s$-space streaming algorithm for $\PC$ over the distribution of $H = \oplus_{r=1}^{\ell} G_r$ over ${\distPC}^{\ell}$ with 
	probability of success $1/2+\delta$, then we will also have a $p$-pass $s$-space streaming algorithm for $\PC_{\hpcw,\hpcd}$ over $\distPC$ with $1/2 + (\delta)^{7/\ell})$ (by re-parameterizing $\delta$ to match that of~\Cref{thm:xor-lemma}). 
	
	We are still however not done because the algorithm $A$ in the lemma works on the distribution $\distPC$ while the distribution $H$ induced by~${\distPC}^{\ell}$ does not match $\distPC$ due to the reduction. However, this is easy to fix. 
	Pick random permutations $\pi_1,\ldots,\pi_{\pcd+1}$ and use $\pi_i$ to relabel vertices of layer $V_{i}$ of the layered graph $H$ to obtain a graph $G$. The distribution of $G$ is now a random $(\pcw,\pcd)$-layered graph and thus we can run $A$
	over this graph for checking if $\pi_1(s)$ is in $\pi_{\pcd+1}(X)$ or $\pi_{\pcd+1}(Y)$ instead and obtain the answer to $H$ as well. An averaging argument for fixing a choice of $\pi_1,\ldots,\pi_{\pcd+1}$ finalizes the proof. 
\end{proof}


\subsection{Step Three: A Lower Bound for the Single-Copy Problem}\label{sec:ngc-1copy} 

The previous step allows us to instead of proving a lower bound for XOR of $\ell$ copies of the problem, prove a weaker lower bound for a single copy, which translates to a ``standard'' lower bound for pointer chasing. Our goal in this step is to prove this weaker lower bound. We prove the following lemma in this section.

\begin{lemma}\label{lem:ngc-1copy}
	Let $A$ be a $p$-pass $s$-space streaming algorithm for $\PC_{\hpcw,\hpcd}$ over $\distPC$ with probability of success at least $\frac12 + \frac{1}{(3\hpcw)^{7/\ell}}$. Then, either $p > \hpcd-1$ or $s = \Omega(\frac1{\hpcd^5} \cdot \hpcw^{1-28/\ell})$. 
\end{lemma}

The proof of this lemma is similar to the known communication complexity lower bounds for pointer chasing such as~\cite{GuhaM09,NisanW91,PonzioRV99,Yehudayoff16}; the catch however is that these lower bounds are for the distribution wherein each
vertex in a layer $V_i$ \emph{independently} samples a neighbor in $V_{i+1}$ (vertices of $V_{i+1}$ can receive more than one edge) as opposed to a random matching. This independence between the choice of vertices is crucial in these
lower bounds but at the same time working with such a distribution breaks multiple of our reductions (there are other minor differences as well, for instance, we will consider the lower bound directly for streaming algorithms to obtain (slightly) improved bounds
but this is similar to~\cite{GuhaM09}). As such, this final step of our proof is to show a new lower bound for the pointer chasing over the desired distribution, following
the proofs of~\cite{GuhaM09,Yehudayoff16} with some key modifications, in particular to allow for handling random matchings. We prove the following proposition, which implies~\Cref{lem:ngc-1copy}.

\begin{proposition}\label{prop:pc-lower}
	Consider a $p$-pass $s$-space streaming algorithm for $\PC_{\pcw,\pcd}$ over random $(\pcw,\pcd)$-layered graphs with matchings $M_1,\ldots,M_{\pcd}$ given in the stream $M_{\pcd} \conc \cdots \conc M_1$. 
	For  $\gamma \in (0,1/2)$,  if the algorithm succeeds with probability at least $1/2+\gamma$ then either $p > \pcd-1$ or $s=\Omega\paren{{\frac{\gamma^4}{\pcd^5}} \cdot \pcw}$. 
\end{proposition}

This is a good place to point out concretely why we need step two of our approach in~\Cref{lem:ngc-xor}, instead of simply applying~\Cref{prop:pc-lower} directly to our original $\PC$ problem in the reduction of~\Cref{lem:ngc-decor}. This is because the best advantage over random guessing
i.e., parameter $\gamma$, this lower bound can provide is $\ll (\frac1\pcw)^{1/4}$ to give any meaningful space bound. Indeed, none of the other pointer chasing lower bounds such as~\cite{GuhaM09,NisanW91,PonzioRV99,Yehudayoff16} can provide
any meaningful guarantees when $\gamma \approx \frac1\pcw$ while we need an almost-linear lower bound for $\gamma < \frac1\pcw$ to apply our reduction in~\Cref{lem:ngc-decor}. As such, the hardness amplification
step of~\Cref{lem:ngc-xor} is the crucial step in our approach. 

We postpone the proof of~\Cref{prop:pc-lower} to~\Cref{sec:pc-lower} to keep the flow of the current argument and instead show how~\Cref{lem:ngc-1copy} follows immediately from this.

\begin{proof}[Proof of~\Cref{lem:ngc-1copy}]
	Let $\pcw=\hpcw$, $\pcd=\hpcd$, and $\gamma = \frac{1}{(3\hpcw)^{7/\ell}}$. The distribution $\distPC$ is the same as 
	the hard distribution of~\Cref{prop:pc-lower}. As such, if $A$ solves $\PC_{\hpcw,\hpcd}$ with probability of success at least $\frac12 + \frac{1}{(3\hpcw)^{7/\ell}} = \frac12 + \gamma$, 
	then, by~\Cref{prop:pc-lower}, either $p > \hpcd-1$ or 
	\[
	s = \Omega(\frac{\gamma^4}{\hpcd^5} \cdot \hpcw) = \Omega(\frac1{\hpcd^5} \cdot \hpcw^{1-28/\ell}),
	\]
	concluding the proof. 
\end{proof}

\subsection{Putting Everything Together: Proof of~\Cref{thm:ngc}}\label{sec:every}

We are now ready to prove~\Cref{thm:ngc}.

\begin{proof}[Proof of~\Cref{thm:ngc}]
 Let $A$ be a $p$-pass $s$-space streaming algorithm for $\NGC_{n,k}$ with probability of success at least $2/3$ over the distribution $\distPC$. Let us go over each of the three steps in our approach below. 

\begin{itemize}[leftmargin=15pt]
	\item \textbf{Step one:} By~\Cref{lem:ngc-decor}, existence of $A$ implies a $p$-pass $s$-space streaming algorithm $B$ for $\PC_{\pcw,\pcd}$ on $\distPC$ for \emph{even} $\pcw := \Theta(n/k)$ and $\pcd := \frac{k-2}{2}$ with probability
	 of success at least $\frac12 + \frac1{6\pcw}$.
	 \item \textbf{Step two:} Pick $\ell := \frac{b-1}{2p+4}$. By~\Cref{lem:ngc-decor}, existence of $B$ implies a $p$-pass $s$-space streaming algorithm $C$ for $\PC_{\hpcw,\hpcd}$ on $\distPC$ for $\hpcw=\frac{\pcw}{2}$ and 
	 $\hpcd = \frac{\pcd-1}{2\ell}-1 = p+1$ with probability of success at least $\frac12 + (\frac{1}{3\pcw})^{7/\ell}$. 
	 
	 Note that $\pcw$ is even by the guarantee of previous part and $\pcd-1$ divides $2\ell$ by the choice of $\ell$ so we can indeed apply~\Cref{lem:ngc-decor} in this step. 
	 \item \textbf{Step three:} By~\Cref{lem:ngc-1copy}, considering $C$ is a $p$-pass $s$-space streaming algorithm for $\PC_{\hpcw,\hpcd}$ with probability of success at least $\frac12+\frac{1}{10\hpcw^{1/\ell}}$ and $p = \hpcd-1$, 
	 we have that $s = \Omega(\frac1{\hpcd^5} \cdot \hpcw^{1-28/\ell})$. 
	 
\end{itemize}

 We can now retrace these parameters  to the original parameters $n,k$ of $\NGC_{n,k}$. Firstly, $\hpcw = \Theta(\pcw) = \Theta(n/k)$ and $\hpcd = p+1$. Secondly, 
 \[
 	\ell = \frac{b-1}{2p+4} = \frac{(k-2)/2-1}{2p+4} = \frac{k-4}{4p+8}. 
 \]
 As such, the lower bound on the space complexity of all algorithms $A,B$ and $C$ above translates to
 \[
 	s = \Omega\paren{\frac1{p^5} \cdot (n/k)^{1-\frac{28(4p+8)}{k-4}}} = \Omega\paren{\frac1{p^5} \cdot (n/k)^{1-O(p/k)}}.
 \]
 This proves~\Cref{thm:ngc} for infinitely many values of $k \in \IN^+$, i.e., the ones where $\frac{k-4}{4p+8}$ is an integer. 
 
We now extend this lower bound to all values of $k$. Given any integer $k$, find the largest integer $\tilde{k} \leq k$ so that $\frac{\tilde{k}-4}{4p+8}$ is an integer. Clearly, 
 $\tilde{k} = \Theta(k)$. Sample a graph $G$ from $\distNGC$ of~\Cref{distngc} for parameters $n$ and $\tilde{k}$. Recall that in step~\eqref{line:det} of~\Cref{distngc}, we connect the sets $S^1$ to $S^3$, and $S^2$ to $S^4$ using identity perfect matchings. 
We now replace each of these edges with a path of length $k-\tilde k + 1$ (or equivalently, put $k-\tilde k$ new vertices in the middle of each path). This can only increase the number of vertices by a constant factor.

In this new graph, the length of each original 
$\tilde{k}$-cycle becomes $(\tilde{k}-1)+k-\tilde k + 1= k$, and each $2\tilde{k}$-cycle becomes $2(\tilde{k}-1)+2(k-\tilde k+1) = 2k$. As such, we can apply the lower bound for parameters $n$ and $\tilde{k}$ to this graph as well and since 
the number of vertices and $k$ are asymptotically the same, we obtain the desired lower bound. This finalizes the proof of~\Cref{thm:ngc}. 
 \end{proof}
 
 \smallskip
 
 We conclude this section by stating a corollary (of the proof) of~\Cref{thm:ngc} and~\Cref{rem:distngc-known} that we will use in some of our reductions (and can also be useful for future reductions from \NGC). 
 
 \begin{corollary}\label{cor:ngc-strong}
	For every  $k \in \IN^+$, the lower bound of~\Cref{thm:ngc} continues to hold even if we additionally provide the following information to the algorithm beforehand: 
	\begin{enumerate}[label=$(\roman*)$]
		\item One endpoint of every noise path in the graph $G$;
		\item A set of $t$ four-tuples of vertices $(u_1,u_2,u_3,u_4)$ such that in the $k$-cycle case $u_1,u_2$ and $u_3,u_4$ belong to two disjoint $k$-cycles each, while in the $2k$-cycle case, all  belong to the same $2k$-cycle. 
	\end{enumerate}
	Moreover, this lower bound also hold when the graph is \emph{directed} with directed $k$-cycles and $(k-1)$-paths or directed $2k$-cycles and $(k-1)$-paths. 

 \end{corollary}
 \begin{proof}
 	The first two parts follow immediately from~\Cref{rem:distngc-known} as when proving the lower bound for~\Cref{distngc}, we anyway assume this information was known by the streaming algorithm. The last 
	part follows because we can alter~\Cref{distngc} to direct the edges from the first layer to the last one and back (left-to-right for ``inside'' edges in~\Cref{fig:ngc} and right-to-left for ``outside'' ones), which makes the cycles and paths directed. Again, the lower bound holds verbatim for this distribution 
	as well because the partitioning of vertices in the layers are fixed in~\Cref{distngc} and so the direction of the edges does not reveal any new information to the algorithm. 
 \end{proof}

%% file: figs/pcX.tex

\begin{tikzpicture}

\tikzset{choose/.style={rectangle, draw, rounded corners=5pt, line width=1pt, inner xsep=5pt, inner ysep=2pt]}}
\tikzset{layer/.style={rectangle, rounded corners=8pt, draw, black, line width=1pt, inner sep=3pt]}}

\node[vertex] (v11){};
\foreach \j in {2,...,6}
{
	\pgfmathtruncatemacro{\jp}{\j-1};
	\node[vertex] (v1\j) [below=1pt of v1\jp]{};
}
\node[layer] (V1) [fit=(v11) (v16)]{};
\foreach \i in {2,3,4,5}
{
	\pgfmathtruncatemacro{\ip}{\i-1};
	\node[vertex] (v\i1) [right=30pt of v\ip1]{};
	\foreach \j in {2,...,6}
	{
		\pgfmathtruncatemacro{\jp}{\j-1};
		\node[vertex] (v\i\j) [below=1pt of v\i\jp]{};
	}
	\node[layer] (V\i) [fit=(v\i1) (v\i6)]{};
}
\foreach \i in {1,...,4}
{
	\pgfmathtruncatemacro{\ip}{\i+1};
	\foreach \j in {1,...,6}
	{
		\pgfmathtruncatemacro{\jp}{Mod(\i+\j+1,6)+1};
		
		\draw[black, line width=0.75pt]
			(v\i\j) to (v\ip\jp);
	}
}

\begin{scope}[on background layer]
\node [left=5pt of v11]{$s$};
\node[choose, blue!25, fill=blue!25] (S) [fit=(v11) (v11)]{};

\node[choose, green!50, fill=green!50] (X) [fit=(v51) (v53)]{};
\node[choose, red!25, fill=red!25] (Y) [fit=(v54) (v56)]{};
\node [right=3pt of v51] {$P(s)$};
\node [below=10pt of v16] {$V_1$};
\node [below=10pt of v56] {$V_5$};
\end{scope}

\draw[red, line width=1pt]
	(v11) to (v24)
	(v24) to (v32)
	(v32) to (v41)
	(v41) to (v51);

\end{tikzpicture}

%% file: figs/pcY.tex

\begin{tikzpicture}

\tikzset{choose/.style={rectangle, draw, rounded corners=5pt, line width=1pt, inner xsep=5pt, inner ysep=2pt]}}
\tikzset{layer/.style={rectangle, rounded corners=8pt, draw, black, line width=1pt, inner sep=3pt]}}

\node[vertex] (v11){};
\foreach \j in {2,...,6}
{
	\pgfmathtruncatemacro{\jp}{\j-1};
	\node[vertex] (v1\j) [below=1pt of v1\jp]{};
}
\node[layer] (V1) [fit=(v11) (v16)]{};
\foreach \i in {2,3,4,5}
{
	\pgfmathtruncatemacro{\ip}{\i-1};
	\node[vertex] (v\i1) [right=30pt of v\ip1]{};
	\foreach \j in {2,...,6}
	{
		\pgfmathtruncatemacro{\jp}{\j-1};
		\node[vertex] (v\i\j) [below=1pt of v\i\jp]{};
	}
	\node[layer] (V\i) [fit=(v\i1) (v\i6)]{};
}
\foreach \i in {1,...,3}
{
	\pgfmathtruncatemacro{\ip}{\i+1};
	\foreach \j in {1,...,6}
	{
		\pgfmathtruncatemacro{\jp}{Mod(\i+\j+4,6)+1};
		
		\draw[black, line width=0.75pt]
			(v\i\j) to (v\ip\jp);
	}
}

\begin{scope}[on background layer]
\node [left=5pt of v11]{$s$};
\node[choose, blue!25, fill=blue!25] (S) [fit=(v11) (v11)]{};
\node[choose, green!50, fill=green!50] (X) [fit=(v51) (v53)]{};
\node[choose, red!25, fill=red!25] (Y) [fit=(v54) (v56)]{};
\node [right=3pt of v56] {$P(s)$};
\node [below=10pt of v16] {$V_1$};
\node [below=10pt of v56] {$V_5$};
\end{scope}

\draw[black, line width=0.75pt]
(v41)to(v53)
(v42)to(v54)
(v43)to(v55)
(v44)to(v56)
(v45)to(v51)
(v46)to(v52);

\draw[red, line width=1pt]
	(v11) to (v21)
	(v21) to (v32)
	(v32) to (v44)
	(v44) to (v56); 

\end{tikzpicture}

%% file: figs/ngc.tex

\begin{tikzpicture}
	

\tikzset{choose/.style={rectangle, draw, rounded corners=5pt, line width=1pt, inner xsep=5pt, inner ysep=2pt]}}
\tikzset{layer/.style={rectangle, rounded corners=8pt, draw, black, line width=1pt, inner sep=3pt]}}

\node[vertex] (v11){};
\foreach \j in {2,...,6}
{
	\pgfmathtruncatemacro{\jp}{\j-1};
	\node[vertex] (v1\j) [below=1pt of v1\jp]{};
}
\node[layer] (V1) [fit=(v11) (v16)]{};
\foreach \i in {2,3,4,5}
{
	\pgfmathtruncatemacro{\ip}{\i-1};
	\node[vertex] (v\i1) [right=30pt of v\ip1]{};
	\foreach \j in {2,...,6}
	{
		\pgfmathtruncatemacro{\jp}{\j-1};
		\node[vertex] (v\i\j) [below=1pt of v\i\jp]{};
	}
	\node[layer] (V\i) [fit=(v\i1) (v\i6)]{};
}
\foreach \i in {1,...,4}
{
	\pgfmathtruncatemacro{\ip}{\i+1};
	\foreach \j in {1,...,6}
	{
		\pgfmathtruncatemacro{\jp}{Mod(\i+\j+1,6)+1};
		
		\draw[black, line width=0.75pt]
			(v\i\j) to (v\ip\jp);
	}
}

\begin{scope}[on background layer]
\node[choose, dashed, line width=1pt, gray, fill=gray!10] (G1) [fit=(V1) (V5)]{};
\node [left=5pt of G1] {$G_1$};
\node[choose, blue!25, fill=blue!25] (S) [fit=(v11) (v12)]{};
\node[choose, green!50, fill=green!50] (X) [fit=(v51) (v53)]{};
\node[choose, red!25, fill=red!25] (Y) [fit=(v54) (v56)]{};
\end{scope}

\node[vertex] (2v11) [below=70pt of v11]{};
\foreach \j in {2,...,6}
{
	\pgfmathtruncatemacro{\jp}{\j-1};
	\node[vertex] (2v1\j) [below=1pt of 2v1\jp]{};
}
\node[layer] (2V1) [fit=(2v11) (2v16)]{};
\foreach \i in {2,3,4,5}
{
	\pgfmathtruncatemacro{\ip}{\i-1};
	\node[vertex] (2v\i1) [right=30pt of 2v\ip1]{};
	\foreach \j in {2,...,6}
	{
		\pgfmathtruncatemacro{\jp}{\j-1};
		\node[vertex] (2v\i\j) [below=1pt of 2v\i\jp]{};
	}
	\node[layer] (2V\i) [fit=(2v\i1) (2v\i6)]{};
}
\foreach \i in {1,...,4}
{
	\pgfmathtruncatemacro{\ip}{\i+1};
	\foreach \j in {1,...,6}
	{
		\pgfmathtruncatemacro{\jp}{Mod(\i+\j+1,6)+1};
		
		\draw[black, line width=0.75pt]
			(2v\i\j) to (2v\ip\jp);
	}
}

\begin{scope}[on background layer]
\node[choose, dashed, line width=1pt, gray, fill=gray!10] (G2) [fit=(2V1) (2V5)]{};
\node [left=5pt of G2] {$G_2$};
\node[choose, blue!25, fill=blue!25] (S) [fit=(2v11) (2v12)]{};
\node[choose, green!50, fill=green!50] (X) [fit=(2v51) (2v53)]{};
\node[choose, red!25, fill=red!25] (Y) [fit=(2v54) (2v56)]{};
\end{scope}

\node[vertex] (3v11) [right=350pt of v11]{};
\foreach \j in {2,...,6}
{
	\pgfmathtruncatemacro{\jp}{\j-1};
	\node[vertex] (3v1\j) [below=1pt of 3v1\jp]{};
}
\node[layer] (3V1) [fit=(3v11) (3v16)]{};
\foreach \i in {2,3,4,5}
{
	\pgfmathtruncatemacro{\ip}{\i-1};
	\node[vertex] (3v\i1) [left=30pt of 3v\ip1]{};
	\foreach \j in {2,...,6}
	{
		\pgfmathtruncatemacro{\jp}{\j-1};
		\node[vertex] (3v\i\j) [below=1pt of 3v\i\jp]{};
	}
	\node[layer] (3V\i) [fit=(3v\i1) (3v\i6)]{};
}
\foreach \i in {1,...,4}
{
	\pgfmathtruncatemacro{\ip}{\i+1};
	\foreach \j in {1,...,6}
	{
		\pgfmathtruncatemacro{\jp}{Mod(\i+\j+1,6)+1};
		
		\draw[black, line width=0.75pt]
			(3v\i\j) to (3v\ip\jp);
	}
}

\begin{scope}[on background layer]
\node[choose, dashed, line width=1pt, gray, fill=gray!10] (G3) [fit=(3V1) (3V5)]{};
\node [right=5pt of G3] {$G_3$};
\node[choose, blue!25, fill=blue!25] (S) [fit=(3v11) (3v12)]{};
\node[choose, green!50, fill=green!50] (X) [fit=(3v51) (3v53)]{};
\node[choose, red!25, fill=red!25] (Y) [fit=(3v54) (3v56)]{};
\end{scope}

\node[vertex] (4v11) [below=70pt of 3v11]{};
\foreach \j in {2,...,6}
{
	\pgfmathtruncatemacro{\jp}{\j-1};
	\node[vertex] (4v1\j) [below=1pt of 4v1\jp]{};
}
\node[layer] (4V1) [fit=(4v11) (4v16)]{};
\foreach \i in {2,3,4,5}
{
	\pgfmathtruncatemacro{\ip}{\i-1};
	\node[vertex] (4v\i1) [left=30pt of 4v\ip1]{};
	\foreach \j in {2,...,6}
	{
		\pgfmathtruncatemacro{\jp}{\j-1};
		\node[vertex] (4v\i\j) [below=1pt of 4v\i\jp]{};
	}
	\node[layer] (4V\i) [fit=(4v\i1) (4v\i6)]{};
}
\foreach \i in {1,...,4}
{
	\pgfmathtruncatemacro{\ip}{\i+1};
	\foreach \j in {1,...,6}
	{
		\pgfmathtruncatemacro{\jp}{Mod(\i+\j+1,6)+1};
		
		\draw[black, line width=0.75pt]
			(4v\i\j) to (4v\ip\jp);
	}
}

\begin{scope}[on background layer]
\node[choose, dashed, line width=1pt, gray, fill=gray!10] (G4) [fit=(4V1) (4V5)]{};
\node [right=5pt of G4] {$G_4$};
\node[choose, blue!25, fill=blue!25] (S) [fit=(4v11) (4v12)]{};
\node[choose, green!50, fill=green!50] (X) [fit=(4v51) (4v53)]{};
\node[choose, red!25, fill=red!25] (Y) [fit=(4v54) (4v56)]{};
\end{scope}

\foreach \j in {1,...,3}
{
	\pgfmathtruncatemacro{\jp}{\j+3};
	\draw[red, line width=1.5pt]
		(v5\j.east) to (3v5\j.west)
		(v5\jp.east) to (4v5\jp.west)
		(2v5\jp.east) to (3v5\jp.west)
		(2v5\j.east) to (4v5\j.west);
}

\draw[red, line width=1.5pt]
	(v11.west) to ($(v11)+(-15pt,0pt)$) to ($(v11)+(-15pt,18pt)$) to ($(3v11)+(15pt,18pt)$) to ($(3v11)+(15pt,0pt)$) to (3v11.east)
	(v12.west) to ($(v12)+(-20pt,0pt)$) to ($(v12)+(-20pt,18pt)$) to ($(3v12)+(20pt,18pt)$) to ($(3v12)+(20pt,0pt)$) to (3v12.east)
	(2v11.west) to ($(2v11)+(-15pt,0pt)$) to ($(2v11)+(-15pt,-58pt)$) to ($(4v11)+(15pt,-58pt)$) to ($(4v11)+(15pt,0pt)$) to (4v11.east)
	(2v12.west) to ($(2v12)+(-20pt,0pt)$) to ($(2v12)+(-20pt,-58pt)$) to ($(4v12)+(20pt,-58pt)$) to ($(4v12)+(20pt,0pt)$) to (4v12.east);

\end{tikzpicture}

%% file: figs/xor-product.tex

\begin{tikzpicture}
	

\tikzset{choose/.style={rectangle, draw, rounded corners=5pt, line width=1pt, inner xsep=5pt, inner ysep=2pt]}}
\tikzset{layer/.style={rectangle, rounded corners=8pt, draw, black, line width=1pt, inner sep=3pt]}}


\node[vertex] (1v11){$\bm{s}$};
\node[vertex] [below=50pt of 1v11](1v21){};
\node[vertex] [right=50pt of 1v11](1v31){};
\node[vertex] [below=50pt of 1v31](1v41){};

\foreach \i in {1,2,3,4}
{
\foreach \j in {2,...,4}
{
	\pgfmathtruncatemacro{\jp}{\j-1};
	\node[vertex] (1v\i\j) [below=1pt of 1v\i\jp]{};
}
\node[layer] (1V\i) [fit=(1v\i1) (1v\i4)]{};
}

\node[vertex] [right=25pt of 1v31](2v11){};
\node[vertex] [below=50pt of 2v11](2v21){};
\node[vertex] [right=50pt of 2v11](2v31){};
\node[vertex] [below=50pt of 2v31](2v41){};

\foreach \i in {1,2,3,4}
{
\foreach \j in {2,...,4}
{
	\pgfmathtruncatemacro{\jp}{\j-1};
	\node[vertex] (2v\i\j) [below=1pt of 2v\i\jp]{};
}
\node[layer] (2V\i) [fit=(2v\i1) (2v\i4)]{};
}

\begin{scope}[on background layer]
\node[choose, draw=none, line width=1pt, blue, fill=blue!10, inner ysep=10pt, inner xsep=15pt] (B1) [fit=(1v11) (2v44)]{};
\node [above=5pt of B1] {$G_1$};
\foreach \r in {1,2}
{
\pgfmathtruncatemacro{\rp}{2*\r-1};
\pgfmathtruncatemacro{\rpp}{2*\r};
\node[choose, draw=none, line width=0pt, black, fill=gray!10] (\r G1) [fit=(\r V1) (\r V3)]{$U^{1,\rp}$};
\node[choose, draw=none, line width=0pt, black, fill=gray!10] (\r G2) [fit=(\r V2) (\r V4)]{$U^{1,\rpp}$};
}
\end{scope}

\begin{scope}[on background layer]

\node[choose, green!50, fill=green!50, inner sep=2pt] (1Xt) [fit=(1v31) (1v32)]{};
\node[choose, green!50, fill=green!50, inner sep=2pt] (1Xb) [fit=(1v41) (1v42)]{};
\node[choose, red!25, fill=red!25, inner sep=2pt]  (1Yt) [fit=(1v33) (1v34)]{};
\node[choose, red!25, fill=red!25, inner sep=2pt] (1Yb) [fit=(1v43) (1v44)]{};

\node[choose, green!50, fill=green!50, inner sep=2pt] (2Xt) [fit=(2v11) (2v12)]{};
\node[choose, green!50, fill=green!50, inner sep=2pt] (2Xb) [fit=(2v21) (2v22)]{};
\node[choose, red!25, fill=red!25, inner sep=2pt]  (2Yt) [fit=(2v13) (2v14)]{};
\node[choose, red!25, fill=red!25, inner sep=2pt] (2Yb) [fit=(2v23) (2v24)]{};

\end{scope}

\draw[red, line width=1.5pt]
	(1v31.east) to (2v11.west)
	(1v32.east) to (2v12.west)
	(1v33.east) to (2v23.west)
	(1v34.east) to (2v24.west)
	(1v41.east) to (2v21.west)
	(1v42.east) to (2v22.west)
	(1v43.east) to (2v13.west)
	(1v44.east) to (2v14.west);

\node[vertex] [right=60pt of 2v31] (1x11){};
\node[vertex] [below=50pt of 1x11](1x21){};
\node[vertex] [right=50pt of 1x11](1x31){};
\node[vertex] [below=50pt of 1x31](1x41){};

\foreach \i in {1,2,3,4}
{
\foreach \j in {2,...,4}
{
	\pgfmathtruncatemacro{\jp}{\j-1};
	\node[vertex] (1x\i\j) [below=1pt of 1x\i\jp]{};
}
\node[layer] (1X\i) [fit=(1x\i1) (1x\i4)]{};
}

\node[vertex] [right=25pt of 1x31](2x11){};
\node[vertex] [below=50pt of 2x11](2x21){};
\node[vertex] [right=50pt of 2x11](2x31){};
\node[vertex] [below=50pt of 2x31](2x41){};

\foreach \i in {1,2,3,4}
{
\foreach \j in {2,...,4}
{
	\pgfmathtruncatemacro{\jp}{\j-1};
	\node[vertex] (2x\i\j) [below=1pt of 2x\i\jp]{};
}
\node[layer] (2X\i) [fit=(2x\i1) (2x\i4)]{};
}

\node [right=15pt of 2X3] {$ = X$};
\node [right=15pt of 2X4] {$ = Y$};

\begin{scope}[on background layer]
\node[choose, draw=none, line width=1pt, blue, fill=blue!10, inner ysep=10pt, inner xsep=15pt] (xB1) [fit=(1x11) (2x44)]{};
\node [above=5pt of xB1] {$G_2$};
\foreach \r in {1,2}
{
\pgfmathtruncatemacro{\rp}{2*\r-1};
\pgfmathtruncatemacro{\rpp}{2*\r};
\node[choose, draw=none, , line width=0pt, black, fill=gray!10] (\r xG1) [fit=(\r X1) (\r X3)]{$U^{2,\rp}$};
\node[choose, draw=none, line width=0pt, black, fill=gray!10] (\r xG2) [fit=(\r X2) (\r X4)]{$U^{2,\rpp}$};
}
\end{scope}

\begin{scope}[on background layer]

\node[choose, green!50, fill=green!50, inner sep=2pt] (1xXt) [fit=(1x31) (1x32)]{};
\node[choose, green!50, fill=green!50, inner sep=2pt] (1xXb) [fit=(1x41) (1x42)]{};
\node[choose, red!25, fill=red!25, inner sep=2pt]  (1xYt) [fit=(1x33) (1x34)]{};
\node[choose, red!25, fill=red!25, inner sep=2pt] (1xYb) [fit=(1x43) (1x44)]{};

\node[choose, green!50, fill=green!50, inner sep=2pt] (2xXt) [fit=(2x11) (2x12)]{};
\node[choose, green!50, fill=green!50, inner sep=2pt] (2xXb) [fit=(2x21) (2x22)]{};
\node[choose, red!25, fill=red!25, inner sep=2pt]  (2xYt) [fit=(2x13) (2x14)]{};
\node[choose, red!25, fill=red!25, inner sep=2pt] (2xYb) [fit=(2x23) (2x24)]{};

\draw[red, line width=1.5pt]
	(1x31.east) to (2x11.west)
	(1x32.east) to (2x12.west)
	(1x33.east) to (2x23.west)
	(1x34.east) to (2x24.west)
	(1x41.east) to (2x21.west)
	(1x42.east) to (2x22.west)
	(1x43.east) to (2x13.west)
	(1x44.east) to (2x14.west);
	
\draw[blue, line width=1.25pt,  dashed, bend right]
	(1v11) to (1v33) 
	(2v23) to (2v41)
	(1x21) to (1x44)
	(2x14) to (2x31);

\draw[ForestGreen, line width=1.25pt, densely dotted, bend left]
	(1v11) to (1v32) 
	(2v12) to (2v31)
	(1x11) to (1x34)
	(2x24) to (2x41);
		
\end{scope}

\draw[red, line width=1.5pt]
	(2v31.east) to (1x11.west)
	(2v32.east) to (1x12.west)
	(2v33.east) to (1x13.west)
	(2v34.east) to (1x14.west)
	(2v41.east) to (1x21.west)
	(2v42.east) to (1x22.west)
	(2v43.east) to (1x23.west)
	(2v44.east) to (1x24.west);

\end{tikzpicture}

%% file: 5reductions.tex

\section{Other Graph Streaming Lower Bounds via Reductions}\label{sec:reductions}

We now present several reductions from $\NGC$ for establishing graph streaming lower bounds. These results collectively formalize~\Cref{res:str}. 
Our reductions are  similar in spirit to the ones in prior work and particularly~\cite{AssadiKSY20} (based on their \emph{One-or-Many-Cycles problem} which is 
another variant of gap cycle counting problems);  the main novelty here
is how we can handle the ``noise'' part of $\NGC$ in the reductions but this is not particularly challenging (and become mostly relevant to problems such as minimum weight spanning tree or property testing of connectivity). 
In the following, we will define each problem and present the most relevant prior work, our new result, and a short discussion on whether our result is/seems to be optimal or not. 
We refer the interested reader to~\cite[Section 7 and Appendix B]{AssadiKSY20} for a comprehensive
summary of the prior results on the problems studied in this section. 

Before moving on, let us note that the work of~\cite{AssadiKSY20} also considers some lower bounds beyond graph streams for problems such as Schatten norms of matrices or Sorting-by-Block-Interchange; our new lower bounds 
also apply to these problems with some additional work but we opted to focus primarily on graph streaming lower bounds in this paper (with the exception of the Matrix Rank problem which follows immediately from our results).

 \subsection{Minimum Spanning Tree} 

Given an undirected graph $G=(V,E)$, with edge-weights $w: E \rightarrow \{1,2, \dots, W \}$, the minimum spanning tree problem asks for an estimate to the weight of a spanning tree in $G$ with the least weight, denoted by 
MST of $G$. This is one the earliest problems studied in the streaming setting~\cite{FeigenbaumKMSZ08,AhnGM12a,HuangP16} and it is known that $O(n\log{(nW)})$ space and a single-pass suffices for finding
an \emph{exact} MST~\cite{FeigenbaumKMSZ08} and $n^{1-\Theta(\eps/W)}$ and a single-pass for finding a $(1+\eps)$-approximation~\cite{HuangP16}. On the lower bound front, $\Omega(n)$ and $n^{1-O(\eps/W)}$ lower bounds 
for exact answer and approximate answer via single-pass algorithms where proved in~\cite{FeigenbaumKMSZ08} and~\cite{HuangP16}, respectively. The only known multi-pass lower bound for $(1+\eps)$-approximation is that of~\cite{AssadiKSY20} that proves 
that $n^{o(1)}$-space needs $\Omega(\log{(1/\eps)})$ passes.

We present the following lower bound for this problem:
\begin{theorem}\label{thm:MST}
For $\eps \in (0,1)$ and $W \in \IN^+$, any $p$-pass streaming algorithm for $(1+\eps)$-approximation of weight of MST on $n$-vertex graphs of maximum weight $W$ with probability  at least $2/3$ requires 
$\Omega\paren{\frac1{p^5} \cdot (\eps \cdot n/W)^{1-O(\eps \cdot p/W)}}$ space. This lower bound continues to hold even on bounded-degree planar graphs 
and also implies that $\Omega(1/\eps)$ passes are needed for $n^{o(1)}$-space  even for $W=O(1)$. 
\end{theorem}

\begin{proof}
Consider  a $\NGC$ instance $G$ on $n=6t \cdot k$ vertices  from $\distNGC$ for largest integer $k \leq \frac{W-1}{12\eps}$. We crucially use \Cref{cor:ngc-strong}  to perform this reduction.

Using~\Cref{cor:ngc-strong} allows us to assume that the algorithm for $G$ has the extra knowledge of the following: a collection
of $t$ tuples $(u^i_1,u^i_2,u^i_3,u^i_4)$ for $i \in [t]$ (covering the cycles) and $4t$ vertices $v_1,\ldots, v_t$ (covering the paths), with the properties specified by~\Cref{cor:ngc-strong}. 

Pick any arbitrary bijection $\phi: [t] \rightarrow [t]$ such that $\phi(i) \neq i$ for $i \in [t]$, and any arbitrary injective one-to-three mapping $\psi: [t] \rightarrow [4t]$. Create the following graph from $G$ (see~\Cref{fig:mst}): 
\begin{enumerate}[label=$(\roman*)$]
	\item For any $i \in [t]$, connect $u^i_1$ to $u^{\phi(i)}_3$ with a new edge of weight $1$. 
	\item For any $i \in [t]$, connect $u^i_2$ to $u^i_4$ with a new edge of weight $W$. 
	\item For any $i \in [t]$, let $\psi(i) = (a^i_1,a^i_2,a^i_3,a^i_4)$ and connect $u^i_1$ to $v_{a^i_1},v_{a^i_2},v_{a^i_3},v_{a^i_4}$. 
\end{enumerate}
This addition to the graph can be created by any streaming algorithm of~\Cref{def:str} without spending any extra space. In the following, we use $H$ to refer to this new graph obtained from $G$ and $H_1$ to be the subgraph of 
$H$ consists of only the edges of weight $1$ in $H$.

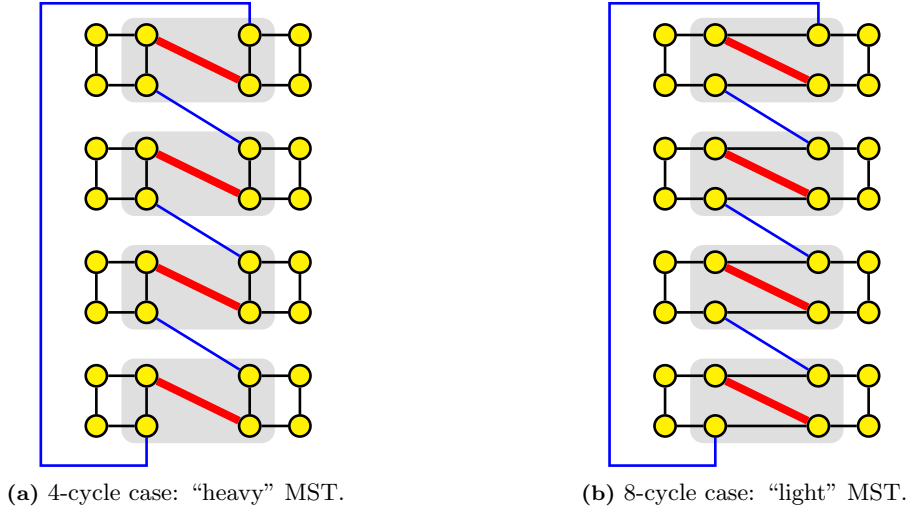
\begin{figure}[h!]	
\centering
\subcaptionbox{\footnotesize $4$-cycle case: ``heavy'' MST.}%
  [.45\linewidth]{\input{figs/mst1}}
\subcaptionbox{\footnotesize $8$-cycle case: ``light'' MST.}%
  [.45\linewidth]{\input{figs/mst2}}
\caption{An illustration of the MST reduction. The vertices inside each block (gray) are one of the four tuples $(u^i_1,u^i_2,u^i_3,u^i_4)$. The middle thick (red) edges have weight $W$ while middle thin (blue) edges have weight $1$; both these groups
of edges are added as part of the reduction. The outer thin (black) edges are the original edges of the $4$-cycle vs $8$-cycle problem. To avoid clutter, we have not drawn the noise paths, however, they can be thought of as being partitioned into $t$ groups 
of size four and connecting group $i$ to vertex $u^i_1$ of each tuple using an edge of weight $1$; clearly, this does not break planarity. The graph on the left has $4$ connected component without the heavy edges while the right one is connected even without those edges. } 
\label{fig:mst}
\end{figure}

Firstly, suppose that $G$ belongs to the $k$-cycle case. We claim that in this case, $H_1$ consists of $t$ connected components. These components are, for all $i \in [t]$, $u^i_1$ and its cycle in $G$ (including $u^i_2$), plus $u^{\phi(i)}_3$ 
and its cycle in $G$ (including $u^{\phi(i)}_3$), plus $v_{a^i_1},v_{a^i_2},v_{a^i_3},v_{a^i_4}$ and their paths in $G$ for $ (a^i_1,a^i_2,a^i_3,a^i_4) = \psi(i)$; all these vertices belong to the same connected and all edges going out of this component
only has weight $W$ and thus does not belong to $H_1$. 

On the other hand, we claim that when $G$ belongs to the $2k$-cycle case, $H_1$ becomes connected. For all $i \in [t]$, the cycle of $u^i_1$ also contains $u^i_3$; for the sake of argument 
suppose we add a new edge of weight $1$ between $u^i_1$ and $u^i_3$ (this cannot change the connectivity of the graph $H_1$). Now we have two edge-disjoint perfect matchings between $\set{u^i_1 \mid i \in [t]}$ on one side
and $\set{u^i_3 \mid i \in [t]}$ on the other side, one from the new (artificial) edges we just added and another from the edges in $H_1$ between $u^i_1$ and $u^{\phi(i)}_3$. These two matchings form a Hamiltonian cycle over 
these sets thus connecting them all together. As all other vertices of the graph are connected to some $u^i_1$, the entire graph becomes connected. 

Now in the first case, any MST of $H$ needs to pick at least $t-1$ edges from $H \setminus H_1$ which are all of weight $W$, 
making its weight at least $n-t + (t-1) \cdot W$ (the fact that $H$ itself is connected is exactly the same $H_1$ being connected in the second case). In the second case, every MST of $H$ has weight $n-1$ as $H_1$ is already connected. 
The choice of $k$ ensures that the weight of MST of $H$ in the first case is $(1+\eps)$ times larger that the second case.

As such a $p$-pass streaming algorithm for $(1+\eps)$-approximation MST can be used as a distinguisher for the two cases of the graph $G$. By~\Cref{cor:ngc-strong}, 
we have that that space complexity of the algorithm should be
\[
\Omega\paren{\frac1{p^5} \cdot (n/k)^{1-O(p/k)}} = \Omega\paren{\frac1{p^5} \cdot (\eps \cdot n/W)^{1-O(\eps \cdot p/W)}}.
\]

The fact that the lower bound holds on bounded-degree planar graphs is simply because the distribution $\distNGC$ in~\Cref{distngc} we use here is supported on graphs which are disjoint-union of cycles and paths and the set of edges we added 
increase the degree by at most a constant factor and does not break the planarity (see~\Cref{fig:mst} for an illustration). 
\end{proof}

\paragraph{Optimality of~\Cref{thm:MST}?}The bounds obtained by our algorithm are actually \emph{asymptotically optimal} for any constant $W \in \IN^+$. To obtain the upper bound, we can run a streaming implementation of 
the query algorithm of~\cite{ChazelleRT05} in $O(1/\eps)$ passes and $O_{\eps}(\poly\!\log{n})$ space (see~\cite{HuangP16} and~\cite{MonemizadehMPS17} for details on simulating these query algorithms in the streaming model -- note
that in general, by allowing $p$ passes over the input, we can simulate $p$ rounds of adaptive querying in a straightforward way).

\subsection{Maximum Matching Size and Matrix Rank}
In the maximum matching size problem, our goal is to output an estimate to the \emph{size} of the maximum matching of the input undirected graph $G(V,E)$, i.e. the largest set of vertex-disjoint edges in $G$. 
Maximum matching is among the most studied problems in the streaming setting and listing the prior results on this problem is beyond the scope of our work. 
We only note that there are various algorithms for approximating matching size in $\polylog{(n)}$ space in arbitrary graphs in random-order streams~\cite{MonemizadehMPS17,KapralovKS14,KapralovMNT20} or planar graphs
 in adversarial streams~\cite{EsfandiariHLMO15,ChitnisCEHMMV16,CormodeJMM17,McGregorV18}. The best single-pass lower bounds for this problem rule out $<(3/2)$-approximation in $o(\sqrt{n})$ space~\cite{EsfandiariHLMO15}
 and $(1+\eps)$-approximation in $n^{1-O(\eps)}$ space~\cite{BuryS15} (and  $n^{2-O(\eps)}$ space on dense graphs~\cite{AssadiKL17}); the only known multi-pass lower bound is that
 $(1+\eps)$-approximation in $n^{o(1)}$-space needs $\Omega(\log{(1/\eps)})$ passes~\cite{AssadiKSY20}.

We present the following lower bound for this problem:

\begin{theorem}\label{thm:maxmatching}
For any $\eps \in (0,1)$, any $p$-pass streaming algorithm for $(1+\eps)$-approximation of size of maximum matching  on $n$-vertex graphs with probability  at least $2/3$ requires 
$\Omega\paren{\frac1{p^5} \cdot (\eps \cdot n )^{1-O(\eps \cdot p)}}$ space. Moreover, this lower bound continues to hold even on bounded-degree planar graphs 
and also implies that $\Omega(1/\eps)$ passes are needed for any $n^{o(1)}$-space algorithm. 
\end{theorem}

\begin{proof}

Consider a $\NGC$ instance $G$ on $n=6t \cdot k$ vertices  from $\distNGC$, for largest \emph{odd} integer $k\leq \frac 1{3 \eps}$.

In one case, $G$ contains $2t$ vertex-disjoint \emph{odd} cycles of length $k$ each, and $4t$ vertex disjoint paths of  length $k-1$. Any maximum matching of $G$ can match $\frac{k-1}{2}$ edges from each odd cycle 
and $\frac{k-1}{2}$ from each path of length $k-1$. Thus the value of maximum matching in this case is $t \cdot (k-1)+2t \cdot (k-1)$, which equals $(n/2)-(n/2k)$.

In the other case, $G$ contains $t$ vertex-disjoint \emph{even} cycles of length $2k$ each, and $4t$ vertex-disjoint paths of length $k-1$. Any maximum matching of $G$ can match $k$ edges from each odd cycle 
and $\frac{k-1}{2}$ from each path of length $k-1$. Thus the value of maximum matching in this case is $t \cdot k+2t \cdot (k-1)$, which equals $(n/2)-(n/2k) + (n/6k) > (n/2)-(n/2k)+\eps (n/2)$.

As such a $p$-pass streaming algorithm for $(1+\eps)$-approximation of maximum matching can be used as a distinguisher for the two cases of the graph $G$. By~\Cref{thm:ngc}, 
we have that that space complexity of the algorithm should be
\[
\Omega\paren{\frac1{p^5} \cdot (n/k)^{1-O(p/k)}} = \Omega\paren{\frac1{p^5} \cdot (\eps \cdot n )^{1-O(\eps \cdot p)}}.
\]

The fact that the lower bound holds on bounded-degree planar is simply because the distribution $\distNGC$ in~\Cref{distngc} we use here is supported on graphs which are disjoint-union of cycles and paths, and thus clearly are both 
bounded-degree and planar. 
\end{proof}

As a consequence of the standard equivalence between estimating  matching size and computing the rank of the Tutte matrix~\cite{Tutte47} with entries chosen randomly established in \cite{Lovasz79}
(\cite{BuryS15} performs this reduction in the streaming model), we get the following result as well. 
\begin{corollary}\label{cor:matrixnorm}
For any $\eps \in (0,1)$, any $p$-pass streaming algorithm for $(1+\eps)$-approximation of rank $n$-by-$n$ matrices with probability  at least $2/3$ requires 
$\Omega\paren{\frac1{p^5} \cdot (\eps \cdot n )^{1-O(\eps \cdot p)}}$ space. Moreover, this lower bound continues to hold even on matrices with $O(1)$ entries per row and column
and also implies that $\Omega(1/\eps)$ passes are needed for any $n^{o(1)}$-space algorithm. 
\end{corollary}

This result considerably strengthen prior bounds in~\cite{BuryS15,LiW16,AssadiKSY20} for this fundamental problem. 

\paragraph{Optimality of~\Cref{thm:maxmatching}?}~\Cref{thm:maxmatching} provides the currently best multi-pass lower bound for $(1+\eps)$-approximation of maximum matching in \emph{any} family of graphs. 
We do not know whether there is a matching upper bound as well, namely, an algorithm with $n^{o(1)}$-space and $O(1/\eps)$ passes on general graphs or the lower bound can be improved further. 

There are  however already known algorithms for this problem on bounded degree graphs. In particular,~\cite{MonemizadehMPS17} give a streaming implementation of the query algorithm of~\cite{NguyenO08} that achieves a $\pm \eps n$ \emph{additive} 
approximation to maximum matching in $O_{\eps}(\log{n})$ bits of space in a single-pass in \emph{random-streams}; the same exact algorithm can also be implemented in this much space and $O_{\eps}(1)$ passes in \emph{arbitrary} streams although the
dependence is much worse than $\Omega(1/\eps)$ in our lower bound. Closing this gap remains a fascinating open question. 

\subsection{Maximum Cut}
In the maximum cut problem, our goal is to estimate the largest \emph{value} of a cut in an input graph $G(V,E)$ i.e. output an estimate of the size of a bi-partition of vertices maximizing the number of crossing edges.
This problem has been studied extensively in the graph streaming model in~\cite{KapralovKS15,KoganK15,BhaskaraDV18,KapralovK19,KapralovKSV17,AssadiKSY20}, with best lower bound of $<2$-approximation in $\Omega(n)$ space
in single-pass graphs~\cite{KapralovK19} and $(1+\eps)$-approximation in $n^{o(1)}$-space in $\Omega(\log{(1/\eps)})$ passes~\cite{AssadiKSY20}. 

We present the following lower bound for this problem.

\begin{theorem}\label{thm:maxcut}
For any $\eps \in (0,1)$, any $p$-pass streaming algorithm for $(1+\eps)$-approximation of value of maximum cut  on $n$-vertex graphs with probability  at least $2/3$ requires 
$\Omega\paren{\frac1{p^5} \cdot (\eps \cdot n )^{1-O(\eps \cdot p)}}$ space. Moreover, this lower bound continues to hold even on bounded-degree planar graphs 
and also implies that $\Omega(1/\eps)$ passes are needed for any $n^{o(1)}$-space algorithm. 
\end{theorem}
\begin{proof}
Consider a $\NGC$ instance $G$ on $n=6t \cdot k$ vertices  from $\distNGC$ for largest \emph{odd} integer $k\leq \frac 1{3 \eps}$.

In one case, $G$ contains $2t$ vertex-disjoint \emph{odd} cycles of length $k$ each, and $4t$ vertex disjoint paths of  length $k-1$. Any cut of $G$ must leave out one edge from each cycle, and thus, the value of maximum cut 
in this case is $2t \cdot (k-1)+4t \cdot (k-1)$, which equals $n-n/k$.

In the other case, $G$ contains $t$ vertex-disjoint \emph{even} cycles of length $2k$ each, and $4t$ vertex-disjoint paths of length $k-1$. Thus, \emph{$G$ is bipartite} and there exists a cut such that all edges in the graph cross that cut. The value of 
maximum cut in this case is $t\cdot 2k +4t \cdot (k-1)$, which equals $n-2n/3k = (n-n/k)+n/3k \geq (n-n/k) + \eps n$.

As such a $p$-pass streaming algorithm for $(1+\eps)$-approximation of maximum cut can be used as a distinguisher for the two cases of the graph $G$. By~\Cref{thm:ngc}, 
we have that that space complexity of the algorithm should be
\[
\Omega\paren{\frac1{p^5} \cdot (n/k)^{1-O(p/k)}} = \Omega\paren{\frac1{p^5} \cdot (\eps \cdot n )^{1-O(\eps \cdot p)}}.
\]

The fact that the lower bound holds on bounded-degree planar is simply because the distribution $\distNGC$ in~\Cref{distngc} we use here is supported on graphs which are disjoint-union of cycles and paths, and thus clearly are both 
bounded-degree and planar. 
\end{proof}

\paragraph{Optimality of~\Cref{thm:maxcut}?}~\Cref{thm:maxcut} provides the currently best multi-pass lower bound for $(1+\eps)$-approximation of maximum cut in \emph{any} family of graphs. We are not sure if it is possible to obtain
a $n^{o(1)}$-space algorithm with $O(1/\eps)$-passes (or for that matter, even independent of $n$) for this problem on general graphs and thus the lower bound can perhaps be improved further. We shall remark however that a $(1+\eps)$-approximation  of maximum cut is possible in a single pass
and $\Ot(n/\eps^2)$ space or even $o(n)$ space for sufficiently dense graphs~\cite{BhaskaraDV18} (by computing a cut sparsifier in streaming; see~\cite{AhnG09}), when we allow exponential time to the algorithms. 

On the other hand, we suspect that such an algorithm should be possible 
for planar graphs (or at least for the bounded-degree ones -- note that finding maximum cut in planar graphs is closely tied to the maximum weight matching problem and is also solvable in polynomial time~\cite{jHadlock75}); we leave this question 
as an interesting open problem for future work. 

\subsection{Maximum Acyclic Subgraph}
Given a directed graph $G(V,E)$, the maximum acyclic subgraph problem asks for an estimate to the \emph{size} of the largest acyclic subgraph in $G$ measured in the number of edges. 
This is a canonical CSP problem (alongside maximum cut) and has been studied in the streaming model in~\cite{GuruswamiT19,GuruswamiVV17,ChakrabartiG0V20,AssadiKSY20,ChouGV20}, with
the best lower bound of $\Omega(\sqrt{n})$ for single-pass algorithm with $<(7/8)$-approximation, and $\Omega(\log{(1/\eps)})$ pass lower bound for $(1+\eps)$-approximation algorithms with $n^{o(1)}$-space~\cite{AssadiKSY20}. 

We present the following lower bound for this problem:

\begin{theorem} \label{thm:acyclicsubgraph}
For any $\eps \in (0,1)$, any $p$-pass streaming algorithm for $(1+\eps)$-approximation of size of a largest acyclic subgraph on $n$-vertex directed graphs with probability  at least $2/3$ requires 
$\Omega\paren{\frac1{p^5} \cdot (\eps \cdot n )^{1-O(\eps \cdot p)}}$ space. Moreover, this lower bound continues to hold even on bounded-degree planar graphs 
and also implies that $\Omega(1/\eps)$ passes are needed for any $n^{o(1)}$-space algorithm. 
\end{theorem}
\begin{proof}
The proof is along the previous lines, except that we are going to apply it to the \emph{directed} version of the problem in~\Cref{cor:ngc-strong}. 
Consider a $\NGC$ instance $G$ on $n=6t \cdot k$ vertices drawn from $\distNGC$, for largest integer $k\leq \frac 1{6 \eps}$. In line with~\Cref{cor:ngc-strong}, we assume that $G$ is directed. 

Since cycles of $G$ are disjoint, the size of the largest acyclic subgraph of $G$ is exactly equal to the number of edges of $G$ minus its number of cycles. The number of edges in our graphs is $2t \cdot k + 4t \cdot (k-1)$ in both cases. 
The number of cycles is $2t$ in one case and $t$ in another. Thus, the size of largest acyclic subgraph in one case is $6t \cdot k - 6t = n-(n/k)$ and in the other case it is $6t \cdot k - 5t = n-(n/k)+(n/6k) = n-(n/k) + \eps \cdot n$. 

As such a $p$-pass streaming algorithm for size of a largest acyclic subgraph can be used as a distinguisher for the two cases of the directed graph $G$. By~\Cref{cor:ngc-strong}, 
we have that that space complexity of the algorithm should be
\[
\Omega\paren{\frac1{p^5} \cdot (n/k)^{1-O(p/k)}} = \Omega\paren{\frac1{p^5} \cdot (\eps \cdot n )^{1-O(\eps \cdot p)}}.
\]

The fact that the lower bound holds on bounded-degree planar is simply because the distribution $\distNGC$ in~\Cref{distngc} we use here is supported on graphs which are disjoint-union of cycles and paths, and thus clearly are both 
bounded-degree and planar. 
\end{proof}

\paragraph{Optimality of~\Cref{thm:acyclicsubgraph}?}~\Cref{thm:acyclicsubgraph} provides the currently best multi-pass lower bound for $(1+\eps)$-approximation of maximum acyclic graph in \emph{any} family of graphs. 
However, we are not aware of any algorithmic work on this problem in the streaming setting except for~\cite{ChakrabartiG0V20} who considered the closely related problem of feedback arc set: minimum number of edges that needs to be deleted
from the graph before making it acyclic (this number is equal to the number of edges minus the answer to our original problem). As such, at this point, we do not know much about the complexity of this problem and consequently optimality of~\Cref{thm:acyclicsubgraph} (note however that a $2$-approximation in $O(\log{n})$ space is trivial by returning half the number of edges).

\subsection{Property Testing: Connectivity, Bipartiteness, and Cycle-freeness} 

Given a graph property $P$ and an $\eps \in (0,1)$, an $\eps$-property tester for $P$ is an algorithm that decides whether an input $G$ has the property $P$ or is \emph{$\eps$-far} from having $P$. We define a graph $G$ as being $\eps$-far from the properties we consider as follows:

\begin{itemize}
\item \textbf {Connectivity:} If at least $\eps\cdot n$ edges need to be inserted to $G$ to make it connected, then $G$ is said to be $\eps$-far from being connected;
\item \textbf {Bipartiteness:} If at least $\eps\cdot n$ edges need to be deleted from $G$ to make it bipartite, then $G$ is said to be $\eps$-far from being bipartite;
\item \textbf {Cycle-freeness:} If at least $\eps\cdot n$ edges need to be deleted from $G$ to remove all its cycles, then $G$ is said to be $\eps$-far from being cycle-free.
\end{itemize}
Traditionally, these problems have been studied in the query complexity model, but more recently, they also received an extensive attention in the streaming model~\cite{HuangP16,CzumajFPS19,PengS18,MonemizadehMPS17}. 
In particular,~\cite{HuangP16} gave an upper bound of $n^{1-\Theta(\eps)}$-space and single-pass for the first two-problems and $n^{1-\Theta(\eps^2)}$ and single-pass for the latter problem on \emph{planar} graphs. 
From the lower bound perspective,~\cite{HuangP16} proved $n^{1-O(\eps)}$ space lower bounds for these problems in single-pass streams and~\cite{AssadiKSY20} proved that $n^{o(1)}$-space algorithms require $\Omega(\log{(1/\eps)})$ passes. 

We prove the following lower bound for these problems: 

\begin{theorem} \label{thm:propertytesting}
For any $\eps \in (0,1)$, any $p$-pass streaming algorithm which is a $\eps$-property tester for connectivity, bipartiteness, and cycle-freeness on $n$-vertex graphs with probability at least $2/3$ requires 
$\Omega\paren{\frac1{p^5} \cdot (\eps \cdot n )^{1-O(\eps \cdot p)}}$ space. Moreover, this lower bound continues to hold even on bounded-degree planar graphs 
and also implies that $\Omega(1/\eps)$ passes are needed for any $n^{o(1)}$-space algorithm. 
\end{theorem}

\begin{proof}
The proofs of each of these parts are very similar to the previous results of this section and thus we only briefly mention each one.

\paragraph{Connectivity:} The proof of this part is identical to that of~\Cref{thm:MST} with the difference that we no longer add edges of weight $W$. Thus, in one case the graph has $\eps n$ connected components and in the other case it is connected. 
The rest follows verbatim as ~\Cref{thm:MST}. 

\paragraph{Bipartiteness:} The proof of this part is identical to that of~\Cref{thm:maxcut}. The graphs in~\Cref{thm:maxcut} in one case miss $\eps n$ edge from any cut, thus need $\eps n$ of their edges removed to become bipartite, 
while in the other case they are bipartite.  The rest follows verbatim as ~\Cref{thm:maxcut}.

\paragraph{Cycle-freeness:} Consider a $\NGC$ instance $G$ on $n=6t \cdot k$ vertices drawn from $\distNGC$ for largest integer $k \leq \frac{1}{6\eps}$. We crucially use \Cref{cor:ngc-strong}  to perform this reduction.

Using~\Cref{cor:ngc-strong} allows us to assume that the algorithm for $G$ has the extra knowledge of the following: a collection
of $t$ tuples $(u^i_1,u^i_2,u^i_3,u^i_4)$ for $i \in [t]$ (covering the cycles) with the properties specified by~\Cref{cor:ngc-strong} (we do not need the extra knowledge of the paths here).  

Given a graph $G$ sampled from $\distNGC$, simply remove one arbitrary edge incident on $u^i_1$ for all $i \in [t]$. This way, $u^i_1$ cannot be part of any cycle in $G$. Now, in one case, $G$ still has $t$ edge-disjoint $k$-cycles and thus we need 
to remove $t = (n/6k) = \eps n$ edges from $G$ to make it cycle-free. In the other case, every cycle of $G$ has lost an edge and thus it is cycle-free. The lower bound now follows from~\Cref{cor:ngc-strong} as all our other lower bounds in this section. 
\end{proof}

\paragraph{Optimality of~\Cref{thm:propertytesting}?} The bounds obtained by our algorithm for connectivity and cycle-freeness are asymptotically optimal for \emph{any graph} by the reductions in~\cite{HuangP16} and 
our discussion for~\Cref{thm:MST}. In particular, both problems can be solved in $O(1/\eps)$ passes and $O_{\eps}(\polylog{(n)})$ space by simulating the algorithm of~\cite{ChazelleRT05} in the streaming model. 
For bipartiteness, we  are not aware of a non-trivial algorithm in arbitrary graphs. However, for \emph{planar} graphs, the approach of~\cite{HuangP16} based on the query algorithm of~\cite{CzumajMOS11} also immediately 
implies a tester in $O(1/\eps^2)$ passes and $O_{\eps}(\polylog{(n)})$ space for this problem. Bridging this gap remains an open question. 

\medskip

This concludes the list of all lower bounds in~\Cref{res:str}. We remark that for all these problems, no lower bound better than $\Omega(\log{(1/\eps)})$ passes for $n^{o(1)}$-space algorithms were known even for arbitrary graphs. Our results 
thus exponentially improves over prior work and in multiple cases lead to asymptotically optimal bounds as discussed above.

%% file: figs/mst1.tex

\begin{tikzpicture}
	

\tikzset{choose/.style={rectangle, draw, rounded corners=5pt, line width=1pt, inner xsep=5pt, inner ysep=2pt]}}
\tikzset{layer/.style={rectangle, rounded corners=8pt, draw, black, line width=1pt, inner sep=3pt]}}


\node[vertex] (v111){};
\node[vertex] [below=10pt of v111](v112){};
\node[vertex] [right=10pt of v111](v113){};
\node[vertex] [below=10pt of v113](v114){};

\foreach \i in {2,3,4}
{

\pgfmathtruncatemacro{\ip}{\i-1};

\node[vertex] [below=15pt of v1\ip2](v1\i1){};
\node[vertex] [below=10pt of v1\i1](v1\i2){};
\node[vertex] [right=10pt of v1\i1](v1\i3){};
\node[vertex] [below=10pt of v1\i3](v1\i4){};
}

\node[vertex] [right=30pt of v113](v211){};
\node[vertex] [below=10pt of v211](v212){};
\node[vertex] [right=10pt of v211](v213){};
\node[vertex] [below=10pt of v213](v214){};

\foreach \i in {2,3,4}
{

\pgfmathtruncatemacro{\ip}{\i-1};

\node[vertex] [below=15pt of v2\ip2](v2\i1){};
\node[vertex] [below=10pt of v2\i1](v2\i2){};
\node[vertex] [right=10pt of v2\i1](v2\i3){};
\node[vertex] [below=10pt of v2\i3](v2\i4){};
}
\begin{scope}[on background layer]
\foreach \i in {1,2,3,4}
{
	\node[choose, draw=none, fill=gray!25, fit=(v1\i3) (v1\i4) (v2\i1) (v2\i2)] {};
}
\end{scope}

\foreach \i in {1,2,3}
{
	\pgfmathtruncatemacro{\ip}{\i+1}

	\draw[blue, line width=1pt]
		(v1\i4) to (v2\ip1);
}

\foreach \i in {1,2,3,4}
{
	\draw[red, line width=3pt]
		(v1\i3) to (v2\i2);
}
\draw[blue, line width=1pt]
	(v144) to ($(v144)+(0pt,-15pt)$) to ($(v144)+(-40pt,-15pt)$) to ($(v144)+(-40pt,160pt)$) to ($(v211)+(0pt,12pt)$) to (v211);

\foreach \x in {1,2}
{
\foreach \i in {1,2,3,4}
{
\draw[black, line width=1pt]
	(v\x\i1) to (v\x\i2)
	(v\x\i2) to (v\x\i4)
	(v\x\i4) to (v\x\i3)
	(v\x\i1) to (v\x\i3);
}
}

\end{tikzpicture}

%% file: figs/mst2.tex

\begin{tikzpicture}
	

\tikzset{choose/.style={rectangle, draw, rounded corners=5pt, line width=1pt, inner xsep=5pt, inner ysep=2pt]}}
\tikzset{layer/.style={rectangle, rounded corners=8pt, draw, black, line width=1pt, inner sep=3pt]}}


\node[vertex] (v111){};
\node[vertex] [below=10pt of v111](v112){};
\node[vertex] [right=10pt of v111](v113){};
\node[vertex] [below=10pt of v113](v114){};

\foreach \i in {2,3,4}
{

\pgfmathtruncatemacro{\ip}{\i-1};

\node[vertex] [below=15pt of v1\ip2](v1\i1){};
\node[vertex] [below=10pt of v1\i1](v1\i2){};
\node[vertex] [right=10pt of v1\i1](v1\i3){};
\node[vertex] [below=10pt of v1\i3](v1\i4){};
}

\node[vertex] [right=30pt of v113](v211){};
\node[vertex] [below=10pt of v211](v212){};
\node[vertex] [right=10pt of v211](v213){};
\node[vertex] [below=10pt of v213](v214){};

\foreach \i in {2,3,4}
{

\pgfmathtruncatemacro{\ip}{\i-1};

\node[vertex] [below=15pt of v2\ip2](v2\i1){};
\node[vertex] [below=10pt of v2\i1](v2\i2){};
\node[vertex] [right=10pt of v2\i1](v2\i3){};
\node[vertex] [below=10pt of v2\i3](v2\i4){};
}
\begin{scope}[on background layer]
\foreach \i in {1,2,3,4}
{
	\node[choose, draw=none, fill=gray!25, fit=(v1\i3) (v1\i4) (v2\i1) (v2\i2)] {};
}
\end{scope}

\foreach \i in {1,2,3}
{
	\pgfmathtruncatemacro{\ip}{\i+1}

	\draw[blue, line width=1pt]
		(v1\i4) to (v2\ip1);
}

\foreach \i in {1,2,3,4}
{
	\draw[red, line width=3pt]
		(v1\i3) to (v2\i2);
}
\draw[blue, line width=1pt]
	(v144) to ($(v144)+(0pt,-15pt)$) to ($(v144)+(-40pt,-15pt)$) to ($(v144)+(-40pt,160pt)$) to ($(v211)+(0pt,12pt)$) to (v211);

\foreach \i in {1,2,3,4}
{
\draw[black, line width=1pt]
	(v1\i1) to (v1\i2)
	(v1\i2) to (v1\i4)
	(v1\i3) to (v2\i1)
	(v1\i1) to (v1\i3)
	(v2\i2) to (v2\i4)
	(v2\i4) to (v2\i3)
	(v2\i1) to (v2\i3)
	(v1\i4) to (v2\i2);

}

\end{tikzpicture}

%% file: 6pointer.tex

\section{A Streaming Lower Bound for Pointer Chasing}\label{sec:pc-lower} 

We present the proof of~\Cref{prop:pc-lower} in this section.

\begin{proposition*}[Restatement of~\Cref{prop:pc-lower}]
	Consider a $p$-pass $s$-space streaming algorithm for $\PC_{\pcw,\pcd}$ over random $(\pcw,\pcd)$-layered graphs with matchings $M_1,\ldots,M_{\pcd}$ given in the stream $M_{\pcd} \conc \cdots \conc M_1$. 
	For  $\gamma \in (0,1/2)$,  if the algorithm succeeds with probability at least $1/2+\gamma$ then either $p > \pcd-1$ or $s=\Omega\paren{{\frac{\gamma^4}{\pcd^5}} \cdot \pcw}$. 
\end{proposition*}

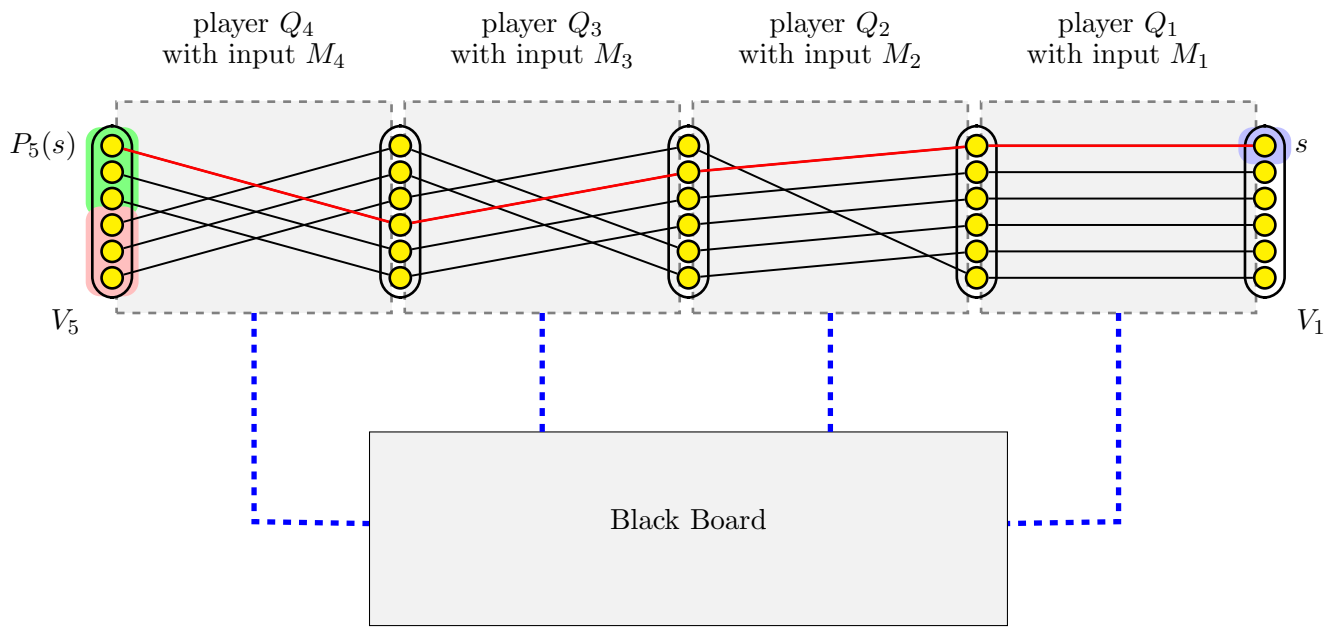
\begin{figure}[t!]	
\centering
\input{figs/game-pc}
\caption{An illustration of the communication game in the proof of~\Cref{prop:pc-lower} for $\pcd=4$. The $\pcw+1$ internal players corresponding to the stream $M_4$ (resp. $M_3$, $M_2$, $M_1$) are merged into a single player $Q_4$ (resp. $Q_3$, $Q_2$, $Q_1$), who gets the matching $M_4$ (resp. $M_3$, $M_2$, $M_1$) as input. The dashed (blue) lines draw the communicated messages between the players of the game and the blackboard.} 
\label{fig:game-pc}
\end{figure}

Recall that a $(\pcw,\pcd)$-layered graph consists of $\pcd+1$ layers $V_1,\ldots,V_{\pcd+1}$ of size $\pcw$ each and $\pcd$ matchings $M_1,\ldots,M_{\pcd}$ where $M_i$ is between $V_{i}$ and $V_{i+1}$. 
We shall consider each matching as a bijection $M_i: V_i \rightarrow V_{i+1}$, where for any vertex $v \in V_i$, $M_i(v)$ denotes the matched pair of $v$ in $V_{i+1}$. 

An instance of $\PC_{\pcw,\pcd}$ consists a $(\pcw,\pcd)$-layered graph, a fixed starting vertex $s \in V_1$, and a fixed equipartition $X,Y$ of $V_{\pcd+1}$. We define the  pointers $P_1(s),\ldots,P_{d+1}(s)$, 
where $P_i(s)$ is the unique vertex reachable in $V_i$ from $s$. This way, we have $P_{i+1}(s) = M_i(P_i(s))$. 
Initially, the pointer $P_1(s)$ is known and the goal is to decide whether $P_{\pcd+1}(s)$ belongs to $X$ or $Y$. 

We prove the lower bound more generally for $\pcd$-party communication protocols in the following: 

\begin{tbox}
\begin{enumerate}[label=$(\roman*)$,leftmargin=25pt]
	\item There are $\pcd$ players $Q_\pcd,\ldots,Q_1$ who receive input matchings $M_\pcd, \ldots , M_1$, respectively. 
	
	\item The players communicate with each other in \emph{rounds} via a \emph{blackboard}. In each round, 
	the players go in turn with $Q_\pcd$ writing a message on the board, followed by $Q_{\pcd-1}$, all the way to $Q_1$; these messages are visible to everyone (and are not altered or erased after  written). 

	\item For any player $Q_i$ and round $j$, we use $\Prot^j_i$ to denote the message written on the board by $Q_i$ in the $j$-th round. 	 
	\item At the end of $p$ rounds, player $Q_1$ outputs the answer on the board, i.e. the last message $\Prot^p_1$ is the answer to the problem.
	\item The \emph{communication cost} of a protocol is the \underline{maximum number of bits} communicated by any player in any round.

\end{enumerate}
\end{tbox}

\Cref{fig:game-pc} gives an illustration of this $\pcd$ player game. Recall that our~\Cref{def:str} corresponds to a $(\pcw \cdot \pcd)$-party communication protocol. Considering the input stream in~\Cref{prop:pc-lower} is $M_{\pcd} \conc \cdots \conc M_1$, a $p$-pass $s$-space streaming algorithm for this stream induces a $p$-round communication protocol with cost $s$ in the above game via a simulation argument, if we ``merge'' the $\pcw$ players of each matching $M_i$ in the streaming algorithm into a single player $Q_i$. In other words, we impose the space limitation only on consecutive players between two different matchings. We prove a lower bound for any $p$-round communication protocol with cost $s$ that solves the above game over the input distribution\footnote{Note that proving a lower bound in the message passing model is enough to prove a lower bound for a streaming algorithm, yet here we are proving it in the blackboard model for simplicity; this can only strengthen our result.}.

For the rest of the proof, let $\Prot$ denote a $(\pcd-1)$-round protocol with probability of success at least $1/2+\gamma$ over the distribution of inputs. We are going to lower bound the communication cost of this protocol. By the easy direction of Yao's minimax principle (an averaging argument), 
we can assume $\Prot$ is deterministic. We use $\Prot^{r}$ to denote all the messages communicated in round $r$, and $\Prot^{r}_{<i}$, $\Prot^{r}_{>i}$, and $\Prot^r_{-i}$ to denote, respectively, the messages communicated by players 
$Q_{1},\ldots,Q_{i-1}$, players $Q_{i+1},\ldots,Q_{\pcd}$, and players $Q_1,\ldots,Q_{i-1},Q_{i+1},\ldots,Q_{\pcd}$ in round $r$. Finally, for any round $r \in [p]$, define: 
\[
Z_r = (P_1(s),\ldots,P_{r+1}(s),\Prot^1,\ldots,\Prot^{r}), 
\]
namely, the information ``revealed'' to \emph{all} players \emph{after} round $r$

Throughout the proof, we use sans serif fonts to denote the random variables for the definitions above, e.g., $\rProt^r_i$ is a random variable for $\Prot^r_i$. Also, with a slight
abuse of notation, we may use random variables and their distributions interchangeably.  
We can now start the proof.

The main part of the proof is the following inductive lemma. Roughly speaking, it states that the distribution of $(r+2)$-th pointer after $r$ rounds of communication is  
close to its original distribution (the protocol is ``lagging behind''  the pointer it needs to chase -- note that the first pointer to chase is $P_2(s)$ as $P_1(s)$ is known globally), even conditioned on the information available to all the players. 
\begin{lemma}\label{lem:pc-ind}
Suppose communication cost of $\pi$ is $s < \paren{\frac1{100r} \cdot (\frac\gamma \pcd)^4 \cdot \pcw}$, then, for any $r \in [p]$, 
\begin{align}
	\Ex_{Z_r}\tvd{\rP_{r+2}(s)}{(\rP_{r+2}(s) \mid Z_r)} < \gamma \cdot \frac{r}{\pcd}.  \label{eq:pc-ind}
\end{align}
\end{lemma}
The proof of this lemma is by induction. The base case for $r=0$ holds because $Z_0 = P_1(s)$ which is fixed deterministically and thus conditioning on it is irrelevant,  and $\Prot^0$ is defined to be empty (as there is no communication yet); so both
LHS and RHS are zero. Suppose~\Cref{eq:pc-ind} is true up to round $r$ and we prove it for round $r$ itself.

We first show that in round $r$, the only player that we need to focus on is $Q_{r+1}$\footnote{Of course, player $Q_r$ would  reveal the next pointer; the interesting question we need to understand is that whether player $Q_{r+1}$ can also reveal the subsequent pointer in this round.}.
In other words, we can remove the contribution of all other players in the conditioning of~\Cref{eq:pc-ind}.

\begin{claim}\label{clm:pc-ind1}
	 For any choice of $Z_r = (Z_{r-1},\Prot^r_{-(r+1)},\Prot^r_{r+1},P_{r+1}(s))$, 
	 \[
	 \rP_{r+2}(s) \perp \rProt^r_{-(r+1)}  \mid Z_{r-1},\Prot^r_{r+1},P_{r+1}(s).
	 \] 
\end{claim}
\begin{proof}
	Define $M^{-}_{r+1}:= (M_1,\ldots,M_{r},M_{r+2},\ldots,M_{\pcd})$. We have, 
	\begin{align*}
		\mi{\rP_{r+2}(s)}{\rProt^r_{-(r+1)} \mid \rZ_{r-1},\rProt^r_{r+1},\rP_{r+1}(s)} &\leq \mi{\rM_{r+1}}{\rM^{-}_{r+1} \mid \rZ_{r-1},\rProt^r_{r+1},\rP_{r+1}(s)} \tag{by applying (\itfacts{data-processing}) twice as new variables
		determine old ones (conditionally)} \\
		&\leq \mi{\rM_{r+1}}{\rM^{-}_{r+1}}, 
	\end{align*}
	where the second inequality is a repeated application of~\Cref{prop:info-decrease} as each conditioned variable is a function of exactly one side of the mutual information term (conditioned on the remaining ones); this is exactly the same argument 
	as in~\Cref{lem:xor-independence} and we omit the tedious calculations. 
	
	We now  have $\mi{\rM_{r+1}}{\rM^{-}_{r+1}} = 0$ as in the input, the matchings are chosen independently. This means the first term 
	is also zero, proving the desired independence (by~\itfacts{info-zero}). 
\end{proof}
By~\Cref{clm:pc-ind1}, we can simplify the LHS of~\Cref{eq:pc-ind} to the following, i.e., get rid of the messages of all players other than $Q_{r+1}$.
\begin{align}
	\Ex_{Z_r}\tvd{\rP_{r+2}(s)}{(\rP_{r+2}(s) \mid Z_r)} = \Ex_{\substack{Z_{r-1},\Prot^r_{r+1}, \\P_{r+1}(s)}}\tvd{\rP_{r+2}(s)}{(\rP_{r+2}(s) \mid Z_{r-1},\Prot^r_{r+1},P_{r+1}(s))}. \label{eq:dist1}
\end{align}
We further simplify the RHS of~\Cref{eq:dist1} as follows, 
\begin{align}
	&\Ex_{\substack{Z_{r-1},\Prot^r_{r+1}, \\P_{r+1}(s)}}\tvd{\rP_{r+2}(s)}{(\rP_{r+2}(s) \mid Z_{r-1},\Prot^r_{r+1},P_{r+1}(s))} \notag \\
	&\quad = \Ex_{Z_{r-1},\Prot^r_{r+1}}\,\Ex_{v \sim \rP_{r+1}(s) \mid Z_{r-1},\Prot^r_{r+1}}\tvd{\rP_{r+2}(s)}{(\rM_{r+1}(v) \mid Z_{r-1},\Prot^r_{r+1},\rP_{r+1}(s)=v)} \tag{as $P_{r+2}(s) = M_{r+1}(P_{r+1}(s))$}  \\
	&\quad = \Ex_{Z_{r-1},\Prot^r_{r+1}}\,\Ex_{v \sim \rP_{r+1}(s) \mid Z_{r-1},\Prot^r_{r+1}}\tvd{\rP_{r+2}(s)}{(\rM_{r+1}(v) \mid Z_{r-1},\Prot^r_{r+1})} \label{eq:tag1} \\
	&\quad = \Ex_{Z_{r-1},\Prot^r_{r+1}}\,\Ex_{v \sim \rP_{r+1}(s) \mid Z_{r-1}}\tvd{\rP_{r+2}(s)}{(\rM_{r+1}(v) \mid Z_{r-1},\Prot^r_{r+1})}, \label{eq:tag2} 
\end{align}
where~\Cref{eq:tag1} and~\Cref{eq:tag2} hold by the following two claims, respectively. 

\begin{claim}\label{clm:tag1}
	For any choice of $Z_{r-1},\Prot^r_{r+1}$ and $v \in V_{r+1}$, 
	\[
	\rM_{r+1}(v) \perp (\rP_{r+1}(s)=v) \mid Z_{r-1},\Prot^r_{r+1}.
	\]
\end{claim}
\begin{proof}
	We have, 
	\begin{align*}
		\mi{\rM_{r+1}(v)}{\rP_{r+1}(s) \mid \rZ_{r-1},\rProt^r_{r+1}} &\leq \mi{\rM_{r+1}}{\rM_r \mid \rZ_{r-1},\rProt^r_{r+1}} 
		 \tag{by data processing inequality (\itfacts{data-processing}) as $\rM_r$ and $\rZ_{r-1}$ determine $\rP_{r+1}(s)$ (conditionally)} \\
		&\leq \mi{\rM_{r+1}}{\rM_r}, 
	\end{align*}
	where the inequality, and the rest of the proof is exactly as in~\Cref{clm:pc-ind1}. 
\end{proof}

\begin{claim}\label{clm:tag2}
	For any choice of $Z_{r-1}$, 
	\[
	\rP_{r+1}(s) \perp \rProt^r_{r+1} \mid Z_{r-1}.
	\]
\end{claim}
\begin{proof}
	We have, 
	\begin{align*}
		\mi{\rP_{r+1}(s)}{\rProt^r_{r+1} \mid \rZ_{r-1}} &\leq \mi{\rM_{r}}{\rM_{r+1} \mid \rZ_{r-1}} 
		 \tag{by applying (\itfacts{data-processing}) as $\rM_{r+1}$ determines $\rProt^r_{r+1}(s)$		  conditioned on $\rZ_{r-1}$} \\
		&\leq \mi{\rM_{r+1}}{\rM_r}, 
	\end{align*}
	where the inequality and the rest of the proof is exactly as in~\Cref{clm:pc-ind1}. 
\end{proof}

So far, we have proved that, 
\begin{align}
	\Ex_{Z_r}\tvd{\rP_{r+2}(s)}{(\rP_{r+2}(s) \mid Z_r)} \leq \Ex_{Z_{r-1},\Prot^r_{r+1}}\,\Ex_{v \sim \rP_{r+1}(s) \mid Z_{r-1}}\tvd{\rP_{r+2}(s)}{(\rM_{r+1}(v) \mid Z_{r-1},\Prot^r_{r+1})} \label{eq:pc-so-far}. 
\end{align}
We now apply the induction hypothesis to the expected term of $v \sim \rP_{r+1}(s) \mid Z_{r-1}$. In particular, let $\unif_{r+1}$ denote the uniform distribution over $V_{r+1}$ which is also $\rP_{r+1}(s)$ \emph{without} any conditioning. 
By~\Cref{fact:tvd-small}, we have that, 
\begin{align}
	&\Ex_{Z_{r-1},\Prot^r_{r+1}}\,\Ex_{v \sim \rP_{r+1}(s) \mid Z_{r-1}}\tvd{\rP_{r+2}(s)}{(\rM_{r+1}(v) \mid Z_{r-1},\Prot^r_{r+1})} \notag \\
	&\qquad \leq \Ex_{Z_{r-1},\Prot^r_{r+1}}\bracket{\Ex_{v \sim \unif_{r+1}}\tvd{\rP_{r+2}(s)}{(\rM_{r+1}(v) \mid Z_{r-1},\Prot^r_{r+1})} + \tvd{\unif_{r+1}}{(\rP_{r+1}(s) \mid Z_{r-1})}}\notag \\
	&\qquad = \Ex_{Z_{r-1},\Prot^r_{r+1}}\Ex_{v \sim \unif_{r+1}}\tvd{\rP_{r+2}(s)}{(\rM_{r+1}(v) \mid Z_{r-1},\Prot^r_{r+1})} + \gamma \cdot \frac{r-1}{\pcd}, \label{eq:pc-good-one}
\end{align}
by the induction hypothesis for round $r-1$ as $\unif_{r+1} = \rP_{r+1}(s)$. 

This conclude the first half of the proof of~\Cref{lem:pc-ind}. In the remaining half, we bound the first term in the RHS of~\Cref{eq:pc-good-one}. It is worth mentioning that in the first part, we did not deal with the fact that input of each player is a random matching as opposed to a random function. So, this part is more or less an extension of two-party approaches for pointer chasing, say~\cite{NisanW91,Yehudayoff16}, to the multi-party setting. From here however, we 
depart from the prior approaches to take into account the distribution of each  $M_i$ being a random matching as opposed to a random function. 

The following claim shows that we can further narrow down the task to player $Q_{r+1}$ 
by removing everything on the RHS that is not function of this player even from prior rounds. Formally, 

\begin{claim}\label{clm:qr+1}
	For any choice of $(\Prot^{1}_{r+1},\ldots,\Prot^r_{r+1})$ (messages of $Q_{r+1}$ so far), and any $v \in V_{r+1}$, 
	\[
		\rM_{r+1}(v) \perp \rZ_{r-1} \mid (\Prot^{1}_{r+1},\ldots,\Prot^r_{r+1}). 
	\]
\end{claim}
\begin{proof}
	We have, 
	\begin{align*}
		\mi{\rM_{r+1}(v)}{\rZ_{r-1} \mid \rProt^1_{r+1},\ldots,\rProt^r_{r+1}} &= \mi{\rM_{r+1}(v)}{\rProt^1_{-(r+1)},\ldots,\rProt^r_{-(r+1)},\rP_1(s),\ldots,\rP_{r}(s) \mid \rProt^1_{r+1},\ldots,\rProt^r_{r+1}} 
		\tag{these are the remaining variables in $\rZ_{r-1}$ after conditioning} \\
		&\leq \mi{\rM_{r+1}}{\rM_{-(r+1)} \mid \rProt^1_{r+1},\ldots,\rProt^r_{r+1}} 
		\tag{by data processing inequality (\itfacts{data-processing}) as new variables determine old ones} \\
		&\leq \mi{\rM_{r+1}}{\rM_{-(r+1)} \mid \rProt^1_{>r+1},\rProt^1_{r+1},\ldots,\rProt^r_{r+1}} 
		\tag{by~\Cref{prop:info-increase} as $\rM_{r+1} \perp \rProt^1_{>r+1}$ by the independence of players' inputs} \\
		&\leq \mi{\rM_{r+1}}{\rM_{-(r+1)} \mid \rProt^1_{<r+1},\rProt^1_{>r+1},\rProt^1_{r+1},\ldots,\rProt^r_{r+1}} 
		\tag{by~\Cref{prop:info-increase} as $\rM_{r+1} \perp \rProt^1_{<r+1} \mid \rProt^1_{>r+1},\rProt^1_{r+1}$ by the independence of players' inputs} \\
		&= \mi{\rM_{r+1}}{\rM_{-(r+1)} \mid \rProt^1, \rProt^2_{r+1}, \ldots,\rProt^r_{r+1}} \tag{as $\rProt^1_{<r+1},\rProt^1_{>r+1},\rProt^1_{r+1} = \rProt^1$} \\
		&\leq \mi{\rM_{r+1}}{\rM_{-(r+1)} \mid \rProt^1, \ldots, \rProt^r} \tag{using exactly the same argument as above} \\
		&\leq \mi{\rM_{r+1}}{\rM_{-(r+1)}}, 
	\end{align*}
	where the inequality and the rest of the proof is exactly as in~\Cref{clm:pc-ind1}. 
\end{proof}
\noindent
By~\Cref{clm:qr+1}, we can simplify the first term of RHS of~\Cref{eq:pc-good-one} as follows:
\begin{align*}
	&\Ex_{Z_{r-1},\Prot^r_{r+1}}\Ex_{v \sim \unif_{r+1}}\tvd{\rP_{r+2}(s)}{(\rM_{r+1}(v) \mid Z_{r-1},\Prot^r_{r+1})} \\
	&\qquad = \Ex_{Z_{r-1},\Prot^r_{r+1}}\Ex_{v \sim \unif_{r+1}}\tvd{\rP_{r+2}(s)}{(\rM_{r+1}(v) \mid \Prot^1_{r+1},\ldots,\Prot^r_{r+1})} \\
	&\qquad =  \Ex_{\Prot^1_{r+1},\ldots,\Prot^r_{r+1}}\Ex_{v \sim \unif_{r+1}}\tvd{\rP_{r+2}(s)}{(\rM_{r+1}(v) \mid \Prot^1_{r+1},\ldots,\Prot^r_{r+1})}.
\end{align*}
Notice that now the RHS of the above term is purely a function of player $Q_{r+1}$. Abstractly, the question is to understand how much $Q_{r+1}$ can change the distribution of a \emph{random} edge in their input by their messages in the first $r$ rounds. 
Intuitively, this should not be much if the size of messages are small. We now formalize this in the following lemma. 

\begin{lemma}\label{lem:clutter}
	Assuming  communication cost of $\pi$ is $s < \paren{\frac1{100r} \cdot (\frac\gamma \pcd)^4 \cdot \pcw}$, 
	\[
		\Ex_{\Prot^1_{r+1},\ldots,\Prot^r_{r+1}}\Ex_{v \sim \unif_{r+1}}\tvd{\rP_{r+2}(s)}{(\rM_{r+1}(v) \mid \Prot^1_{r+1},\ldots,\Prot^r_{r+1})} < \gamma \cdot \frac{1}{\pcd}. 
	\]
\end{lemma}
\begin{proof}
	To avoid the clutter, throughout this proof (only), we are going to drop all the subscripts and superscripts on $r$ and in particular denote: 
	\begin{align*}
		Q &:= Q_{r+1}  &\text{(the player $Q_{r+1}$ we are focusing on)} \\
		\Prot &:= \Prot^1_{r+1},\ldots,\Prot^r_{r+1}  &\text{(the messages of player $Q_{r+1}$)} \\
		M &:= M_{r+1} &\text{(the input of player $Q_{r+1}$ which is a random matching)} \\
		D &:= V_{r+1}&\text{(domain of $M_{r+1}$ interpreted as a function)} \\
		R &:= V_{r+2}&\text{(range of $M_{r+1}$ interpreted as a function)} \\
		\unif_D &= \unif_{r+1} &\text{(uniform distribution on the domain $D$)} \\
		\unif_R &= \rP_{r+2}(s) &\text{(uniform distribution on the range $R$)}. 
	\end{align*}
	Thus, with this notation, our goal is to prove that 
	\[
		\Ex_{\Prot}\,\Ex_{v \sim \unif_{D}}\tvd{\unif_R}{(\rM(v) \mid \Prot)} \leq \gamma \cdot \frac{1}{\pcd}. 
	\]
	Firstly, by Pinsker's inequality (\Cref{fact:pinskers}), 
\begin{align}
	&\Ex_{\Prot}\,\Ex_{v \sim \unif_{D}}\tvd{\unif_R}{(\rM(v) \mid \Prot)} = \Ex_{\Prot}\,\Ex_{v \sim \unif_{D}}\tvd{(\rM(v) \mid \Prot)}{\rM(v)} \tag{symmetry of $\tvd{}{}$ and since unconditional distribution of $\rM(v) = \unif_R$ for all $v$} \\
	&\hspace{2cm}\leq \Ex_{\Prot}\,\Ex_{v \sim \unif_{D}}\sqrt{\frac12 \cdot \kl{\rM(v) \mid \Prot}{\rM(v)}} \tag{by~\Cref{fact:pinskers}} \\
	&\hspace{2cm}\leq \sqrt{\frac12 \Ex_{\Prot}\,\Ex_{v \sim \unif_{D}}\kl{\rM(v) \mid \Prot}{\rM(v)}} \tag{by concavity of $\sqrt{\cdot}$ and Jensen's inequality} \\
	&\hspace{2cm}= \sqrt{\frac12 \cdot \frac{1}{\card{D}} \cdot \sum_{v \in D} \Ex_{\Prot}\,\kl{\rM(v) \mid \Prot}{\rM(v)}} \tag{as $\unif_D$ is uniform distribution  over $D$} \\
	&\hspace{2cm}= \sqrt{\frac12 \cdot \frac{1}{\card{D}} \cdot \sum_{v \in D} \mi{\rM(v)}{\rProt}}, \label{eq:pc-ent}
\end{align}
where the last inequality  is by~\Cref{fact:kl-info}. We thus need to bound the mutual information in the RHS above. 
For any $v \in D$, 
\begin{align}
	\mi{\rM(v)}{\rProt} = \en{\rM(v)} - \en{\rM(v) \mid \rProt} = \log{\card{D}} -  \en{\rM(v) \mid \rProt}, \label{eq:pc-above}
\end{align}
as $\rM(v)$ is uniform over $D$ and thus we can apply~\itfacts{uniform}. Let us take a quick detour that would help put the rest of the argument in some context. 

\begin{remark*}
It is tempting now to apply the sub-additivity of entropy (\itfacts{sub-additivity}) to have,
	\[
	\sum_{v} \en{\rM(v) \mid \rProt} \geq \en{\rM \mid \rProt} \geq \en{\rM} - \en{\rProt}; 
	\]
	the problem with this approach is that $\rM$ is a random matching and not a random function and thus $\en{\rM} = \log{(\card{D}!)} = \card{D}\log{\card{D}} - \Theta(\card{D})$ by Stirling's approximation. This means that ignoring even subtraction of 
	$\en{\rProt}$, which means even if $Q$ does not communicate(!), the bound we get on the average entropy is $\log{\card{D}} - \Theta(1)$; plugging this in~\Cref{eq:pc-above} will bound the information term by $\Theta(1)$ which is too much (and quite loose 	as without $\rProt$, this term should be zero). 
\end{remark*}
	
	We will use Entropy Subset Inequality\footnote{We note that this part of the proof can also be done using a similar argument in~\cite{AssadiKSY20} on high-entropy random permutations, although we found the current approach simpler and more direct.} of~\Cref{fact:esi} to bound~\Cref{eq:pc-above}. Let $\beta \leq \card{D}/2$ be an integer to be determined later. For any $S \subseteq D$, define $\rM_S := \set{\rM(v) \mid v \in S}$.
	By Entropy Subset Inequality (\Cref{fact:esi}),
	\begin{align*}
		\frac{1}{\card{D}} \sum_{v \in D} \en{\rM(v) \mid \rProt_{r}} &\geq \frac{1}{\binom{\card{D}}{\beta}} \sum_{S \subseteq D: \card{S}=\beta} \frac{\en{\rM_S \mid \rProt}}{\beta} \\
		&\geq  \frac{1}{\binom{\card{D}}{\beta}} \sum_{S \subseteq D: \card{S}=\beta} \frac{\en{\rM_S} - \en{\rProt}}{\beta} \\  \tag{by the chain rule of entropy (\itfacts{ent-chain-rule})} \\
		&= \frac{\log{\paren{\frac{\card{D}!}{(\card{D}-\beta)!}}} - \en{\rProt}}{\beta} \tag{because $\rM_S$ is a uniform partial matching with one endpoint in $S$} \\
		&\geq \log{(\card{D}-\beta)}-\frac{\en{\rProt}}{\beta} \tag{as $\frac{a!}{(a-b)!} \geq (a-b)^{b}$} \\
		&\geq  \log{(\card{D}-\beta)}-\frac{r \cdot s}{\beta}. \tag{as $\Prot$ consists of $r$ messages of size $s$ bits each and by using \itfacts{uniform}}
	\end{align*}
	
	The rest of the proof is simply calculations. By plugging this bound in~\Cref{eq:pc-above}, we have, 
	\begin{align*}
		\frac1{\card{D}} \cdot \sum_{v \in D} \mi{\rM_r(v)}{\rProt} &\leq \log{\card{D}}- \log{(\card{D}-\beta)} + \frac{r \cdot s}{\beta} \\
		&= \log{\paren{1+\frac{\beta}{\card{D}-\beta}}} + \frac{r \cdot s}{\beta} \\
		&\leq \frac{\beta}{\card{D}-\beta} \cdot \log{(e)} + \frac{r \cdot s}{\beta} \tag{as $1+x \leq e^x$} \\
		&\leq \frac{4\beta}{\card{D}} + \frac{r\cdot s}{\beta} \tag{as $\beta \leq \card{D}/2$ and $\log(e) \leq 2$}. 
	\end{align*} 
	We can now set $\beta = \sqrt{r \cdot s \cdot \card{D}} < \card{D}/2$ and obtain that 
	\[
		\frac1{\card{D}} \cdot \sum_{v \in D} \mi{\rM_r(v)}{\rProt} \leq 5\sqrt{\frac{r \cdot s}{\card{D}}}. 
	\]
	Plugging in further in~\Cref{eq:pc-ent} proves that, 
	\[
		\Ex_{\Prot}\,\Ex_{v \sim \unif_{D}}\tvd{\unif_R}{(\rM(v) \mid \Prot)} \leq \sqrt{\frac12 \cdot 5\sqrt{\frac{r \cdot s}{\card{D}}}} < \gamma \cdot \frac1\pcd, 
	\]
	as $s < \paren{\frac1{100r} \cdot (\frac\gamma \pcd)^4 \cdot \pcw}$ and $\card{D} = \pcw$. This concludes the proof. \Qed{lem:clutter} 
	
\end{proof}

By using~\Cref{lem:clutter} and~\Cref{clm:qr+1} in~\Cref{eq:pc-good-one}, we obtain that, 
\[
	\Ex_{Z_{r-1},\Prot^r_{r+1}}\,\Ex_{v \sim \rP_{r+1}(s) \mid Z_{r-1}}\tvd{\rP_{r+2}(s)}{(\rM_{r+1}(v) \mid Z_{r-1},\Prot^r_{r+1})}< \gamma \cdot \frac1\pcd + \gamma \cdot \frac{r-1}\pcd = \gamma \cdot \frac r \pcd. 
\]
Plugging this bound into~\Cref{eq:pc-so-far}, we get that, 
\[
	\Ex_{Z_r}\tvd{\rP_{r+2}(s)}{(\rP_{r+2}(s) \mid Z_r)} < \gamma \cdot \frac{r}{\pcd},
\]
finalizing the proof of the induction step and thus~\Cref{lem:pc-ind}. 

We now conclude the proof of~\Cref{prop:pc-lower} as follows.

\begin{proof}[Proof of~\Cref{prop:pc-lower}]
	If $p > \pcd-1$ we are already done so let us assume that the total number of rounds of the protocol (or passes of the streaming algorithm) is $p=\pcd-1$. 
	By~\Cref{lem:pc-ind} for the last round $r=p = \pcd-1$, 
	\begin{align*}
		\Ex_{Z_p}\tvd{\rP_{\pcd+1}(s)}{(\rP_{\pcd+1}(s) \mid Z_p)} < \gamma.
	\end{align*}
	On the other hand, the last message of the protocol $\prot$, included in $Z_p$, specifies the answer, which depends on $\rP_{\pcd+1}$. Let $O(Z_p) \in \set{X,Y}$ denote the answer of the protocol.
	As such, 
	\begin{align*}
		\Pr\paren{\text{$\prot$ is correct}} &= \Ex_{Z_p}\Pr_{P_{\pcd+1} \mid Z_p}\paren{P_{d+1}(s) \in O(Z_p)} \\
		&=  \Ex_{Z_p}\Pr_{P_{\pcd+1} \mid Z_p} \paren{\text{$P_{\pcd+1}(s) \in$ a fixed choice of $X$ or $Y$}} \tag{because $O(Z_p)$ is deterministically fixed} \\
		&\leq \Ex_{Z_p} \bracket{\Pr_{P_{\pcd+1}} \paren{\text{$P_{\pcd+1}(s) \in$ a fixed choice of $X$ or $Y$}} + \tvd{\rP_{\pcd+1}(s)}{(\rP_{\pcd+1}(s) \mid Z_p)}}  \tag{by~\Cref{fact:tvd-small}} \\
		&= \frac12 + \Ex_{Z_p}\tvd{\rP_{\pcd+1}(s)}{(\rP_{\pcd+1}(s) \mid Z_p)} \tag{as $X,Y$ is an equipartition of $V_{\pcd+1}$ and $P_{\pcd+1}(s)$ is uniform over $V_{\pcd+1}$}\\
		&< \frac12 + \gamma, 
	\end{align*}
	by the first equation above. This concludes the proof. 
\end{proof}

%% file: figs/game-pc.tex

\begin{tikzpicture}

\tikzset{choose/.style={rectangle, draw, rounded corners=5pt, line width=1pt, inner xsep=5pt, inner ysep=2pt]}}
\tikzset{layer/.style={rectangle, rounded corners=8pt, draw, black, line width=1pt, inner sep=3pt]}}

\tikzset{Qp/.style={rectangle, draw, dashed, line width=1pt, gray, minimum width=104pt, fill=gray!10, minimum height=80pt, inner sep=5pt]}}

\node[vertex] (v11){};
\foreach \j in {2,...,6}
{
	\pgfmathtruncatemacro{\jp}{\j-1};
	\node[vertex] (v1\j) [below=1pt of v1\jp]{};
}
\node[layer] (V1) [fit=(v11) (v16)]{};
\foreach \i in {2,3,4,5}
{
	\pgfmathtruncatemacro{\ip}{\i-1};
	\node[vertex] (v\i1) [right=100pt of v\ip1]{};
	\foreach \j in {2,...,6}
	{
		\pgfmathtruncatemacro{\jp}{\j-1};
		\node[vertex] (v\i\j) [below=1pt of v\i\jp]{};
	}
	\node[layer] (V\i) [fit=(v\i1) (v\i6)]{};
}
\foreach \i in {1,...,4}
{
	\pgfmathtruncatemacro{\ip}{\i+1};
	\foreach \j in {1,...,6}
	{
		\pgfmathtruncatemacro{\jp}{Mod(\i+\j+1,6)+1};
		
		\draw[black, line width=0.75pt]
			(v\i\j) to (v\ip\jp);
	}
}

\begin{scope}[on background layer]

\node[Qp] (Q1)[above right=-37pt and -2pt of v14] {};
\node[Qp] (Q2)[above right=-37pt and -2pt of v24] {};
\node[Qp] (Q3)[above right=-37pt and -2pt of v34] {};
\node[Qp] (Q4)[above right=-37pt and -2pt of v44] {};

\end{scope}

\begin{scope}[on background layer]
\node [left=5pt of v11]{$P_5(s)$};
\node[choose, blue!25, fill=blue!25] (S) [fit=(v51) (v51)]{};

\node[choose, green!50, fill=green!50] (X) [fit=(v11) (v13)]{};
\node[choose, red!25, fill=red!25] (Y) [fit=(v14) (v16)]{};
\node [right=3pt of v51] {$s$};
\node [below left=5pt and 5pt of v16] {$V_5$};
\node [below right=5pt and 5pt of v56] {$V_1$};
\end{scope}

\draw[red, line width=1pt]
	(v11) to (v24)
	(v24) to (v32)
	(v32) to (v41)
	(v41) to (v51);

\node[rectangle, draw, fill=gray!10, minimum width=60pt] (bb) [below=120pt of V3][fit=(V2) (V4)] {Black Board};

\draw[line width=2pt, blue, dashed]
	($(Q1) + (0pt,-40pt)$) to ($(Q1) + (0pt,-119pt)$)to (bb);
	\draw[line width=2pt, blue, dashed]
	($(Q2) + (0pt,-40pt)$) to ($(Q2) + (0pt,-85pt)$);

	\draw[line width=2pt, blue, dashed]
	($(Q3) + (0pt,-40pt)$) to ($(Q3) + (0pt,-85pt)$);

	\draw[line width=2pt, blue, dashed]
	($(Q4) + (0pt,-40pt)$) to ($(Q4) + (0pt,-119pt)$) to (bb);

\node [above=20pt of Q1]{player $Q_4$};
\node [above=8pt of Q1]{with input $M_4$};
\node [above=20pt of Q2]{player $Q_3$};
\node [above=8pt of Q2]{with input $M_3$};
\node [above=20pt of Q3]{player $Q_2$};
\node [above=8pt of Q3]{with input $M_2$};
\node [above=20pt of Q4]{player $Q_1$};
\node [above=8pt of Q4]{with input $M_1$};

\end{tikzpicture}

%% file: appendixA-info.tex

\section{Basic Tools From Information Theory}\label{sec:info}

We now briefly introduce some definitions and facts from information theory that are needed in this paper. We refer the interested reader to the text by Cover and Thomas~\cite{CoverT06} for an excellent introduction to this field. 

For a random variable $\rA$, we use $\supp{\rA}$ to denote the support of $\rA$ and $\distribution{\rA}$ to denote its distribution. 
When it is clear from the context, we may abuse the notation and use $\rA$ directly instead of $\distribution{\rA}$, for example, write 
$A \sim \rA$ to mean $A \sim \distribution{\rA}$, i.e., $A$ is sampled from the distribution of random variable $\rA$.

We denote the \emph{Shannon Entropy} of a random variable $\rA$ by
$\en{\rA}$, which is defined as: 
\begin{align}
	\en{\rA} := \sum_{A \in \supp{\rA}} \Pr\paren{\rA = A} \cdot \log{\paren{1/\Pr\paren{\rA = A}}} \label{eq:entropy}
\end{align} 
\noindent
The \emph{conditional entropy} of $\rA$ conditioned on $\rB$ is denoted by $\en{\rA \mid \rB}$ and defined as:
\begin{align}
\en{\rA \mid \rB} := \Ex_{B \sim \rB} \bracket{\en{\rA \mid \rB = B}}, \label{eq:cond-entropy}
\end{align}
where 
$\en{\rA \mid \rB = B}$ is defined in a standard way by using the distribution of $\rA$ conditioned on the event $\rB = B$ in Eq~(\ref{eq:entropy}).

The \emph{mutual information} of two random variables $\rA$ and $\rB$ is denoted by
$\mi{\rA}{\rB}$ and  defined as:
\begin{align}
\mi{\rA}{\rB} := \en{A} - \en{A \mid  B} = \en{B} - \en{B \mid  A}. \label{eq:mi}
\end{align}
\noindent
The \emph{conditional mutual information} $\mi{\rA}{\rB \mid \rC}$ is $\en{\rA \mid \rC} - \en{\rA \mid \rB,\rC}$ and hence by linearity of expectation:
\begin{align}
	\mi{\rA}{\rB \mid \rC} = \Ex_{C \sim \rC} \bracket{\mi{\rA}{\rB \mid \rC = C}}. \label{eq:cond-mi}
\end{align}

\subsection{Useful Properties of Entropy and Mutual Information}\label{sec:prop-en-mi}

We use the following basic properties of entropy and mutual information throughout. 

\begin{fact}[cf.~\cite{CoverT06}]\label{fact:it-facts}
  Let $\rA$, $\rB$, $\rC$, and $\rD$ be four (possibly correlated) random variables.
   \begin{enumerate}
  \item \label{part:uniform} $0 \leq \en{\rA} \leq \log{\card{\supp{\rA}}}$. The right equality holds
    iff $\distribution{\rA}$ is uniform.
  \item \label{part:info-zero} $\mi{\rA}{\rB}[\rC] \geq 0$. The equality holds iff $\rA$ and
    $\rB$ are \emph{independent} conditioned on $\rC$.
  \item \label{part:cond-reduce} \emph{Conditioning on a random variable reduces entropy}:
    $\en{\rA \mid \rB,\rC} \leq \en{\rA \mid  \rB}$.  The equality holds iff $\rA \perp \rC \mid \rB$.
    \item \label{part:sub-additivity} \emph{Subadditivity of entropy}: $\en{\rA,\rB \mid \rC}
    \leq \en{\rA \mid C} + \en{\rB \mid  \rC}$.
   \item \label{part:ent-chain-rule} \emph{Chain rule for entropy}: $\en{\rA,\rB \mid \rC} = \en{\rA \mid \rC} + \en{\rB \mid \rC,\rA}$.
  \item \label{part:chain-rule} \emph{Chain rule for mutual information}: $\mi{\rA,\rB}{\rC \mid \rD} = \mi{\rA}{\rC \mid \rD} + \mi{\rB}{\rC \mid  \rA,\rD}$.
  \item \label{part:data-processing} \emph{Data processing inequality}: for a deterministic function $f(\rA)$, $\mi{f(\rA)}{\rB \mid \rC} \leq \mi{\rA}{\rB \mid \rC}$. 
   \end{enumerate}
\end{fact}

We also use the following generalization of sub-additivity of entropy. 

\begin{fact}[Entropy Subset Inequality~\cite{Te78}]\label{fact:esi}
	For any set of $n$ random variables $\rX_1,\ldots,\rX_n$ and set $S \subseteq [n]$, we define $\rX_S := \set{\rX_i \mid i \in S}$. For every $k \in [n]$, define: 
	\[
		\HH^{(k)}(\rX) := \frac{1}{\binom{n}{k}} \sum_{S \subseteq [n]: \card{S}=k} \frac{\en{\rX_S}}{k}.
	\]
	Then, ${\HH^{(1)}(\rX)} \geq \cdots \geq {\HH^{(n)}(\rX)}$. This equation also holds for conditional entropy. 
\end{fact}

\noindent

We also use the following two standard propositions, regarding the effect of conditioning on mutual information.

\begin{proposition}\label{prop:info-increase}
  For random variables $\rA, \rB, \rC, \rD$, if $\rA \perp \rD \mid \rC$, then, 
  \[\mi{\rA}{\rB \mid \rC} \leq \mi{\rA}{\rB \mid  \rC,  \rD}.\]
\end{proposition}
 \begin{proof}
  Since $\rA$ and $\rD$ are independent conditioned on $\rC$, by
  \itfacts{cond-reduce}, $\HH(\rA \mid  \rC) = \HH(\rA \mid \rC, \rD)$ and $\HH(\rA \mid  \rC, \rB) \ge \HH(\rA \mid  \rC, \rB, \rD)$.  We have,
	 \begin{align*}
	  \mi{\rA}{\rB \mid  \rC} &= \HH(\rA \mid \rC) - \HH(\rA \mid \rC, \rB) = \HH(\rA \mid  \rC, \rD) - \HH(\rA \mid \rC, \rB) \\
	  &\leq \HH(\rA \mid \rC, \rD) - \HH(\rA \mid \rC, \rB, \rD) = \mi{\rA}{\rB \mid \rC, \rD}. \qed
	\end{align*}
	
\end{proof}

\begin{proposition}\label{prop:info-decrease}
  For random variables $\rA, \rB, \rC,\rD$, if $ \rA \perp \rD \mid \rB,\rC$, then, 
  \[\mi{\rA}{\rB \mid \rC} \geq \mi{\rA}{\rB \mid \rC, \rD}.\]
\end{proposition}
 \begin{proof}
 Since $\rA \perp \rD \mid \rB,\rC$, by \itfacts{cond-reduce}, $\HH(\rA \mid \rB,\rC) = \HH(\rA \mid \rB,\rC,\rD)$. Moreover, since conditioning can only reduce the entropy (again by \itfacts{cond-reduce}), 
  \begin{align*}
 	\mi{\rA}{\rB \mid  \rC} &= \HH(\rA \mid \rC) - \HH(\rA \mid \rB,\rC) \geq \HH(\rA \mid \rD,\rC) - \HH(\rA \mid \rB,\rC) \\
	&= \HH(\rA \mid \rD,\rC) - \HH(\rA \mid \rB,\rC,\rD) = \mi{\rA}{\rB \mid \rC,\rD}.  \qed
 \end{align*}
 
\end{proof}

\subsection{Measures of Distance Between Distributions}\label{sec:prob-distance}

We will use the following two standard measures of distance (or divergence) between distributions. 

\paragraph{KL-divergence.} For two distributions $\mu$ and $\nu$, the \emph{Kullback-Leibler divergence} between $\mu$ and $\nu$ is denoted by $\kl{\mu}{\nu}$ and defined as: 
\begin{align}
\kl{\mu}{\nu}:= \Ex_{a \sim \mu}\Bracket{\log\frac{\Pr_\mu(a)}{\Pr_{\nu}(a)}}. \label{eq:kl}
\end{align}
We note that KL-divergence between two distributions is always non-negative. 

The following states the relation between mutual information and KL-divergence. 
\begin{fact}\label{fact:kl-info}
	For random variables $\rA,\rB,\rC$, 
	\[\mi{\rA}{\rB \mid \rC} = \Ex_{(b,c) \sim {(\rB,\rC)}}\Bracket{ \kl{\distribution{\rA \mid \rB=b,\rC=c}}{\distribution{\rA \mid \rC=c}}}.\] 
\end{fact}

We use the following standard facts about KL-divergence. 
\begin{fact}\label{fact:kl-chain-rule}
  Let $\mu(\rX,\rY)$ and $\nu(\rX,\rY)$ be two distributions for random variables $\rX,\rY$. Then, 
\[
	\kl{\mu(\rX,\rY)}{\nu(\rX,\rY)} = \kl{\mu(\rX)}{\nu(\rY)} + \Exp_{x \sim \mu(\rX)} \kl{\mu(\rY \mid \rX=x)}{\nu(\rY \mid \rX=x)}. 
\]
\end{fact}
\begin{fact}\label{fact:kl-event}
  Let $\mu$ be any distribution and $\event$ be any event fixed by $\mu$.  Then, 
\[
	\kl{\mu \mid \event}{\mu} = \log{\paren{\frac{1}{\Pr_{\mu}(W)}}}. 
\]
\end{fact}

\paragraph{Total variation distance.} We denote the total variation distance between two distributions $\mu$ and $\nu$ on the same 
support $\Omega$ by $\tvd{\mu}{\nu}$, defined as: 
\begin{align}
\tvd{\mu}{\nu}:= \max_{\Omega' \subseteq \Omega} \paren{\mu(\Omega')-\nu(\Omega')} = \frac{1}{2} \cdot \sum_{x \in \Omega} \card{\mu(x) - \nu(x)}.  \label{eq:tvd}
\end{align}
\noindent
We use the following basic properties of total variation distance. 
\begin{fact}\label{fact:tvd-small}
	Suppose $\mu$ and $\nu$ are two distributions for $\event$, then, 
	$
	\Pr_{\mu}(\event) \leq \Pr_{\nu}(\event) + \tvd{\mu}{\nu}.
$
\end{fact}

\begin{fact}\label{fact:tvd-sample}
	Suppose $\mu$ and $\nu$ are two distributions over the same support $\Omega$; then, given one sample $s$ from either $\mu$ or $\nu$, the probability we can decide whether $s$ came from $\mu$ or $\nu$ 
	is $\frac12 + \frac12\cdot\tvd{\mu}{\nu}$. 
\end{fact}

The following Pinsker's inequality bounds the total variation distance between two distributions based on their KL-divergence, 

\begin{fact}[Pinsker's inequality]\label{fact:pinskers}
	For any distributions $\mu$ and $\nu$, 
	$
	\tvd{\mu}{\nu} \leq \sqrt{\frac{1}{2} \cdot \kl{\mu}{\nu}}.
	$ 
\end{fact}

We shall also use the following extension of Pinsker's inequality for larger values of KL-divergence that do not imply any meaningful bounds using only Pinsker's inequality. 

\begin{proposition}[c.f.{~\cite[p. 88-89]{paraEstimation}}]\label{prop:pinsker} 
For any distributions $\mu$ and $\nu$, 
\begin{align*}
\tvd{\mu}{\nu} \le 1 - \dfrac{1}{2} \exp\paren{-\kl{\mu}{\nu}}. 
\end{align*}
\end{proposition} 

\subsection{XOR and Biases}\label{sec:xor-bias} 

For a $0/1$ random variable $\rX$, we define the \emph{bias} of $\rX$ as
\[
	\bias(\rX) := {\card{\Pr(\rX=0)-\Pr(\rX=1)}} = 2\argmax_{x \in \set{0,1}} \Pr(\rX=x) - 1,
\]
which measures how much $\rX$ is biased a particular value. For $b := \argmax_{x \in \set{0,1}} \Pr(\rX=x)$:
\begin{align*}
	\Pr\paren{\rX=x} = \begin{cases} 
	\frac{1}{2} + \frac{\bias(\rX)}{2} \qquad \textnormal{if $x=b$} \\ 
	\frac{1}{2} - \frac{\bias(\rX)}{2} \qquad \textnormal{if $x=1-b$}
	\end{cases}.
\end{align*}

We use the following standard equality that shows taking XOR of independent random variables significantly dampen the resulting bias\footnote{This is basically the intuition why one would expect any form of XOR Lemma to hold in the first place.}. 
\begin{proposition}\label{prop:xor-dampen-bias}
	For \emph{independent} random variables $\rX_1,\ldots,\rX_t$, 
	\[
		\bias(\rX_1 \oplus \cdots \oplus \rX_t) = \prod_{i=1}^{t} \bias(\rX_i). 
	\]
\end{proposition}
\begin{proof}
Let $\beta_i := \bias(\rX_i)$ and $b_i = \argmax_{x\in \{0,1\}}\Pr[\rX_i=x]$, and so $\Pr[\rX_i=b_i]=\frac 12 (1+\beta_i)$. We prove this proposition by induction. Consider the base case of $t=2$. We have, 
\begin{align*}
\Pr[\rX_1\oplus \rX_2=b_1\oplus b_2] &= \Pr[\rX_1=b_1 \wedge \rX_2=b_2] + \Pr[\rX_1=1-b_1 \wedge \rX_2=1-b_2] \\
&=\frac 12(1+\beta_1) \frac 12(1+\beta_2)+\frac 12(1-\beta_1) \frac 12(1-\beta_2) \\
&=\frac 14\paren{(1+\beta_1)(1+\beta_2)+(1-\beta_1) (1-\beta_2)} \\
&=\frac 14\paren{1+\beta_1 + \beta_2 + \beta_1\beta_2 + 1-\beta_1 -\beta_2 + \beta_1\beta_2} \\
&= \frac12\paren{1+\beta_1\beta_2} = \frac12 \paren{1+\bias(\rX_1)\bias(\rX_2)}. 
\end{align*}
Hence $\rX_1\oplus \rX_2$ has bias $\bias(\rX_1)\bias(\rX_2)$. For the induction step $k$, we can set $\rY = \rX_1 \oplus \ldots \oplus \rX_{k-1}$ and by the induction hypothesis, we have  $\bias (\rY)= \Pi_{i=1}^{k-1}\bias(\rX_i)$. 
We can then apply the above argument again for random bits $\rY$ and $\rX_k$, and conclude the proof. 
\end{proof}

%% file: appendixB-xor-lemma-optimality.tex

\section{Strong vs Weak XOR Lemmas and Optimality of~\Cref{thm:xor-lemma}}\label{sec:xor-examples}

The literature on XOR Lemmas distinguishes between the notion of ``weak'' XOR Lemmas vs ``Strong'' XOR Lemmas. Roughly speaking, in a weak XOR Lemma, the goal is to show that the advantage over random guessing for $\fxor$ 
is exponentially smaller than that of $f$, given \emph{the same resources} (or even less resources). In a strong XOR Lemma however, the goal is to show that this exponential drop
in the advantage happens even if we are given almost \emph{$\ell$ times more resources} than what is needed for solving $f$ itself. The reason to expect strong XOR Lemmas to hold, and they do hold in certain setting such as query
complexity~\cite{BrodyKLS20}, is that we expect the $\ell$-copy algorithm to have to ``spread'' its resources across each copies of $f$ and thus get to spend $\ell$ times less resources per each individual copy. Which of these two categories 
our~\Cref{thm:xor-lemma} belongs to?

Syntactically speaking, our Streaming XOR Lemma is clearly a weak XOR lemma as it uses the same exact resource of $p$-pass and $s$-space for both $f$ and $\fxor$. On the other hand, it is easy to see that 
one of the resources of interest here, \emph{space}, is inherently different from the other resources considered in XOR lemmas such as circuit size, communication cost, or query complexity: unlike all these resources, \emph{space is reusable}. As a result, it is pretty simple to show some 
examples where \Cref{thm:xor-lemma} is in fact optimal. For instance: 
\begin{itemize}[label=$-$, leftmargin=15pt]
	\item \emph{Example 1.} Let $f: \set{0,1}^{2m} \rightarrow \set{0,1}$ be the inner product of the first half and second half of $x \in \set{0,1}^{2m}$, i.e., $f(x) = \sum_{i=1}^{m} x_i \cdot x_{m+i} \mod 2$. 
	
	It follows from standard communication complexity lower bounds of the inner product problem (see, e.g.~\cite{KushilevitzN97}) that any single-pass $m/3$-space algorithm for $f$ over the uniform distribution $\mu$ on $\set{0,1}^{2m}$  
	has probability of success $\leq 1/2+2^{-\Theta(m)}$. 
	
	On the other hand, for the problem $\fxor(x^1,\ldots,x^\ell)$, there is a trivial single-pass $(m+O(1))$-space algorithm that succeeds with probability one: Simply compute each $f(x^i)$ using $m$ space and then maintain a running XOR of 
	these values. 
	
	This example shows that in any streaming XOR lemma, we cannot increase the space of the algorithm for $\fxor$ by more than a factor of $3$. As a result for any $\ell > 3$, there is no hope of getting any strong form of XOR lemma, at least
	for single pass algorithms. 
	
	\item \emph{Example 2.} Let $f$ be a lopsided multi-party pointer chasing problem as follows: 
	\[
		f: [m^2]^m \times \underbrace{[m^2]^{m^2} \times \cdots [m^2]^{m^2}}_r \rightarrow \set{0,1},
	\]
	and for a given $h: [m^2]^m \rightarrow [m^2]^{m^2}$ and $g^1,\ldots,g^r:  [m^2]^{m^2} \rightarrow [m^2]^{m^2}$, 
	compute the parity of 
	\[
	h(g^1(g^2(\cdots(g^r(1))))).
	\] 
	
	It again follows from standard communication complexity lower bounds for pointer chasing (e.g.~\cite{Yehudayoff16}; see also~\cite{GuhaM08} for extension to multi-party version and streaming) that over the uniform distribution
	over $h,g^1,\ldots,g^r$ (with this order of arrival in the stream), any $(r-1)$-pass streaming algorithm requires $\Omega_r(m)$ space while any $(r-2)$-pass algorithm requires $\Omega_r(m^2)$ space and both bounds are tight (because size of $h$ is small 
	enough to be ``shortcut'' by an $\Omega(m)$ space algorithm). 
	
	Again, let $s$ be such that any $(r-1)$-pass algorithm for $f$ with space $s$ only succeeds with probability, say, $\leq 2/3$. It suffices to have $s = \alpha(r) \cdot m$ for a sufficiently small $\alpha(r)$ as a function of $r$. 
	Now, for the problem $\fxor$, an algorithm that has $\ell \cdot m$ space for $\ell > \alpha(r)^{-1}$ can solve the problem even in $(r-2)$ passes breaking a strong XOR Lemma.   
	
	This example shows that in any streaming XOR lemma, we cannot increase the space of the algorithm for $\fxor$ by some large factor (which is at least proportional to the number of passes). Note that we chose $f$ in this example so that 
	the algorithm for $\fxor$ has to still use $(r-2)$ passes; we can also play with number of $h$ and $g$ functions in definition of $f$ to change this quantity to any other number of passes. This suggests that for \emph{any} number of passes, 
	one cannot hope for a strong streaming XOR Lemma. 
\end{itemize}
\noindent
The above examples point out the optimality of~\Cref{thm:xor-lemma} in various cases. More conceptually, the intuition behind XOR Lemmas is that the naive algorithm that attempts to solve each subproblem individually (and thus independently), 
is more or less optimal. In most settings, this corresponds to a strong XOR Lemma. However, as above examples suggest, in the streaming, even the naive algorithm that solve each problem individually does not need to spend 
more resources per each subproblem as space is reusable. As such, it seems that semantically speaking, our~\Cref{thm:xor-lemma} is the strongest type of XOR Lemma in our setting\footnote{Although we should note that one can consider
more general type of~\Cref{thm:xor-lemma} which, for instance, is defined by \emph{interleaving} the streams of subproblems instead of \emph{concatenating} them; such an approach is not particularly meaningful for our application in this paper but is an interesting
question on its own.}.

\begin{remark}
	It is worth pointing out  the work of Shaltiel~\cite{Shaltiel03} who gave several nice examples in showing limitations of strong XOR Lemmas in general (e.g., circuit complexity, communication complexity, or query complexity). However, those examples and ours seem entirely different. In particular, Shaltiel's examples are based on working with distributions that are ``often'' easy and just become hard with small probability, which allows the algorithm for $\fxor$ to spend more of its resources on a particular subproblem among the $\ell$ ones. Our examples do not have such a property and in fact they hold even for the notion of ``fair'' algorithms studied in~\cite{Shaltiel03} that are algorithms that use the same amount of resources for each of the $\ell$ subproblems in $\fxor$ (and break the examples of~\cite{Shaltiel03}). We believe these differences are  rooted in the different nature of space as a resource. 
\end{remark}

Finally, let us conclude this section by a (rather unrelated) comment about~\Cref{thm:xor-lemma}.

\begin{remark}
\Cref{thm:xor-lemma} is stated as a distributional lower bound, which we found more natural for our purpose in this paper. However, it can also be stated as a worst-case lower bound by applying Yao's minimax principle twice: once to get a hard distribution for $f$, then applying~\Cref{thm:xor-lemma}, and another one to get a worst-case lower bound for $\fxor$.  
\end{remark}

%% file: appendixC-alg-ngc.tex

\section{An  Optimal Algorithm for Noisy Gap Cycle Counting}\label{sec:alg-cycle} 

We give a short and simple proof that demonstrate the asymptotic optimality of our lower bounds for noisy gap cycle counting in~\Cref{thm:ngc}. 

\begin{proposition}\label{prop:alg-cycle}
	For any $p,k \in \IN^+$ such that $k$ divides $2p$, there is a $p$-pass streaming algorithm for Noisy Gap Cycle Counting $\NGC_{n,k}$ that uses $O_{p,k}((n/k)^{1-2p/k} \cdot \log{n})$ bits of space and outputs the correct answer with probability at least $2/3$. 
\end{proposition}

For simplicity of exposition, we presented~\Cref{prop:alg-cycle} for the case when $2p \mid k$. One can trivially extends the bounds to an $(n/k)^{1-O(p/k)}$ algorithm using this result (by picking the $\tilde p< p$ such that $2 \tilde p \mid k$ and running the 
algorithm for $\tilde p$ passes). A bit more careful calculation can also give a sharper bound of $O_{p,k}((n/k)^{1-\frac{1}{\ceil{k/2p}}}\log{n})$ for every  $p \leq k$; we omit the details. 

\begin{proof}[Proof of~\Cref{prop:alg-cycle}]
	
	Consider the following algorithm: 
	
	\medskip
	
\begin{algorithm}[H]
\caption{An  optimal algorithm for the noisy gap cycle counting problem}\label{alg}
\label{algo:psi}

\begin{algorithmic}[1]

\medskip

\State Sample each vertex independently  with probability $q := 36 \cdot (n/k)^{-2p/k}$ to get a set $U$ -- if $\card{U} > 100  q\cdot n$ terminate and return FAIL. 

\smallskip

\State For every vertex $u_i \in U$, find the depth-$p$ BFS tree of $u_i$ in $G$ in $p$ passes.  

\smallskip

\State If there is a $k$-cycle among the visited BFS trees, output $G$ is a ``$k$-cycle case''; otherwise output it is a ``$(2k)$-cycle case''. 

\smallskip

\end{algorithmic}

\end{algorithm}
	
	\medskip
	
The pass complexity of~\Cref{alg} is clearly $p$. For space complexity, considering the graph consists of only cycles and paths, every depth-$p$ BFS tree can have at most $O(p)$ vertices and so the total space needed by the algorithm is 
$O(p \card{U} \log{n}) = O_{p,k}((n/k)^{1-2p/k} \log{n})$ bits by the condition on terminating the algorithm. 

We now prove the correctness. In the following, for simplicity, we remove the condition in the algorithm that terminates whenever $U$ is large; as by Markov bound this event happens with probability $1/100$ anyway, 
this can only change the final probability of success calculated this way by a value of $-1/100$. Moreover, considering the algorithm has a one-sided error, i.e., never outputs ``$k$-cycle'' on a graph that does not have a $k$-cycle, 
we only need to prove that when the graph has a $k$-cycle, the algorithm finds it with probability at least $2/3$, which we do in the following. 

Recall that $n= 6t \cdot k$ in $\NGC_{n,k}$ and we are interested in the case $G$ has $2t$ vertex-disjoint $k$-cycles. Let $C_1,\ldots,C_{2t}$ denote these cycles. Fix an arbitrary cycle $C= (v_0,v_1,\ldots,v_{k-1})$ from this list. Pick 
the following $\beta := (k/2p)$ vertices from $C$: $v_0,v_{2 \cdot p},v_{4\cdot p},\ldots,v_{2\beta \cdot p}(=v_{k})$. The distance between any consecutive pairs of vertices in this list inside $G$ is at most $2p$ and thus their depth-$p$ BFS trees
intersect with each other. As such, if we sample all these vertices in the algorithm, the set of BFS trees returned by the algorithm contains a $k$-cycle. As each vertex is sampled independently with probability $q$, the probability that this 
event happens is $q^{\beta}$. 

On the other hand, since all cycles $C_1,\ldots,C_{\beta}$ are vertex-disjoint, the above event happens for each of them independently. As such, the probability that we do not see a $k$-cycle is at most
\[
	(1-q^{\beta})^{2t} \leq \exp\paren{-q^{\beta} \cdot 2t} = \exp\paren{-(36 \cdot (k/n)^{2p/k})^{k/2p} \cdot (n/3k)} \leq \exp\paren{-12} < 1/100. 
\]
As such~\Cref{alg} outputs the correct answer with probability at least $1-2/100 > 2/3$. 
\end{proof}

%% file: appendixD-comparison.tex

\section{A Technical Comparison with the Lower Bound of~\cite{AssadiKSY20}} \label{sec:technical-comparison}
\cite{AssadiKSY20} previously proved a multi-pass lower for the (no-noise) gap cycle counting problem. Their work starts by proving a lower bound for a generalization of pointer chasing wherein the goal is to chase multiple pointers simultaneously; this lower bound is also 
in the low-probability regime, i.e., $1/\poly(n)$ advantage over random guessing, but \emph{only works for single-pass algorithms}. This is then plugged into an intricate \emph{round/pass-elimination} argument
that, informally speaking, shows that solving the problem in $p$ passes for depth-$k$ pointers is 
as hard as solving the problem in $p-1$ passes but for depth $k/2$ instead. Applying this step inductively gives an $\Omega(\log{k})$ pass lower bound for $n^{o(1)}$-space algorithms. 
 A quick technical summary of~\cite{AssadiKSY20} is then the following: prove a ``strong'' single-pass lower bound first and then increase the number of passes without ``weakening'' the lower bound too much. 
This somewhat the exact opposite of what we do in this paper: we  prove a ``weak'' multi-pass lower bound first and then plug it into our streaming XOR Lemma to
amplify its hardness and obtain a ``strong'' lower bound via a reduction that slightly reduces the number of  passes. This natural hardness amplification approach leads to 
exponentially stronger lower bound of $\Omega(k)$ passes for $n^{o(1)}$-space algorithms with a considerably simpler proof.

It is worth pointing out two other differences (both technical and conceptual) as well. Firstly, owing to the introduction of noise, we obtain a clear reduction from pointer chasing instead of the implicit connection in~\cite{AssadiKSY20} that is scattered throughout the proof. 
Secondly, we directly consider streaming algorithms as opposed to  working with the  two-party communication model in~\cite{AssadiKSY20}
and reducing it to streaming only at the end; this is critical in our work as our XOR Lemma is for streaming algorithms and not (two-party) communication.